\pdfoutput=1
\documentclass[letterpaper, 12pt]{article}

\usepackage{amsmath}
\usepackage{amsthm}
\usepackage{amssymb}
\usepackage{graphicx} % Required for inserting images
\usepackage[top=1in,bottom=1in,left=1in,right=1in,marginparwidth=2cm,marginparsep=0.3cm]{geometry}
\usepackage[onehalfspacing]{setspace}
\usepackage{xcolor}
\usepackage{todonotes}
\usepackage[colorlinks=true, linkcolor=blue, citecolor=blue, urlcolor=blue, pdfborder={0 0 0}]{hyperref}
\usepackage[ruled, vlined]{algorithm2e}
\usepackage{booktabs}
\usepackage{comment}
\usepackage{placeins}
\usepackage{subcaption}
\usepackage{amsfonts} 
\usepackage{multirow}
\usepackage{siunitx}
    \usepackage{comment}
\allowdisplaybreaks

\usepackage{threeparttable}
\usepackage{array}

\usepackage{natbib} 
\definecolor{maroon}{RGB}{128,0,0}
\colorlet{revcolor}{maroon}
\newcommand{\rev}[1]{#1} % revision mark-up switched off; see main_highlight_changes.tex
\hypersetup{pdftitle={A Stochastic Nested Fixed Point Algorithm for Large-Scale BLP Estimation},
  pdfauthor={Zhentong Lu, Myung Hwan Seo, Youngki Shin, Qichen Zhang}}
\newtheorem{assum}{Assumption}
\newtheorem{lemma}{Lemma}
\newtheorem{theorem}{Theorem}
\newtheorem{corollary}{Corollary}
\newtheorem{remark}{Remark}
\DeclareMathOperator*{\argmin}{arg\,min}
\newcommand{\pto}{\overset{p}{\to}}

\title{\rev{A Stochastic Nested Fixed Point Algorithm for Large-Scale BLP Estimation}}
\author{
    Zhentong Lu\thanks{Financial Stability Department, Bank of Canada, Email:
  \texttt{zlu@bankofcanada.ca}},   
    Myung Hwan Seo\thanks{Department of Economics,  The Hong Kong University of Science and Technology, Email: \texttt{myunghseo@ust.hk}},
    Youngki Shin\thanks{Department of Economics, McMaster University, Email:
  \texttt{shiny11@mcmaster.ca}}, 
  and Qichen Zhang\thanks{Department of Economics, McMaster University, Email:
  \texttt{zhanq73@mcmaster.ca}}
  }
\date{\today}

\begin{document}

\maketitle

\begin{abstract}
We develop a stochastic nested fixed point (SNFP) estimator for random coefficients logit demand models that updates model parameters using stochastic gradients and performs demand inversion one market at a time. Relative to the conventional nested fixed point (NFP) estimator, SNFP substantially reduces memory requirements and computational cost, making estimation feasible in very large datasets. We establish the large-$T$ (number of markets) asymptotic properties of the estimator under regularity conditions. \rev{We also characterize the effect of sharing one block of simulation draws across markets and show how to correct for it.} Monte Carlo simulations show that the SNFP estimator achieves statistical accuracy comparable to the NFP estimator\rev{, and in our benchmark a single online pass estimates a model with 100 million markets in about 5.5 hours}. An empirical application using scanner data further demonstrates the practical advantages of SNFP for large-scale demand estimation.
\end{abstract}



\newpage
\section{Introduction}\label{section:intro}

The random coefficients logit model of \citet{berry1995automobile}, henceforth BLP, is the workhorse framework for demand analysis of differentiated product markets. Using market-level data on prices, shares, and product characteristics, it recovers flexible substitution patterns while accommodating unobserved consumer heterogeneity and endogenous prices. These demand estimates underpin a wide range of applications, including the measurement of market power, the evaluation of mergers and trade policies, and the welfare analysis of new products \citep{berry1995automobile,goldberg1995,nevo2001measuring,petrin2002quantifying,berry1999,nevo2000mergers,millerweinberg2017}. For an overview, see \citet{berryhaile2021}.

Estimating the model is computationally demanding, and the cost grows with the number of markets. The standard NFP algorithm and its reformulation as a mathematical program with
equilibrium constraints (MPEC) \citep{dube2012improving,su2012constrained} carry the whole sample into every parameter update. Specifically, NFP \rev{inverts} the share equations of all \rev{markets} at each iteration of the generalized method of moments (GMM) optimization. Although MPEC replaces the inversion by constraints, it still optimizes over the mean utilities of all $T$ markets at once. \rev{For an NFP implementation that retains the simulated individual-level utilities and market shares for every market, storage scales as $O(T \cdot J \cdot R)$ with $T$ markets, $J$ products, and $R$ simulation draws}, and every iteration passes over all $T$ markets. 

Modern datasets routinely push $T$ into the millions. The Nielsen Retail Scanner panel records weekly sales for roughly $35{,}000$--$50{,}000$ stores. Over a decade, about $520$ weeks, a single product category observed at the store-week level yields millions of markets. Daily product--market records on an e-commerce platform and fine geographic cells in credit-card spending data are of similar scale. \rev{At these dimensions, every NFP GMM objective evaluation inverts the share equations for millions of markets, while memory requirements continue to grow linearly with the number of markets.}

To circumvent the computational burden, researchers may restrict attention to a small, computable sample, or may shrink
a large one before estimation by aggregating markets into coarser units or by keeping a subset of markets. Either
choice has a cost, because the BLP estimator may perform poorly when the number of markets is small and instruments
have low power. The following simulation experiment shows this clearly. We simulate a simple random-coefficient logit model with an
endogenous price, in which the parameter of interest is the standard deviation $\sigma_p$ of the random coefficient on
price. We estimate it by NFP with the standard BLP moment condition, $R = 500$
simulation draws, and $1{,}000$ replications at each $T \in \{2{,}000,\; 8{,}000,\; 32{,}000,\; 128{,}000\}$.
Appendix~\ref{appendix:small_sample_details} gives the full simulation design. At $T = 2{,}000$ the GMM objective is
nearly flat in $\sigma_p$, and the optimizer stops at a boundary solution ($\hat\sigma_p \leq 0.02$) in $33\%$ of
replications, against less than $1\%$ at $T = 128{,}000$. Figure~\ref{fig:small_sample_blp} plots the $t$-statistic
density for $\hat\sigma_p$. At $T = 2{,}000$ the density is dispersed and visibly non-normal, whereas by
$T = 128{,}000$ it tracks $N(0,1)$ closely. Researchers are thus caught between a full-sample estimator \rev{that is costly to}
compute at the $T$ that many-markets asymptotics require \citep{freyberger2015asymptotic} and a reduced sample that
distorts the parameter of interest.

\begin{figure}[t]
    \centering
    \caption{Small-Sample Distribution of the BLP Estimator}
    \label{fig:small_sample_blp}
    \vspace{2pt}
    \includegraphics[width=0.62\textwidth]{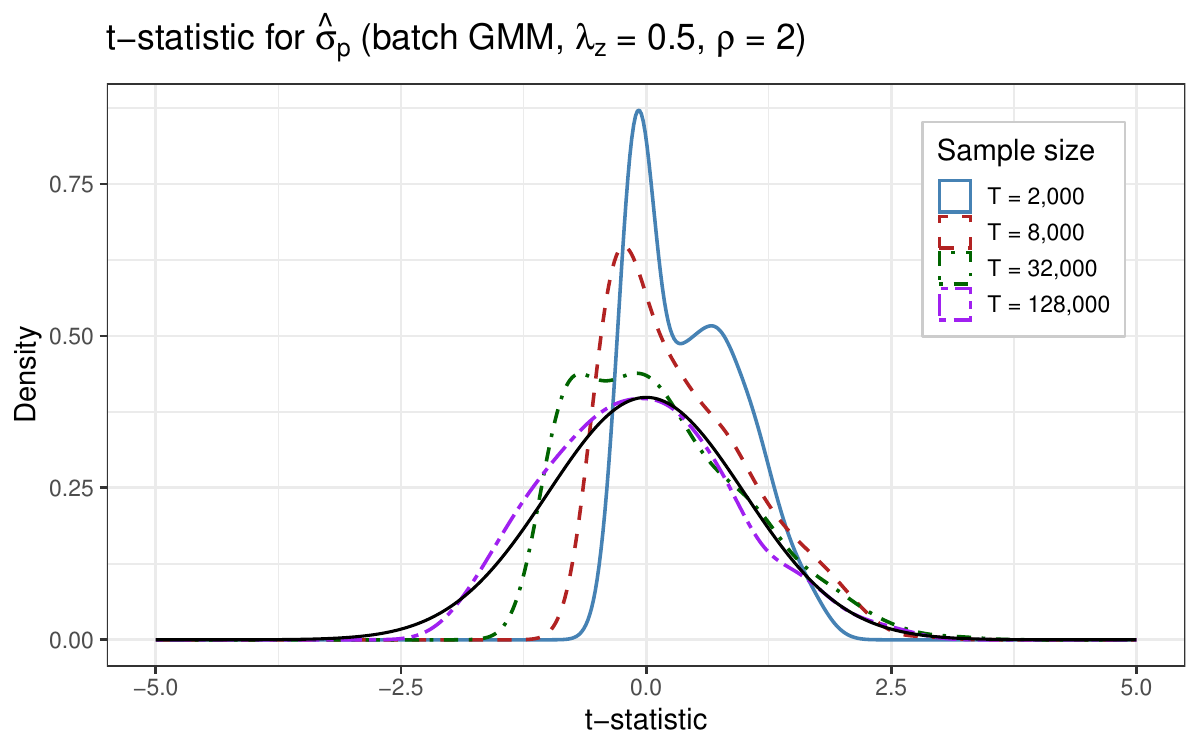}

    \vspace{2pt}
    \begin{minipage}{0.62\textwidth}
    \footnotesize \textit{Notes:} Density of the $t$-statistic for $\hat\sigma_p$. Solid black line: $N(0,1)$. Boundary solutions are excluded.
    \end{minipage}
\end{figure}

This paper aims to relax this trade-off. We propose the SNFP algorithm, to our knowledge, the first online GMM estimator for
the BLP model. In the spirit of online learning, SNFP visits the markets one at a time, needs memory for only
one market, and has the same large-$T$ distribution as full-sample GMM.

Rather than recomputing sample moments using all $T$ markets at every parameter update, SNFP processes one market at a time. For each market, it runs the BLP contraction mapping to recover the structural error $\xi_t(\theta_{t-1})$, evaluates the market's GMM moment contribution $g_t(\theta_{t-1}) =
Z_t'\,\xi_t(\theta_{t-1})$, with $Z_t$ the market's instruments, and takes one stochastic gradient descent (SGD) step \citep{robbins1951stochastic}. \rev{Because SNFP preserves the standard BLP market-level demand inversion, advances in market-share computation and mean-utility recovery can be incorporated directly.} 

The update is scaled by a matrix that is computed once from a pilot sample of $T_0$ markets and then held fixed, so no full-sample Jacobian or weighting matrix is ever recomputed. The final estimator is the Polyak--Ruppert average of the iterates,
$\bar\theta_T = (T-T_0)^{-1}\sum_{t=T_0+1}^T \theta_t$ \citep{polyak1992acceleration,ruppert1988efficient}. Because only one market is active at any step, memory per step is $O(J \cdot R)$, independent of $T$. \rev{The theory covers a
single pass over the markets, although a small number of shuffled passes can further improve finite-sample accuracy.}

The asymptotic theory \rev{for the SNFP estimator} is non-standard, because the BLP moment depends on $\theta$ through an implicitly defined,
simulation-based fixed point, a case existing stochastic GMM theory does not cover. We show that the Polyak--Ruppert
average is consistent and $\sqrt{T}$-asymptotically normal with the same sandwich covariance as the
full-sample GMM estimator \rev{with the same weighting matrix}, provided the number of simulation draws grows faster than $\sqrt{T}$. A plug-in
sandwich estimator delivers standard errors. 

\rev{We also analyze the implementation in which one block of $R$ i.i.d.\
simulation draws is shared by every market.} Sharing makes the simulation error common across markets, so it no longer averages
out. \rev{The estimator remains consistent as $R \to \infty$. When $T/R \to c < \infty$, its asymptotic variance acquires
an additional term $c\,\Lambda_\nu$ from the shared draws, which plug-in standard errors omit, and when $R = o(T)$ the $\sqrt T$ rate is
lost. For full-sample estimators, \citet[Remark~3]{freyberger2015asymptotic} notes that shared draws require $T/R$ to
be bounded for $\sqrt T$-consistency, and \citet{hong2021blp} establish $\sqrt{\min(R,T)}$ asymptotics, with a consistent
variance estimator, for the GMM estimator under shared draws. We establish the corresponding result for the online
estimator and show that the additional variance term can be estimated within the same pass. We further show that standard
inference is restored if the shared block is taken instead from a low-discrepancy, or quasi-Monte Carlo, point set,
such as a scrambled Sobol' or Halton sequence. Such a block integrates the choice probabilities with an error
$\epsilon_R$ that shrinks faster than the $R^{-1/2}$ of random draws, and what the theory requires is
$\sqrt T\,\epsilon_R \to 0$. For digital nets and Halton sequences, $R = T^{b}$ with any $b > 1/2$ suffices.} 

The computational gains are large, and statistical accuracy is comparable. In Monte Carlo experiments, \rev{a single pass of} SNFP is \rev{between $45$ and $146$} times faster
than NFP and uses far less memory, and estimates a model with $100$ million markets in about $5.5$ hours on a
single core. By contrast, \rev{in our benchmark} NFP does not finish within a ten-hour budget already at about two million markets. \rev{From $8{,}000$ markets onward, the single-pass estimator's} bias \rev{and} root mean squared error are close to NFP's, \rev{its coverage approaches the nominal level as $T$ grows,} and its accuracy improves at the rate
the theory predicts as $T$ grows. \rev{In an application to $817{,}031$ store--week markets of IRI yogurt data, SNFP
estimates the model on the full sample, whereas estimates from repeated subsamples remain dispersed and systematically differ from the
full-sample estimate, particularly for the parameters governing heterogeneity in price sensitivity.}

The paper makes two contributions. First, we propose the SNFP algorithm, which removes the full-sample iteration from
BLP estimation. Existing accelerations, from MPEC \citep{dube2012improving,su2012constrained} to approximate BLP
\citep{leeseo2015} and the Anderson-accelerated contraction in \texttt{PyBLP} \citep{conlon2020best}, reduce the cost of
each full-sample iteration. SNFP embeds the market-level inversion inside a stochastic gradient update, so that
working
memory is $O(JR)$ rather than $O(TJR)$, and because the within-market inversion is unchanged, these accelerations can be
used inside it. In benchmarks, SNFP is one to two orders of magnitude faster than NFP and uses up to an order of
magnitude less memory, with \rev{bias and RMSE} close to NFP's \rev{at all but the smallest sample size}. Second, we provide the asymptotic theory
for a stochastic GMM estimator whose moment is an implicitly defined, simulation-based object. The results of
\citet{chen2025sgmm} and \citet{chen2026slim} assume moments available in closed form after one forward pass. We
extend
them to a moment defined by a simulated fixed point and \rev{show that the condition $\sqrt T/R \to 0$, under which
simulation error is negligible for full-sample BLP with fresh per-market draws \citep{freyberger2015asymptotic}, also
suffices for the online recursion.}

\subsection*{Related Literature}\label{subsection:related_literature}

Our paper is related to two strands of literature. The first seeks to reduce the computational burden of BLP estimation. \citet{dube2012improving} and \citet{su2012constrained}
recast estimation as mathematical programming with equilibrium constraints (MPEC), optimizing jointly over parameters
and mean utilities and so avoiding repeated inner-loop solves at each outer step. \citet{leeseo2015} propose approximate
BLP (ABLP), replacing the numerical contraction with a closed-form inversion based on a linearized market-share system.
\citet{conlon2020best} implement Anderson acceleration in the contraction mapping and provide the widely used
\texttt{PyBLP} package. \rev{\citet{fukasawa2024fast} propose fast inner-loop mappings and acceleration for static and
dynamic BLP estimation.} Part of the motivation for this literature is numerical:
\citet{knittel2014estimation} document that BLP estimates can be sensitive to the choice of optimizer, the
starting values, and the inner-loop tolerance, and that many runs terminate at points that fail the first-order
conditions. \citet{fukasawa2026sequential} develop a unified theory of sequential algorithms for structural
estimation with equilibrium constraints, show that ABLP satisfies the zero-Jacobian property, and propose a
Jacobian-free sequential linearly constrained (SLC) algorithm as an alternative to NFP and MPEC.
\citet{aguirregabiria2026nested} propose a nested pseudo-GMM (NP-GMM) estimator that, in the spirit of nested
pseudo-likelihood methods for dynamic discrete choice, swaps the order of the GMM optimization and the share inversion
so that mean utilities are recovered in closed form. \citet{lushitao2023} develop
a semi-nonparametric two-step estimator that re-expresses the demand system as a partial linear model and avoids
numerical share inversion while allowing flexible specifications of the random-coefficient distribution.
\rev{\citet{salanie2022frac} propose a fast approximate alternative that replaces the inversion by a small-heterogeneity
expansion estimated by linear instrumental variables.}

\rev{These methods accelerate or reformulate full-sample estimation, and those that iterate still pass over all $T$
markets at every parameter update.} SNFP
targets that loop instead. Markets enter the outer update sequentially, working memory is $O(JR)$, and no step
requires
a full pass. Because the within-market inversion problem is unchanged, the methods above are complementary. Anderson
acceleration, analytic inversions, and the pseudo-GMM shortcuts of MPEC, ABLP, SLC, or NP-GMM, as well as inversion-free
semi-nonparametric estimation \citep{lushitao2023}, can in principle run inside the stochastic outer loop. The moment
inequalities of \citet{gandhilushi2022} for demand estimation with zero market shares, which are common under fine
disaggregation, can likewise be combined with SNFP's online processing of markets. Motivated by weak-instrument
concerns in BLP applications, \citet{lushimizu2026} take a different route, estimating demand under a sparsity
restriction on the market--product shocks without relying on instruments.

The second strand is stochastic approximation. The Robbins--Monro framework \citep{robbins1951stochastic} underlies
a large literature on online estimation, and Polyak--Ruppert averaging
\citep{polyak1992acceleration,ruppert1988efficient} restores $\sqrt{n}$ convergence rates in settings where raw SGD
iterates converge more slowly. \citet{chen2025sgmm} bring these ideas to GMM, showing that SGD on moment conditions
yields $\sqrt{T}$-consistent, asymptotically normal estimators with the classical sandwich covariance, and
\citet{chen2026slim} extend the analysis to overidentified models. Closer to discrete choice, \citet{lee2026sauss}
estimate multinomial choice models by averaged stochastic approximation with unbiased simulated scores, which avoids
the simulation bias of simulated maximum likelihood. These methods assume that the moment or score and its gradient
can be evaluated, exactly or by an unbiased simulation, after one forward pass. BLP violates that
assumption, because the structural error $\xi_t(\theta)$ solves a market-level fixed-point equation, so even a
single market's moment
requires an inner iterative solve. SNFP
embeds that inversion in the stochastic gradient update, coupling inner inversion and outer optimization in one
sequential pass without ever forming the full-sample mean-utility matrix, and we specialize and extend the theory of
\citet{chen2026slim,chen2025sgmm} to this implicitly defined, simulation-based moment function.

% ----  Online inference via random scaling ----
% Random-scaling inference is not included in this version.
\begin{comment}
Because SNFP processes data sequentially, the partial-sum process of the Polyak-averaged iterates
is available as a byproduct of estimation at no additional cost. This makes the random-scaling
inference method of \citet{lee2022fast}, originating with the fixed-$b$ approach of
\citet{kiefer2000simple}, the natural device for inference. The self-normalized test statistic,
formed by dividing the centered estimator by a functional of the partial-sum process, converges to
a pivotal distribution under the null without requiring estimation of the long-run variance.
Confidence intervals therefore require neither a second pass over the data nor an explicit estimate
of the sandwich covariance matrix.
\end{comment}

% ---- Roadmap ----
\medskip
The remainder of the paper is organized as follows. Section~\ref{section:model} sets up the random coefficients
logit model and the full-sample NFP--GMM estimator. Section~\ref{section:algorithm} states the SNFP algorithm
(Algorithm~\ref{alg:snfp-skeleton}), including the pilot step that produces the starting value and the fixed scaling
matrix. Section~\ref{section:asymptotic_theory} develops the asymptotic theory. Section~\ref{section: Monte Carlo}
reports Monte Carlo simulations. Section~\ref{section:empirical} presents the empirical application.
Section~\ref{section:conclusion} concludes. Proofs and supplementary simulation results are collected in the
appendices.

\section{Discrete Choice Model and the BLP Estimator}\label{section:model}

\subsection{Setup}

Consider a standard BLP demand model for differentiated products observed over markets $t=1,\dots,T$. In each market $t$, there are $J_t$ inside goods indexed by $j=1,\dots,J_t$, as well as an outside option indexed by $0$.

Consumer $i$ in market $t$ derives (indirect) utility from choosing product $j$
\begin{align}
u_{ijt} = \delta_{jt} + \mu_{ijt} + \varepsilon_{ijt},
\label{eq:utility}
\end{align}
where $\delta_{jt}$ is the mean utility of product $j$ in market $t$, $\mu_{ijt}$ captures deviations from mean utility arising from preference heterogeneity, and $\varepsilon_{ijt}$ is i.i.d.\ type I extreme value. The outside option utility is
normalized to \(u_{i0t} = \varepsilon_{i0t}\).

We introduce preference heterogeneity through a random-coefficients specification. Let $x_{jt}$ denote the vector
of observed product characteristics (typically including price), and let $\xi_{jt}$ denote an unobserved product
characteristic observed by firms but not by the econometrician. The systematic part of
utility in \eqref{eq:utility} is linear in the characteristics with consumer-specific coefficients,
\[
\delta_{jt} + \mu_{ijt} = x_{jt}'\beta_i + \xi_{jt}, \qquad \beta_i = \theta_1 + \Sigma v_i,
\]
where{} $\theta_1$ is the vector of mean taste parameters, $v_i \sim N(0,I)$ is a vector of independent
standard normal random variables, and $\Sigma$ is a matrix of parameters governing the dispersion of tastes across
consumers. Let $\theta_2 := \mathrm{vec}(\Sigma)$ and \(\theta := (\theta_1',\theta_2')'\).

Substituting $\beta_i$ separates the systematic part of utility into a component common
to all consumers and an individual-specific deviation,
\[
\delta_{jt} = x_{jt}'\theta_1 + \xi_{jt}, \qquad \mu_{ijt} = x_{jt}'\Sigma v_i,
\]
so that the mean utility $\delta_{jt}$ absorbs the unobserved characteristic $\xi_{jt}$ and $\mu_{ijt}$ carries the
preference heterogeneity.

Given $(\delta_t,\Sigma)$, the choice probability (market share) of product $j$ in market $t$
is
\begin{align}
\tilde{s}_{jt}(\delta_t,\theta_2)
&=
\int
\frac{\exp\!\left(\delta_{jt} + \mu_{ijt}\right)}
{1+\sum_{k=1}^{J_t}\exp\!\left(\delta_{kt} + \mu_{ikt}\right)}
\, dF_v(v_i),
\label{eq:share_integral_delta}
\end{align}
where $\delta_t := (\delta_{1t},\ldots,\delta_{J_t t})'$ and $F_v$ denotes the standard multivariate normal
distribution of $v_i$.

\subsection{The BLP Estimator}

Let $\xi_t := (\xi_{1t},\ldots,\xi_{J_t t})'$ denote the vector of 
unobserved product characteristics in market $t$ and assume the existence of instruments 
$Z_t \in \mathbb{R}^{J_t \times d_z}$ such that
\begin{equation}
E[ Z_t' \xi_t ] = 0. 
\label{eq:true_moment_condition}
\end{equation}
It follows that
\begin{equation}
E\!\left[ Z_t' \Big( \tilde{\delta}_t(S_t,\theta_2) - X_t \theta_1 \Big) \right] = 0,
\label{eq:moment_theta}
\end{equation}
where $X_t = (x_{1t},\ldots,x_{J_t t})'\in \mathbb{R}^{J_t \times d_x}$, and $\tilde{\delta}_t(S_t,\theta_2)$ is the unique vector of mean utilities that rationalizes
the observed market share vector \(S_t\) under heterogeneity parameters $\Sigma$, i.e.,
$\tilde{\delta}_t(S_t,\theta_2)$ solves the system 
\begin{equation}
    S_t = \tilde{S}_t(\delta_t,\theta_2).
\end{equation}
The existence and uniqueness of the solution $\tilde{\delta}_t(S_t,\theta_2)$ are established by \cite{berry1994estimating} and \cite{berry2013connected}.

Based on the moment conditions~\eqref{eq:moment_theta}, the standard BLP approach estimates the parameter of interest by a nonlinear GMM:
\begin{align}
    \hat\theta = \argmin_{\theta \in \Theta} \left( \dfrac{1}{T} \sum_t g_t(\theta) \right)' W_T \left( \dfrac{1}{T}\sum_t g_t(\theta) \right) \label{eq: gmm obj}
\end{align}
where $g_t(\theta) := Z_t'\bigl(\tilde{\delta}_t(S_t,\theta_2) - X_t \theta_1\bigr)$, and $W_T$ is a weighting matrix that converges to a positive definite matrix $W$. 

In practice, $\tilde{\delta}_t(S_t,\theta_2)$ does not have a closed-form expression but can be obtained numerically by the BLP contraction mapping
\begin{align}
    \delta_{t}^{h+1} = \delta_{t}^{h} + \log S_{t} - \log \tilde{S}_{t}(\delta_{t}^h,\theta_2), \label{eq:contraction mapping}
\end{align} 
where \(h\) indexes the iteration step, and the integral in $\tilde{S}_{t}(\delta_{t}^h;\theta_2)$ (defined by~\eqref{eq:share_integral_delta}) can be evaluated by numerical integrations or Monte Carlo simulations. Since $\tilde{\delta}_t(S_t,\theta_2)$ needs to be solved for any $\theta_2$, the GMM estimation procedure here is known as the nested fixed point (NFP) algorithm. Note that the NFP algorithm uses all data points for each iteration in the contraction mapping, so computation becomes substantially burdensome as $T$ increases, often preventing researchers from using the full set of available observations (say $T>100{,}000$).
 
\section{The SNFP Algorithm}\label{section:algorithm}

In this section, we introduce the SNFP algorithm. Building on the stochastic generalized method of moments (SGMM) approach in \citet{chen2025sgmm}, we use stochastic gradient descent (SGD) in the outer loop, while the inner loop solves the BLP contraction mapping for a \emph{single} market corresponding to the current update. The main feature of the SNFP algorithm is that \rev{each update uses the data of a \emph{single} market, and a single pass visits each market once}, substantially saving computation time and memory usage when the \rev{number of markets} $T$ is large. \rev{SNFP is to NFP what SGD is to gradient descent (GD). Each step solves the simulated share equations of a single market, exactly as NFP does for every market, and takes one stochastic-approximation step on that market's GMM moment, scaled by a matrix fixed after an initial pilot step.} 

We apply SGD to the GMM problem~\eqref{eq: gmm obj}, updating the parameter vector after each market $t$ is processed. Given this value of the updated parameter, the mean utility is calculated only for the next market $t+1$, and the procedure continues until all markets have been visited. Because each update depends only on a single market, the remaining markets need not be held in memory during the estimation. Relative to the SGMM framework of \citet{chen2025sgmm}, the key distinction is the inner loop computation. The SNFP computes the mean utility only for the current market, which is the only quantity required for the next stochastic gradient update. 

First, we compute an initial estimate using a pilot sample of size $T_0$. Specifically, applying the standard NFP algorithm to the pilot sample, we can estimate the following initial values:
\begin{align}
    \hat\theta_{T_0} & = \argmin_{\theta} \bar{g}_{T_0} (\theta)'\bar{g}_{T_0} (\theta)\\
    \hat\Phi_{T_0} & = \nabla_\theta \bar{g}_{T_0}(\hat\theta_{T_0}) \\
    \hat W_{T_0} &= \left( \frac{1}{T_0} \sum_{t=1}^{T_0} g_t(\hat\theta_{T_0}) g_t(\hat\theta_{T_0})' + \eta^\ast   I_{d_z}  \right)^{-1} \\
    \hat A_{T_0} &= \left( \hat\Phi_{T_0}^\top \hat W_{T_0} \hat\Phi_{T_0} \right)^{\dagger} \hat\Phi_{T_0}^\top \hat W_{T_0},
\end{align}
where $\bar{g}_{T_0}(\theta) = T_0^{-1}\sum_t g_t(\theta)$ and $\eta^\ast >0$ is a small constant. The matrix $\hat A_{T_0}$ is the pilot scaling matrix. It is computed once and then held fixed for all subsequent updates.

Next, we conduct the SGD procedure for $t=T_0+1,\ldots,T$. For each $t$, we use the BLP contraction mapping~\eqref{eq:contraction mapping} to compute the mean utility vector $\tilde{\delta}_t(S_t,\theta_{t-1,2})$, given $\theta_{t-1,2}$ and data $\{S_t, X_t\}$ from market \(t\), and then update the parameter by a projected stochastic gradient step using the \emph{fixed} pilot scaling matrix $\hat A_{T_0}$, followed by Polyak--Ruppert averaging:
\begin{align}
    \theta_t & = \Pi_\Theta\!\left\{ \theta_{t-1} - \gamma_t\, \hat A_{T_0}\, g_t(\theta_{t-1}) \right\}, \qquad \gamma_t = \gamma_0\, t^{-a}, \\
    \bar{\theta}_t & = \dfrac{t-1-T_0}{t-T_0}\,\bar{\theta}_{t-1} + \frac{1}{t-T_0}\, \theta_t,
\end{align}
where $\Pi_\Theta$ denotes the Euclidean projection onto the parameter space $\Theta$. The scaling matrix $\hat A_{T_0}$ is the fixed pilot quantity defined above. Unlike the standard NFP--GMM iteration, no full-sample Jacobian or weighting matrix is recomputed at any step.
Note that at iteration \(t\), we apply the BLP contraction mapping to compute the vector of mean utilities for market \(t\) only, rather than for all markets. The full algorithm is shown in Algorithm~\ref{alg:snfp-skeleton}.

% 
% 10:18 pm 
% 
% exact gradient formula. how to prove the consistency. 
% need to look at the consistency proof: 
%
\SetKwInOut{Input}{Input}
\SetKwInOut{Output}{Output}
\SetKwFunction{OneDeltaStep}{OneDeltaStep}
\SetKwFunction{StochMoment}{StochMoment}
\SetKwFunction{ThetaUpdate}{ThetaUpdate}

\begin{algorithm}[tbp]\small
\caption{Stochastic Nested Fixed Point (SNFP)}
\label{alg:snfp-skeleton}
\DontPrintSemicolon
\Input{Stream or large dataset $\{(S_{jt},X_{jt},Z_{jt})\}$; initial $(\hat\theta_{T_0},\delta^\ast)$; step sizes $\{\gamma_t\}$}
\Output{Final estimate $\hat\theta$}

\BlankLine

\textbf{Pilot initialization} \;
Set a pilot dataset from $t=1,\ldots,T_0$\;
Compute a consistent pilot estimate $\hat{\theta}_{T_0}$ (e.g.\ (offline) NFP, MPEC, or SLIM \citep{chen2026slim})\;
Compute a consistent scaling matrix $\hat A_{T_0}$\;
\[
\hat A_{T_0} = \left(\hat \Phi_{T_0}' \hat W_{T_0} \hat \Phi_{T_0}\right)^{\dagger} \hat \Phi_{T_0}' \hat W_{T_0}
\]where, for a small number $\eta^\ast>0$,
\[
\hat \Phi_{T_0} = \frac{1}{T_0}\sum_{t=1}^{T_0} \nabla_\theta g_t(\hat \theta_{T_0})\mbox{  and  }
\hat W_{T_0} = \left(\frac{1}{T_0}\sum_{t=1}^{T_0} g_t(\hat\theta_{T_0}) g_t(\hat\theta_{T_0})' + \eta^\ast I_{d_z} \right)^{-1}.
\]\;

\BlankLine

\textbf{Main stochastic estimation procedure}\;
Set $\theta_{T_0}\leftarrow \hat\theta_{T_0}$; $\bar\theta_{T_0}\leftarrow \theta_{T_0}$; and $A\leftarrow \hat A_{T_0}$\;

\BlankLine

\For{$t=T_0+1,\,T_0+2,\,\ldots,\,T$}{

  \BlankLine
  \textbf{Inner loop (market $t$ only)}\;

  \Repeat{$\|\delta_t^{(m+1)}-\delta_t^{(m)}\|_\infty \le \eta $}{
    \tcp{Simulated market shares $\tilde S_t(\delta_t^{(m)};\theta_{t-1,2})$}
\For{$r=1,\ldots,R$}{
  Compute $\nu_{jrt} \leftarrow x_{jt}' \Sigma(\theta_{t-1,2}) v_r$, $j=1,\ldots,J_t$\;
  Compute
  $
  P_{jrt}
  \leftarrow
  \dfrac{\exp(\delta_{jt}^{(m)}+\nu_{jrt})}
  {1+\sum_{k=1}^{J_t}\exp(\delta_{kt}^{(m)}+\nu_{krt})}
  $,
  $j=1,\ldots,J_t$\;
}
Set
$
\tilde S_{jt}(\delta_t^{(m)};\theta_{t-1,2})
\leftarrow
\frac{1}{R}\sum_{r=1}^R P_{jrt}
$,
$j=1,\ldots,J_t$\;
    Update
    $
    \delta_t^{(m+1)} \leftarrow \delta_t^{(m)} + \log S_t - \log \tilde S_t(\delta_t^{(m)};\theta_{t-1,2})
    $\;
  }
  Set $\tilde\delta_t \leftarrow \delta_t^{(m)}$\;

  \BlankLine
  \tcp{Construct the demand shock and market moment}
  Compute
  $
  \tilde\xi_t(\theta_{t-1})
  \leftarrow
  \tilde\delta_t - X_t \theta_{t-1,1}
  $\;
  Set the market moment
  $
  g_t(\theta_{t-1})\leftarrow Z_t'\tilde\xi_t(\theta_{t-1})
  $\;

  \BlankLine
  \tcp{Robbins--Monro update}
  Update parameter:
  $
  \theta_t\leftarrow \Pi_\Theta\bigl\{\theta_{t-1} - \gamma_t\, A\, g_t(\theta_{t-1})\bigr\}
  $\;

  \BlankLine
  \tcp{Polyak--Ruppert averaging}
  Update
  $
  \bar\theta_t\leftarrow \frac{t-1-T_0}{t-T_0}\bar\theta_{t-1} + \frac{1}{t-T_0}\theta_t
  $\;
}

\BlankLine
\Return $\hat\theta \leftarrow \bar \theta_T$\;

\BlankLine

\end{algorithm}

We have a few remarks. First, this algorithm can be viewed as a variant of the nonlinear stochastic GMM (SGMM) in
\citet{chen2025sgmm}, with mean utilities updated in a fully online manner. Combined with SGD updates for $\theta$,
this yields substantial gains in computational speed relative to the standard NFP--GMM approach, which requires
full-sample updates and nonlinear optimization. Second, the SNFP algorithm is also memory efficient, because
each iteration uses only a single observation (market), and no full dataset storage is required during estimation.
These computational advantages are illustrated in the Monte Carlo simulations in Section~\ref{section: Monte Carlo}.
Finally, building on the MPEC formulation for BLP \citep{dube2012improving} (and subsequent practical guidance in
\citet{conlon2020best}), it would be natural to develop a stochastic MPEC variant within an online framework, which
we leave to future work. 

\subsection*{Implementation Details}
We conclude this section by summarizing the implementation choices used throughout the paper and justified by the theoretical results developed below.
\rev{As a starting configuration, we recommend a pilot of $T_0 \ge 1{,}000$ markets to compute the starting value and the
scaling matrix, step sizes $\gamma_t = t^{-2/3}$, that is, $(\gamma_0, a) = (1, 2/3)$, a single pass as the
initial budget when $T$ is in the hundreds of thousands or more, and additional shuffled passes at smaller $T$.
When one block of draws is shared across markets, we recommend scrambled Sobol' or Halton draws with $R$ large
enough for the integration-error condition of Corollary~\ref{cor:qmc}. The plug-in sandwich variance applies
under that condition or under fresh per-market draws with $\sqrt T/R \to 0$, and shared i.i.d.\ draws require
the correction of Remark~\ref{rem:common-draws-use}.}

\begin{comment}

\subsection*{Testing for Local Optima (To Be Completed)}

Since the preceding GMM objective function is not globally convex there is no guarantee that a stable point from the algorithm is the global optimizer. Since the global optimum converges to zero in case of GMM, we apply the overidentification test to check it out. If the test rejects, then we look for another stable point with a different starting point. 

Since the stochastic algorithm does not directly solve the first order condition for the full sample GMM, we follow \cite{chen2026slim} and employ the debiased J-test statistic in their Section 5.2. Specifically, let $\bar{\Phi}=\Phi_{T}$, $\bar{W}=W_{T}$, and $\bar{g}=\bar{g}_{T}(\bar{\theta}_T )$.\footnote{Alternatively, we may recompute $\bar{\Phi}$ and $\bar{W}$ using $\bar{\theta}_T$ and full sample. } 
Define
\[
\bar{g}_{n}^{D}=\Big(I-\bar{W}^{1/2}\bar{\Phi}
(\bar{\Phi}' \bar{W} \bar{\Phi})^{-1}\bar{\Phi}'\bar{W}^{1/2}\Big)\bar{W}^{1/2}\bar{g}.
\] and
\[
J_{\mathrm{D}} := n \, \bar{g}_{n}^{D}{}' \bar{g}_{n}^{D}
= n \, \bar{g}' \!\left(\bar{W}-\bar{W}\bar{\Phi}(\bar{\Phi}' \bar{W}\bar{\Phi})^{-1}\bar{\Phi}'\bar{W}\right)\! \bar{g}.
\]

the $J$-statistic
\[
J := n \, \bar{g}' \!\left(\bar{W}-\bar{W}\bar{\Phi}(\bar{\Phi}' \bar{W}\bar{\Phi})^{-1}\bar{\Phi}'\bar{W}\right)\! \bar{g}.
\]
\end{comment}

\section{Asymptotic Theory}\label{section:asymptotic_theory}

In this section, we establish the large-sample properties of the SNFP estimator as the number of markets $T$ grows.
When each market receives a fresh block of simulation draws, Section~\ref{subsection:large_sample} shows that the
SNFP estimator is consistent and attains the classical GMM limit distribution once $R$ grows faster than
$\sqrt T$. Section~\ref{subsection:common_draws} then turns to a common block of draws shared across markets and
shows how both the rate requirement on $R$ and the limiting variance change.

\subsection{Large-Sample Properties of the SNFP Estimator}\label{subsection:large_sample}

The SNFP algorithm operates in two phases.
The pilot phase uses the first $T_0$ markets and a shared block $\{\nu_r^{\mathrm{pilot}}\}_{r=1}^R$ of $R$ simulation
draws to produce the offline NFP--GMM estimate $\hat\theta_{T_0}$ and the pilot scaling matrix $\hat A_{T_0}$.
In the online phase, market $O_t = (Z_t, X_t, S_t)$ arrives at $t = T_0+1, \ldots, T$ with a fresh i.i.d.\ block
$\{\nu_{r,t}\}_{r=1}^R$ of $R$ simulation draws and the iterate is updated as
\begin{equation}\label{eq:snfp}
  \theta_t = \Pi_\Theta\!\bigl\{\theta_{t-1} - \gamma_t\,\hat A_{T_0}\,g(O_t;\theta_{t-1},P_t^R)\bigr\},
  \qquad \gamma_t = \gamma_0\, t^{-a},
\end{equation}
where $\Pi_\Theta$ is Euclidean projection onto $\Theta\subset\mathbb R^d$, $g$ is the BLP moment function, and
$P_t^R := R^{-1}\sum_{r=1}^R \delta_{\nu_{r,t}}$ is the empirical measure of the $R$ simulation draws for market $t$
, and $P^\ast$ denotes the common distribution of the draws, so that $P_t^R$ is its empirical counterpart.
The Polyak--Ruppert average is $\bar\theta_T := (T-T_0)^{-1}\sum_{t=T_0+1}^T \theta_t$.
Define the population objects
\[
  G(\theta,P) := E[\nabla_\theta g(O;\theta,P)], \qquad G^\ast := G(\theta^\ast,P^\ast), \qquad
  \Omega := \mathrm{Var}[g(O;\theta^\ast,P^\ast)],
\]
\[
  A^\ast := (G^{*\top}WG^\ast)^{-1}G^{*\top}W, \qquad \Sigma_\theta := A^\ast\,\Omega\,A^{*\top},
\]
where $W$ is the almost sure limit of the GMM weighting matrix and by construction $A^\ast G^\ast = I_d$.
Let $\mathcal F_t$ denote the natural online filtration, formally defined in Appendix~\ref{appendix:lemmas}
together with the shorthand $\Delta A$, $g_t^\circ$, $u_t^R$, and $e_t$ used throughout the proofs.

Our analysis follows the SLIM framework of \citet{chen2026slim}. Because the pilot phase supplies a consistent
estimate $\hat\theta_{T_0}$ and a consistent scaling matrix $\hat A_{T_0}$ that is held fixed thereafter, the
recursion~\eqref{eq:snfp} is the efficient (second-order) version of SLIM applied to the \emph{simulated} moments
$g(O_t;\theta,P_t^R)$. If $R = \infty$, the simulated moment coincides with the exact BLP moment
$g(O_t;\theta,P^\ast)$ and the estimation problem reduces to the efficient SLIM without simulation error. This
observation organizes the analysis. We first establish the asymptotic properties of the simulation-free
oracle and
then show that the approximation error induced by the finite number of simulation draws is negligible provided $R$
grows sufficiently fast. Formally, define the \emph{oracle recursion}
\begin{equation}\label{eq:snfp-oracle}
  \theta_t^o = \Pi_\Theta\!\bigl\{\theta_{t-1}^o - \gamma_t\,\hat A_{T_0}\,g(O_t;\theta_{t-1}^o,P^\ast)\bigr\},
  \qquad \theta_{T_0}^o = \hat\theta_{T_0},
\end{equation}
driven by the same markets, the same pilot quantities, and the same learning rates as~\eqref{eq:snfp}, and let
$\bar\theta_T^o := (T-T_0)^{-1}\sum_{t=T_0+1}^T \theta_t^o$ denote its Polyak--Ruppert average.
Theorems~\ref{thm:consistency-snfp} and \ref{thm:oracle-clt} below establish consistency and asymptotic normality
for the oracle sequence, and Theorem~\ref{thm:negligible-sim} shows that
$\sqrt T\,(\bar\theta_T - \bar\theta_T^o) \pto 0$ whenever $\sqrt T / R \to 0$, so that the feasible SNFP estimator
inherits the oracle limit distribution.

We first collect the conditions the results require, and discuss them after the statements.

\setcounter{assum}{0}

\begin{assum}[Random sample and independence]\label{ass:sampling}
$\{O_t\}_{t\ge 1}$ is i.i.d.
The pilot simulation block $\{\nu_r^{\mathrm{pilot}}\}_{r=1}^R$, the online simulation blocks $\{\nu_{r,t}\}_{r=1}^R$ for
$t > T_0$, and the markets $\{O_t\}$ are mutually independent. Online blocks are i.i.d.\ across $t$.
\end{assum}

\begin{assum}[Parameter space and projection]\label{ass:param-space}
$\Theta \subset \mathbb R^d$ is compact and convex, and $\theta^\ast \in \mathrm{int}(\Theta)$.
The online recursion is~\eqref{eq:snfp}.
\end{assum}

\begin{assum}[Identification and rank]\label{ass:identification}
$\theta^\ast$ is the unique solution of $E[g(O;\theta,P^\ast)] = 0$ in $\Theta$.
$G^\ast$ has full column rank $d$ and the population weighting matrix $W$ is symmetric positive definite.
Consequently $G^{*\top}WG^\ast$ is invertible and $A^\ast G^\ast = I_d$.
\end{assum}

\begin{assum}[Rate conditions]\label{ass:rates}
\hspace{0pt}
\begin{enumerate}
\item[(i)] $\gamma_t = \gamma_0\, t^{-a}$ with $\gamma_0 > 0$ and $a \in (1/2, 1)$.
\item[(ii)] $R = R(T) \to \infty$ as $T \to \infty$.
\item[(iii)] $T_0 = T_0(T) \to \infty$ and $T_0/T \to 0$ as $T \to \infty$.
\end{enumerate}
\end{assum}

\begin{assum}[Regular markets]\label{ass:regular-markets}
There exist constants $\bar J < \infty$, $C_0 < \infty$, and $\underline s > 0$ such that, almost surely,
the following hold.
(i) $J_t \le \bar J$.
(ii) $\max_{1 \le j \le J_t}\bigl(\|x_{jt}\| \vee \|z_{jt}\|\bigr) \le C_0$.
(iii) $\min_{0 \le j \le J_t} S_{jt} \ge \underline s$, where $S_{0t}$ denotes the market share of the outside
good.
\end{assum}

\begin{assum}[Local Lyapunov drift]\label{ass:drift}
There exist $\delta_0 > 0$, $C_1, C > 0$ such that, with the population mean moment $\bar m(\theta,P^\ast) := E[g(O;\theta,P^\ast)]$,
\begin{align*}
  \|\theta - \theta^\ast\| \ge \delta_0 &\;\Longrightarrow\;
   (\theta - \theta^\ast)^\top A^\ast \bar m(\theta,P^\ast) \ge C_1, \\
  \|\theta - \theta^\ast\| \le \delta_0 &\;\Longrightarrow\;
   (\theta - \theta^\ast)^\top A^\ast \bar m(\theta,P^\ast) \ge C\|\theta - \theta^\ast\|^2.
\end{align*}
Together with compactness of $\Theta$ (Assumption~\ref{ass:param-space}), there exists $\tilde C > 0$ such that
\[
(\theta - \theta^\ast)^\top A^\ast \bar m(\theta,P^\ast) \ge \tilde C\|\theta - \theta^\ast\|^2
\quad \text{for all } \theta \in \Theta.
\]
The map $\bar m(\cdot,P^\ast)$ is $L_G$-Lipschitz on $\Theta$ in $\theta$.
\end{assum}

We state the pilot conditions in terms of the measures that generate the pilot quantities. Let
$P^R := R^{-1}\sum_{r=1}^R\delta_{\nu_r^{\mathrm{pilot}}}$ be the empirical measure of the pilot draws,
$P_w^{T_0} := T_0^{-1}\sum_{t=1}^{T_0}\delta_{O_t}$ the empirical measure of the pilot markets, and $P_w^\ast$ the
distribution of $O$. For a parameter $\theta$, a simulation measure $P$, and a data measure $P_w$, let
$\Phi(\theta,P,P_w)$ and $W(\theta,P,P_w)$ be the pilot Jacobian and weighting matrices of
Section~\ref{section:algorithm}, with the sample averages over pilot markets replaced by integrals against $P_w$ and
the pilot draws replaced by $P$, and let $A(\theta,P,P_w) := (\Phi^\top W\Phi)^{\dagger}\Phi^\top W$. Then
$\hat\Phi_{T_0} = \Phi(\hat\theta_{T_0},P^R,P_w^{T_0})$, $\hat W_{T_0} = W(\hat\theta_{T_0},P^R,P_w^{T_0})$, and
$\hat A_{T_0} = A(\hat\theta_{T_0},P^R,P_w^{T_0})$.

\begin{assum}[Pilot regularity]\label{ass:pilot}
There exist matrix functions $\Phi^\ast$ and $W^\ast$ on $\Theta$, a neighborhood $N(\theta^\ast)$ of
$\theta^\ast$, and a constant $c > 0$ such that the following hold, with limits taken as $T_0, R \to \infty$.
\begin{enumerate}
\item[(i)] $\sup_{\theta\in \rev{N(\theta^\ast)}}\|\Phi(\theta,P^R,P_w^{T_0}) - \Phi^\ast(\theta)\| = o_p(1) $,
  and likewise for $W$ and $W^\ast$.
\item[(ii)] $\Phi^\ast(\theta^\ast) = G^\ast$ and $W^\ast(\theta^\ast) = W$, the weighting matrix in $A^\ast$.
\item[(iii)] $\inf_{\theta\in N(\theta^\ast)}\lambda_{\min}\bigl(\Phi^{\ast\top}(\theta) W^\ast(\theta)
  \Phi^\ast(\theta)\bigr) \ge c$, and $\Phi^\ast$ and $W^\ast$ are Lipschitz on $N(\theta^\ast)$.
\item[(iv)] $\hat\theta_{T_0} \pto \theta^\ast$.
\end{enumerate}
\end{assum}

Assumptions~\ref{ass:sampling}--\ref{ass:identification} are standard GMM conditions. The i.i.d.\ markets
ensure that the oracle moments at
the truth form an i.i.d.\ sequence (the key CLT input), the compact convex parameter space supports the projection
in~\eqref{eq:snfp}, and the interior identification with full-rank Jacobian guarantees that $A^\ast G^\ast = I_d$.

The rate conditions in Assumption~\ref{ass:rates} are mild. The exponent $a \in (1/2,1)$ places us in the
Robbins--Monro regime
needed for Polyak--Ruppert averaging, while $T_0 \to \infty$ with $T_0 = o(T)$ trades off pilot accuracy (so that
$\hat\theta_{T_0}$ and $\hat A_{T_0}$ concentrate) against sample negligibility (so that discarding the pilot
markets from the average is innocuous). For the number of simulation draws $R$, consistency requires only $R \to \infty$ in
Assumption~\ref{ass:rates}(ii). Asymptotic normality at the $\sqrt T$ scale additionally requires $\sqrt T / R \to 0$\rev{, as for the full-sample
estimator with fresh per-market draws \citep{freyberger2015asymptotic}}. In the polynomial parametrization
$R = T^{b_1}$ and
$T_0 = T^{b_2}$, this amounts to $b_1 > 1/2$ with $b_2 \in (0,1)$ unrestricted. The requirement is driven solely by
the simulation \emph{bias}. Because the mean utility $\tilde\delta_t$ is a nonlinear function of the
simulated
shares, the simulated moment is biased of order $O(R^{-1})$ (Lemma~\ref{lem:sim-error}), and this bias enters the
Polyak--Ruppert average multiplied by $\sqrt T$. The simulation \emph{variance}, by contrast, contributes only
$O_p(R^{-1/2})$ to the scaled average and vanishes for any divergent $R$. This contrasts with the method of
simulated moments with moments linear in the simulator \citep{mcfadden1989method,pakes1989simulation}, where a
fixed $R$ leaves the estimator consistent and merely inflates the asymptotic variance by the factor $1 + R^{-1}$.
These conditions presume a fresh block for each market. Section~\ref{subsection:common_draws} treats \rev{a
common
block shared by all markets}, under which both the rate requirement on $R$ and the limiting variance change.

Assumption~\ref{ass:regular-markets} imposes primitive regularity on the markets. The number of products, the
product
characteristics, and the instruments are bounded, and the observed shares, including that of the outside good, are
bounded away from zero. No high-level conditions on the moment function are required. Because the mixed-logit
share map is smooth in $(\delta,\theta)$ and \emph{linear} in the simulation measure, and because the simulated
shares coincide with the observed shares at the fixed point of the BLP inversion, Assumption~\ref{ass:regular-markets} implies
that $g(O;\theta,P)$ and $\nabla_\theta g(O;\theta,P)$ are uniformly bounded and Lipschitz in $\theta$, and that,
for $h \in \{g, \nabla_\theta g\}$,
\begin{align*}
  & \bigl\|E[h(O;\theta,P_t^R)] - E[h(O;\theta,P^\ast)]\bigr\| = O(R^{-1}), \\
  & E\bigl\|h(O;\theta,P_t^R) - h(O;\theta,P^\ast)\bigr\|^2 = O(R^{-1}),
\end{align*}
uniformly in $\theta \in \Theta$. The first line gives a sharp $O(R^{-1})$ control on the deterministic
simulation bias, which arises because the mean utility is a nonlinear function of the simulated shares. The
second gives the $O(R^{-1})$ mean-square deviation of the simulated moment from its oracle counterpart, which
holds because the simulated shares average $R$ i.i.d.\ draws with conditional mean equal to the exact shares
given $O$ (Assumption~\ref{ass:sampling}). These properties are recorded in Lemmas~\ref{lem:blp-regularity} and
\ref{lem:sim-error} in the appendix. Analogous expansions of the BLP inversion in the simulation error are
developed by \citet{berry2004limit} and \citet{freyberger2015asymptotic}.

Assumption~\ref{ass:drift} is the key stability condition. The scaled population mean moment
$A^\ast \bar m(\theta,P^\ast)$ always
points toward $\theta^\ast$, with quadratic drift locally and a uniform drift globally that lifts to a quadratic drift
on all of $\Theta$ via compactness.

Assumption~\ref{ass:pilot} collects the regularity conditions for the pilot phase to be consistent. \rev{By Assumption~\ref{ass:pilot}(i) and (iv), $\hat\Phi_{T_0} = \Phi^\ast(\hat\theta_{T_0}) + o_p(1)$ and $\hat W_{T_0} = W^\ast(\hat\theta_{T_0}) + o_p(1)$. Continuity of $\Phi^\ast$ and $W^\ast$ at $\theta^\ast$, (ii), and the eigenvalue bound in (iii) then give $\hat A_{T_0} \pto A^\ast$ as $T_0, R \to \infty$ by the continuous mapping theorem.} 

We are now ready to state the main results. The first theorem establishes consistency of both the oracle and the feasible SNFP sequences, together with their Polyak--Ruppert averages.

\begin{theorem}[Consistency]\label{thm:consistency-snfp}
Under Assumptions~\ref{ass:sampling}--\ref{ass:pilot}, as $T \to \infty$,
\[
  \theta_T^o \pto \theta^\ast, \qquad \bar\theta_T^o \pto \theta^\ast, \qquad
  \theta_T \pto \theta^\ast, \qquad \bar\theta_T \pto \theta^\ast.
\]
\end{theorem}

The projection, rather than the compactness of $\Theta$, is the fundamental device in the convergence proofs. Its
nonexpansiveness, $\|\Pi_\Theta(x) - \theta^\ast\| \le \|x - \theta^\ast\|$, is the only property of the constraint
used in the Lyapunov recursion behind Theorem~\ref{thm:consistency-snfp} and in the coupling recursion behind
Theorem~\ref{thm:negligible-sim}. It confines the iterates to the set on which the drift and moment conditions
are imposed, and it is asymptotically innocuous, since by Lemma~\ref{lem:projection} it binds at most finitely
many times along the oracle path and, with probability approaching one, never during the online phase. Compactness of
$\Theta$ is a simplification. With the projection retained the proofs extend to an unbounded parameter space
provided the quadratic drift holds globally and the moment bounds grow at most linearly in
$\|\theta - \theta^\ast\|$, as is standard \citep{robbins1951stochastic,polyak1992acceleration}.

The next theorem gives the limit distribution of the oracle average. It is the analogue, in the streaming
large-$T$ setting, of the asymptotic normality of the efficient SLIM estimator in \citet{chen2026slim}. With
$R = \infty$ there is no simulation error, and the only departures from the classical averaged stochastic
approximation of \citet{polyak1992acceleration} and \citet{ruppert1988efficient} are the estimated scaling matrix
$\hat A_{T_0}$ and the projection, both of which are shown to be asymptotically innocuous.

\begin{theorem}[Oracle asymptotic normality]\label{thm:oracle-clt}
Under Assumptions~\ref{ass:sampling}--\ref{ass:pilot}, as $T \to \infty$,
\[
  \sqrt T\,(\bar\theta_T^o - \theta^\ast) \;\xrightarrow{d}\; N(0,\Sigma_\theta).
\]
\end{theorem}

The third theorem is a key result. The approximation error due to the finite number of simulation draws
is
negligible at the $\sqrt T$ scale once $R$ grows faster than $\sqrt T$.

\begin{theorem}[Negligibility of the simulation error]\label{thm:negligible-sim}
Under Assumptions~\ref{ass:sampling}--\ref{ass:pilot}, if in addition $\sqrt T / R \to 0$, then, as $T \to \infty$,
\[
  \sqrt T\,(\bar\theta_T - \bar\theta_T^o) \pto 0.
\]
\end{theorem}

Combining Theorems~\ref{thm:oracle-clt} and \ref{thm:negligible-sim} yields the limit distribution of the feasible
SNFP estimator.

\begin{corollary}[Asymptotic normality of $\bar\theta_T$]\label{thm:asymptotic_dist}
Under Assumptions~\ref{ass:sampling}--\ref{ass:pilot} and $\sqrt T/R \to 0$, as $T \to \infty$,
\[
  \sqrt T\,(\bar\theta_T - \theta^\ast) \;\xrightarrow{d}\; N(0,\Sigma_\theta),
\]
where $\Sigma_\theta = (G^{*\top}WG^\ast)^{-1}G^{*\top}W\,\Omega\,WG^\ast(G^{*\top}WG^\ast)^{-1}$ is the classical
GMM sandwich variance. If, in addition, $W = \Omega^{-1}$, then
$\Sigma_\theta = (G^{*\top}\Omega^{-1}G^\ast)^{-1}$, the efficient GMM variance.
\end{corollary}

The corollary shows that the SNFP estimator, which touches each market exactly once and never recomputes a
full-sample Jacobian or weighting matrix, attains the same first-order asymptotic distribution as the infeasible
full-sample GMM estimator based on the exact share integrals.
Proofs are provided in Appendices~\ref{appendix:lemmas}--\ref{appendix:thm1_proof}.

\rev{The theory above covers a single pass over independent markets, and the main Monte Carlo results in
Section~\ref{section: Monte Carlo} run SNFP for a single pass. In practice additional shuffled passes (epochs) are
often run; Remark~\ref{rem:epochs} records their asymptotic status.}
 
\subsection{A Common Block of Simulation Draws}\label{subsection:common_draws}

Assumption~\ref{ass:sampling} draws a fresh block for each market. \rev{An alternative, used for instance whenever the
same agent data are supplied to every market, draws} one block $\{\nu_r\}_{r=1}^R$ once and \rev{reuses it in all
markets}, so that
$P_t^R = P^R := R^{-1}\sum_{r=1}^R \delta_{\nu_r}$ for all $t$. A common block is cheaper \rev{and} permits draw-level
precomputation\rev{. Like any fixed set of draws, it also avoids} the simulation ``chatter'' \rev{of redrawing across
parameter evaluations}
\citep{berry2004limit,conlon2020best}. It also changes the asymptotics, and the change is not innocuous.
\rev{\citet{freyberger2015asymptotic}'s Remark 3 states, without proof, the corresponding result for the full-sample estimator, and \citet{hong2021blp} prove $\sqrt{\min (R,T)}$ asymptotic normality with a consistent variance estimator; both concern a fixed full-sample objective, so neither carries over to SNFP, which never forms that objective and reuses the block in every update, so the online estimator requires a separate analysis.}

% Assumption 1' uses the assum environment with its number shown as 1'; the
% counter is restored afterwards so it does not consume an assumption number.
\begingroup
\renewcommand{\theassum}{1$'$}
\begin{assum}[Common simulation block]\label{ass:sampling-common}
$\{O_t\}_{t \ge 1}$ is i.i.d. The draws $\nu_1,\ldots,\nu_R$ are i.i.d.\ $N(0,I_{d_v})$, drawn once and
independently of $\{O_t\}$, and the same block is used in the pilot phase and in every online update, so that
$P_t^R = P^R$ for all $t$ and the recursion is~\eqref{eq:snfp} with $P^R$ in place of $P_t^R$.
\end{assum}
\endgroup
\addtocounter{assum}{-1}

Because the block is now common to all markets, $E[g(O;\theta,P^R) \mid P^R]$ is random, and we write it
$\bar m^{c}_R(\theta)$, set $b_R(\theta) := \bar m^{c}_R(\theta) - \bar m(\theta,P^\ast)$, and let the pseudo-true
parameter $\theta^R$ solve $A^\ast\bar m^{c}_R(\theta) = 0$. The deterministic $\bar m_R$ and $\beta_R$ of
Appendix~\ref{appendix:lemmas} are reserved for the fresh-draw objects of Lemma~\ref{lem:sim-error}. Define the
market-averaged simulation influence function
\begin{equation}\label{eq:psi-common}
  \psi(\nu;\theta) := E\bigl[\mathcal L_O(\theta)\bigl\{\sigma(\tilde\delta^\ast(O,\theta_2),\theta_2,\nu)
    - \tilde s(\tilde\delta^\ast(O,\theta_2),\theta_2;P^\ast)\bigr\}\bigr],
\end{equation}
where $\sigma(\delta,\theta_2,\nu)$ is the vector of inside-good choice probabilities at a single draw,
$\tilde\delta^\ast(O,\theta_2) := \tilde\delta(S,\theta_2;P^\ast)$, and $\mathcal L_O(\theta)$ is the BLP inversion
multiplier $\mathcal L_t(\theta)$ of the proof of Lemma~\ref{lem:sim-error}, written for a generic market $O$.
The expectation is over the market alone, so $E_\nu[\psi(\nu;\theta)] = 0$. Set
$\Omega_\nu := E[\psi(\nu;\theta^\ast)\psi(\nu;\theta^\ast)^\top]$ and
$\Lambda_\nu := A^\ast\Omega_\nu A^{\ast\top}$.

The conditional-mean property behind Assumption~\ref{ass:sampling}, namely that the simulated shares average $R$
draws that are independent across markets, is exactly what fails here. Conditionally on the block, \eqref{eq:snfp}
is an \emph{exact} stochastic approximation to $\bar m^{c}_R$. No simulation noise is left to control, the
coupling
argument of Theorem~\ref{thm:negligible-sim} is not needed, and Theorems~\ref{thm:consistency-snfp}
and~\ref{thm:oracle-clt} apply verbatim, centered at $\theta^R$ rather than at $\theta^\ast$. The cost appears in the
recentering. All the simulation randomness is frozen into $b_R$, which linearizes in $\psi$
(Lemma~\ref{lem:perturbation}), that is, $b_R(\theta) = R^{-1}\sum_{r=1}^R \psi(\nu_r;\theta) + O_p(R^{-1})$
uniformly in
$\theta$, so $\theta^R - \theta^\ast = -A^\ast b_R(\theta^\ast) + O_p(R^{-1})$ and
$\sqrt R\,(\theta^R - \theta^\ast) \xrightarrow{d} N(0,\Lambda_\nu)$. A common block displaces the estimand. The
next theorem shows how this displacement enters the limit distribution of $\bar\theta_T$.

\begin{theorem}[Common simulation draws]\label{thm:common-draws}
Under Assumption~\ref{ass:sampling-common} and Assumptions~\ref{ass:param-space}--\ref{ass:pilot}, as
$T \to \infty$:
\begin{enumerate}
\item[(i)] $\theta_T \pto \theta^\ast$ and $\bar\theta_T \pto \theta^\ast$, with no further condition linking $R$ and $T$.
\item[(ii)] If $T/R \to c \in [0,\infty)$, then
$\sqrt T\,(\bar\theta_T - \theta^\ast) \xrightarrow{d} N\bigl(0,\; \Sigma_\theta + c\,\Lambda_\nu\bigr)$, the two
Gaussian components being independent. In particular $c = 0$ recovers Corollary~\ref{thm:asymptotic_dist}.
\item[(iii)] If $R/T \to 0$, then $\sqrt R\,(\bar\theta_T - \theta^\ast) \xrightarrow{d} N(0,\Lambda_\nu)$. The
simulation error dominates and the rate drops to $\sqrt R$.
\end{enumerate}
\end{theorem}

The displacement is a property of i.i.d.\ draws. If the common block is instead built from a low-discrepancy point
set, its quadrature error is of smaller order and the fresh-draw limit theory returns.

\begin{corollary}[Common low-discrepancy block]\label{cor:qmc}
Suppose the i.i.d.\ block of Assumption~\ref{ass:sampling-common} is replaced by $\nu_r = F_v^{-1}(u_r)$,
$r = 1,\ldots,R$, generated coordinate-wise from a low-discrepancy point set $\{u_r\} \subset (0,1)^{d_v}$ whose
quadrature error for the choice probabilities and their first derivatives is bounded, uniformly on the relevant
compact index set, by a sequence $\epsilon_R \to 0$ (deterministic or, for a randomly scrambled net, satisfying
$\epsilon_R = O_p(R^{-1+\eta})$ for every $\eta > 0$, in which case the conclusions hold conditionally on the
scramble). Under Assumptions~\ref{ass:param-space}--\ref{ass:pilot}, as $T \to \infty$,
$\|b_R\|_\infty \le C\epsilon_R$, $\bar\theta_T \pto \theta^\ast$ under $R \to \infty$ alone, and
$\sqrt T\,(\bar\theta_T - \theta^\ast) \xrightarrow{d} N(0,\Sigma_\theta)$ whenever $\sqrt T\,\epsilon_R \to 0$,
with the plug-in sandwich variance valid without correction.
\end{corollary}

Digital nets and Halton sequences attain $\epsilon_R = O\bigl(R^{-1}(\log R)^{d_v}\bigr)$
\citep{niederreiter1992random,owen2006halton,dick2010digital}, so in the polynomial parametrization $R = T^{b_1}$
any $b_1 > 1/2$ suffices, exactly as under fresh draws.

The gap between the requirement $T/R \to 0$ of Theorem~\ref{thm:common-draws}, that is, $b_1 > 1$, and the
requirement $b_1 > 1/2$ noted after Assumption~\ref{ass:rates} is driven by the correlation structure of the simulation
error, not by its size. With
fresh draws the mean-zero $O_p(R^{-1/2})$ part of that error is a martingale difference across markets. It
diversifies in the Polyak--Ruppert average, leaving only the $O(R^{-1})$ nonlinearity bias, whose cost is
$\sqrt T/R$. With a common block it is the \emph{same} random variable in every market, so it cannot average out,
it displaces the estimand by $\theta^R - \theta^\ast = O_p(R^{-1/2})$, and its cost is $\sqrt{T/R}$. The
nonlinearity bias is then dominated.

\begin{remark}\label{rem:common-draws-use}
Part~(ii) of Theorem~\ref{thm:common-draws} rehabilitates moderate $R$ provided inference accounts for
$c\,\Lambda_\nu$, and every ingredient is a byproduct of the online pass. The plug-in sandwich variance, widened by
\begin{equation}\label{eq:widened-variance}
  (T/R)\,\hat A_{T_0}\hat\Omega_\nu\hat A_{T_0}^\top,
\end{equation}
is valid. Here $\hat\Omega_\nu := R^{-1}\sum_{r} \hat\psi_r\hat\psi_r^\top$ with
$\hat\psi_r := (T-T_0)^{-1}\sum_{t > T_0} \hat{\mathcal L}_t\bigl[\sigma_t(\nu_r) - \tilde s_t\bigr]$, where
$\sigma_t(\nu_r) := \sigma(\tilde\delta_t,\theta_{t-1,2},\nu_r)$ and $\tilde s_t := R^{-1}\sum_r \sigma_t(\nu_r)$ are
the single-draw and simulated shares of market $t$ at the online iterate, and $\hat{\mathcal L}_t :=
-Z_t'\bigl[\partial \tilde s_t/\partial\delta'\bigr]^{-1}$ is the inversion
multiplier $\mathcal L_t(\theta)$ of the proof of Lemma~\ref{lem:sim-error} evaluated at the same point.
\rev{\citet{hong2021blp} give an analogous variance estimator for the full-sample GMM estimator under shared draws.} Inference computed from the iterate path itself, such as
random scaling \citep{lee2022fast}, is conditional on the block by construction, so its pivot is centered at the
pseudo-true parameter rather than at $\theta^\ast$, is blind to $c\,\Lambda_\nu$, and requires the same widening. \rev{For implementations that use a common block},
Corollary~\ref{cor:qmc} points to the combination we recommend, namely a single
shared block of
\rev{low-discrepancy (quasi-Monte Carlo) draws, such as scrambled Sobol' or Halton sequences}
\citep{owen1997scrambled,train2009discrete}, \rev{with $R$ large enough that $\sqrt T\,\epsilon_R$ is small}.
\end{remark}

Both results are proved in Appendix~\ref{appendix:common_draws}, which also collects the two perturbation lemmas on
which they rest. Section~\ref{section: Monte Carlo} examines all three regimes numerically.

\begin{remark}[Multiple shuffled epochs]\label{rem:epochs}
\rev{Additional shuffled passes (epochs) revisit markets, so Assumption~\ref{ass:sampling} fails from the second
epoch on. The remedy is the conditioning device of this section, now applied to the data: given the sample, the
draws, and the pilot, a shuffled epoch samples markets without replacement from the empirical distribution, and
the recursion is a stochastic approximation to $\hat A_{T_0}\bar g_T(\theta)$, whose zero $\hat\theta_T$ is the
full-sample estimating-equation estimator under the same draw scheme \citep{chen2026slim}.}

\rev{Data reuse thus displaces the target from $\theta^\ast$ to $\hat\theta_T$, exactly as draw reuse displaces it to
$\theta^R$. As epochs accumulate the averaged iterate converges, conditionally, to $\hat\theta_T$ and inherits
the full-sample asymptotics --- Corollary~\ref{thm:asymptotic_dist} under market-specific blocks, and the
shared-block theory of \citet[Remark~3]{freyberger2015asymptotic}, \citet{hong2021blp}, and
Theorem~\ref{thm:common-draws} otherwise --- while a small fixed number of epochs interpolates between the
single-pass and full-sample estimators, with the same first-order distribution at both ends
(Appendix~\ref{appendix:snfp_ep1}).}
\end{remark}

\section{Monte Carlo Experiments}\label{section: Monte Carlo}

In this section, we evaluate the SNFP estimator in Monte Carlo experiments. We examine its computational efficiency,
its statistical accuracy relative to the NFP benchmark, and the effect of simulation error from the random draws used
to approximate market shares.

The baseline data-generating process is based on \citet{freyberger2015asymptotic} and
\citet{dube2012improving}.
Each market $t = 1, \dots, T$ contains $J = 4$ differentiated products.
Consumer $i$'s indirect utility from product $j$ in market $t$ is
\[
u_{ijt} = X_{jt}\beta_i + \alpha_i p_{jt} + \xi_{jt} + \varepsilon_{ijt},
\]
where $X_{jt} = (1,\, x_{1jt})$ is a vector of observed product characteristics,
$x_{1jt} \sim \text{TN}(0,1)$ is a standard normal random variable truncated to $[-2,2]$,
$\xi_{jt} \sim \text{TN}(0,1)$ is an unobserved product characteristic, and
$\varepsilon_{ijt}$ is i.i.d.\ standard type I extreme value.
The endogenous price is generated by
\[
p_{jt} = \tfrac{1}{2} \bigl| 2 + 0.5\,\xi_{jt} + 1.1\,x_{1jt} + \eta_{jt} \bigr|,
\quad \eta_{jt} \sim \text{TN}(0,1).
\]
The dependence of $p_{jt}$ on $\xi_{jt}$ introduces price endogeneity.
The random coefficients $(\beta_i^0, \beta_i^1, \alpha_i)'$ on
$(1,\, x_{1jt},\, p_{jt})'$ follow the joint normal distribution
\[
\begin{pmatrix} \beta_i^0 \\ \beta_i^1 \\ \alpha_i \end{pmatrix}
\sim N\!\left(
\begin{pmatrix} -5 \\ 1.5 \\ -2 \end{pmatrix},
\begin{pmatrix} 0 & 0 & 0 \\ 0 & 0.5 & 0 \\ 0 & 0 & 0.5 \end{pmatrix}
\right).
\]
Thus, the standard deviations of the random coefficients on $x_{1jt}$ and $p_{jt}$ are
$\sigma_x = \sigma_p = \sqrt{0.5} \approx 0.707$, respectively. These are the nonlinear
parameters targeted by the NFP and SNFP estimators.
Market shares are simulated using $R = 1{,}000$ consumer draws:
\[
s_{jt} = \frac{1}{R} \sum_{r=1}^{R}
\frac{\exp\!\bigl(X_{jt}\beta_r + \alpha_r p_{jt} + \xi_{jt}\bigr)}
{1 + \sum_{k=1}^{J} \exp\!\bigl(X_{kt}\beta_r + \alpha_r p_{kt} + \xi_{kt}\bigr)}.
\]
Price endogeneity is addressed using instruments constructed from exogenous product
characteristics and auxiliary variation. The baseline instrument is
\[
z_{jt} = a_{jt} + 0.25(\eta_{jt} + 1.1\, x_{1jt}),
\quad a_{jt} \sim \text{U}(0,1).
\]
The full instrument vector consists of $1$, $x_{1jt}$, $z_{jt}$,
$x_{1jt} z_{jt}$, $x_{1jt}^2$, $z_{jt}^2$, $x_{1jt}^3$, and $z_{jt}^3$.
Instrument relevance follows from the dependence of $z_{jt}$ on $x_{1jt}$ and $\eta_{jt}$,
while the exclusion restriction holds because $z_{jt}$ is independent of $\xi_{jt}$.
Simulations are conducted at market sizes
\[
T \in \{2{,}000,\, 8{,}000,\, 32{,}000,\, 128{,}000,\, 512{,}000\},
\]
with $J = 4$ products per market and $R = 1{,}000$ simulated consumers. Computation time and memory usage are evaluated over 10 replications in a standardized computing environment. Statistical accuracy is evaluated over $1{,}000$ replications on the cluster system.

\subsection{Computational Efficiency}\label{subsection:mc_efficiency}

Table~\ref{tab:time_memory} summarizes the computational cost comparison between SNFP and NFP. Each
replication runs serially on one core of an AMD EPYC 9654 node of the Rorqual cluster of the Digital Research
Alliance of Canada. Both estimators read the same data file, a CSV file with one line per market. SNFP reads it one
line at a time in file order\rev{, visiting each market exactly once}. NFP
reads the whole file into memory once, holds the data as full-sample arrays\rev{, and is run to convergence}. We run 10 Monte Carlo replications
and report the average computation time and peak memory usage, together with their standard deviations. The computation time for SNFP includes both the pilot step for
constructing the scaling matrix and the SGD \rev{updates}, whereas the computation time for NFP is the sum of
data-loading time and offline optimization time using the L-BFGS-B method. The pilot sample size for SNFP is
$T_0 = 1{,}000$ markets, \rev{the initial value is drawn as $\theta_0 \sim \mathrm{Unif}(\theta^\ast \pm 0.1)$,} and the remaining tuning parameters are $(\gamma_0, a) = (1,\, 2/3)$ with parameter
clipping $\sigma \in [10^{-4}, 2]$.

% =============================================================================
% Numbers: T = 2,048,000 from tab:time_memory (local 3970X, .rds per market).
%   T = 100,000,000: Experimental_Design_2026_09_12_bigN, rorqual jobs
%   20985186-20985187, 10 reps, one epoch, results/tables/summary_T100000000.csv.
%   RMSE over the 10 reps (4.0e-4, 1.2e-4, 8.4e-4, 7.8e-5, 4.8e-4), 0.92-1.09
%   extrapolated NFP SE (SE at 512k times sqrt(512000/1e8)).
%   Min/max at 2,048,000: simulations_in_the_paper/tab1_time_memory/results
%   (est_sec from snfp_stream/T2048000/result_rep*.csv, max_rss_kb from timing_T0_1000.csv).
% =============================================================================
\begin{table}[tbp]
\centering
\caption{Computational cost of NFP vs.\ SNFP.}
\label{tab:time_memory}
\small
\begin{threeparttable}
\begin{tabular}{r c r r r r r r}
\toprule
\multicolumn{1}{c}{\multirow{2}{*}{$T$}} & \multirow{2}{*}{\shortstack{Accuracy\\ level $d_T$}} & \multicolumn{3}{c}{Time (seconds)} & \multicolumn{3}{c}{Peak memory (MB)} \\
\cmidrule(lr){3-5} \cmidrule(lr){6-8}
 & & SNFP & NFP & NFP/SNFP & SNFP & NFP & NFP/SNFP \\
\midrule
$2{,}000$ & $\underset{(2.19)}{1.73}$ & $\underset{(0.1)}{0.8}$ & $\underset{(20.4)}{34.9}$ & $45.1\times$ & $\underset{(0.7)}{86.8}$ & $\underset{(0.1)}{87.1}$ & $1.0\times$ \\[4pt]
$8{,}000$ & $\underset{(0.23)}{0.59}$ & $\underset{(0.1)}{1.9}$ & $\underset{(83.3)}{161.9}$ & $84.8\times$ & $\underset{(0.0)}{87.2}$ & $\underset{(0.2)}{98.5}$ & $1.1\times$ \\[4pt]
$32{,}000$ & $\underset{(0.53)}{0.46}$ & $\underset{(0.5)}{6.9}$ & $\underset{(487.2)}{1{,}008.9}$ & $146.1\times$ & $\underset{(0.8)}{86.8}$ & $\underset{(0.6)}{140.5}$ & $1.6\times$ \\[4pt]
$128{,}000$ & $\underset{(0.14)}{0.39}$ & $\underset{(1.0)}{25.0}$ & $\underset{(1{,}533.9)}{2{,}291.8}$ & $91.6\times$ & $\underset{(0.9)}{86.6}$ & $\underset{(1.0)}{288.6}$ & $3.3\times$ \\[4pt]
$512{,}000$ & $\underset{(0.15)}{0.35}$ & $\underset{(3.4)}{101.1}$ & $\underset{(5{,}149.8)}{12{,}441.4}$ & $123.1\times$ & $\underset{(0.6)}{86.9}$ & $\underset{(1.2)}{859.0}$ & $9.9\times$ \\[4pt]
$2{,}048{,}000$ & --- & $\underset{(18.7)}{409.6}$ & --- & --- & $\underset{(0.6)}{97.4}$ & --- & --- \\[4pt]
$100{,}000{,}000$ & --- & $\underset{(873.2)}{19{,}643.4}$ & --- & --- & $\underset{(1.3)}{97.6}$ & --- & --- \\
\bottomrule
\end{tabular}
\begin{tablenotes}[flushleft]
\footnotesize
\item \textit{Notes.} \rev{Time and memory are means across 10 replications, with standard deviations in
parentheses. The accuracy level is the median of $d_T$ across replications, with its standard deviation in
parentheses.}
\end{tablenotes}
\end{threeparttable}
\end{table}

\rev{SNFP is substantially cheaper than NFP at every sample size.} Table~\ref{tab:time_memory} reports speedup factors between $45.1$ and $146.1$. At $T = 512{,}000$, SNFP
completes in $101.1$ seconds against $12{,}441.4$ seconds for NFP. The two estimators also differ in the variability of their computation times.
The standard deviation of computation time is \rev{smaller for SNFP by two orders of magnitude or more} at every $T$, and the computation time of
NFP has a heavy right tail. At $T = 128{,}000$ the slowest NFP replication took $6{,}241$ seconds against a
median of $1{,}736$ seconds.

In terms of peak memory usage, NFP requires more memory than SNFP as the sample size increases, from parity at
$T = 2{,}000$ to $9.9$ times as much at $T = 512{,}000$. The peak memory of SNFP stays near $87$~MB at every $T$ \rev{up to $512{,}000$},
because SNFP holds one market in memory at a time while NFP retains the whole sample.

\rev{SNFP scales to $100$ million markets within a 10-hour computation-time budget, which NFP already exceeds at
$T = 2{,}048{,}000$. The last two rows of Table~\ref{tab:time_memory} show this.} At $T = 2{,}048{,}000$, SNFP completes
in a mean of $409.6$ seconds\rev{, between $387.9$ and $439.4$ seconds across replications,} with a peak memory of $97.4$~MB. At $T = 100{,}000{,}000$ markets with $J = 4$ products,
or $400$ million product-market observations, SNFP completes in a mean of $19{,}643.4$ seconds\rev{, between $18{,}204.0$ and $20{,}982.2$ seconds}, about $5.5$ hours and
well within the budget, with a peak memory of $97.6$~MB. Peak memory is essentially unchanged between the two sample
sizes, so memory use does not grow with $T$.

\rev{Finally, the accuracy level column of Table~\ref{tab:time_memory} records how far the SNFP estimate is from the
full-sample estimator.} Let $\hat\theta^{\mathrm{NFP}}$ denote the NFP estimate and
$\widehat{\mathrm{se}}_k$ its standard error for coordinate $k$, computed on the same replication and seed, and let
$\bar\theta_t$ denote the Polyak--Ruppert average of the SNFP iterates after market $t$. Define
\[
  d_t \;=\; \max_k \frac{|\bar\theta_{t,k} - \hat\theta^{\mathrm{NFP}}_k|}{\widehat{\mathrm{se}}_k},
\]
\rev{the largest coordinate-wise distance between the two estimators in units of NFP standard errors. The accuracy
level reported in Table~\ref{tab:time_memory} is $d_T$, the distance after the last market, as the median across
the 10 replications. We report the median rather than the mean because the distance is right-skewed.
At $T = 2{,}000$, SNFP is
$1.73$ NFP standard errors from the NFP estimate, but the distance falls to $0.59$ at $T = 8{,}000$ and to $0.35$
at $T = 512{,}000$, so larger samples come closer to the NFP estimate from a single reading of the data. In
practice, the performance at small sample sizes can be improved by running multiple epochs, that is, additional
shuffled passes over the data; Appendix~\ref{appendix:additional_mc} reports these results.}

\subsection{Statistical Accuracy}\label{subsection:mc_accuracy}

Table~\ref{tab:stat_accuracy} documents the statistical accuracy of SNFP. In this simulation, we replicate each model $1{,}000$ times for the sample sizes
$T \in \{2{,}000,\, 8{,}000,\, 32{,}000,\, 128{,}000\}$. Given the computation time limit, NFP is applied only up to $T=32{,}000$.
The table reports, for each parameter, the bias, root mean squared error (RMSE), average 95\% confidence
interval length, and empirical coverage.
SNFP uses $T_0 = 1{,}000$ pilot markets and a sandwich
plug-in variance. Other tuning parameters are set to be the same as above.

% =============================================================================
% Numbers: single-epoch SNFP (SNFP v4 ep1 snapshots for T in {2k, 8k, 32k};
% v6 ep1 for T=128k) and NFP (second_trial, finite-only filter), n=1000 per
% tier. These are the SNFP-1 and NFP columns of the former tab:stat_accuracy_ep1,
% now tab:stat_accuracy_epochs in Appendix E. Three-epoch numbers live there.
% =============================================================================
\begin{table}[tbp]
\centering
\small
\caption{Comparison of NFP and SNFP.}
\label{tab:stat_accuracy}
\begin{threeparttable}
\begin{tabular}{ll *{8}{S[table-format=-1.3]}}
\toprule
& & \multicolumn{2}{c}{$T = 2{,}000$} & \multicolumn{2}{c}{$T = 8{,}000$} & \multicolumn{2}{c}{$T = 32{,}000$} & \multicolumn{2}{c}{$T = 128{,}000$} \\
\cmidrule(lr){3-4}\cmidrule(lr){5-6}\cmidrule(lr){7-8}\cmidrule(lr){9-10}
& & {NFP} & {SNFP} & {NFP} & {SNFP} & {NFP} & {SNFP} & {NFP} & {SNFP} \\
\midrule
& Bias \\
& \quad Intercept  &  0.003 &  0.061 & -0.003 &  0.014 &  0.001 &  0.007 &  {---} & -0.003 \\
& \quad $\beta_x$  &  0.001 &  0.009 & -0.001 &  0.003 &  0.000 &  0.002 &  {---} & -0.001 \\
& \quad $\beta_p$  & -0.002 & -0.086 &  0.006 & -0.025 & -0.001 & -0.014 &  {---} &  0.007 \\
& \quad $\sigma_x$ &  0.000 &  0.003 &  0.000 &  0.001 &  0.000 &  0.001 &  {---} &  0.000 \\
& \quad $\sigma_p$ & -0.008 &  0.028 & -0.004 & -0.009 &  0.000 & -0.006 &  {---} & -0.005 \\
[3pt] & RMSE \\
& \quad Intercept  &  0.089 &  0.133 &  0.045 &  0.053 &  0.023 &  0.025 &  {---} &  0.012 \\
& \quad $\beta_x$  &  0.025 &  0.044 &  0.013 &  0.016 &  0.006 &  0.007 &  {---} &  0.003 \\
& \quad $\beta_p$  &  0.176 &  0.234 &  0.088 &  0.101 &  0.045 &  0.050 &  {---} &  0.023 \\
& \quad $\sigma_x$ &  0.019 &  0.025 &  0.009 &  0.010 &  0.005 &  0.005 &  {---} &  0.002 \\
& \quad $\sigma_p$ &  0.114 &  0.099 &  0.054 &  0.056 &  0.028 &  0.030 &  {---} &  0.015 \\
[3pt] & CI length \\
& \quad Intercept  &  0.346 &  0.356 &  0.172 &  0.173 &  0.086 &  0.087 &  {---} &  0.043 \\
& \quad $\beta_x$  &  0.099 &  0.101 &  0.049 &  0.050 &  0.025 &  0.025 &  {---} &  0.012 \\
& \quad $\beta_p$  &  0.681 &  0.707 &  0.336 &  0.339 &  0.170 &  0.170 &  {---} &  0.084 \\
& \quad $\sigma_x$ &  0.074 &  0.075 &  0.037 &  0.037 &  0.018 &  0.018 &  {---} &  0.009 \\
& \quad $\sigma_p$ &  0.435 &  0.428 &  0.210 &  0.213 &  0.105 &  0.106 &  {---} &  0.052 \\
[3pt] & Coverage (95\%) \\
& \quad Intercept  &  0.946 &  0.839 &  0.939 &  0.910 &  0.943 &  0.923 &  {---} &  0.944 \\
& \quad $\beta_x$  &  0.955 &  0.800 &  0.954 &  0.885 &  0.948 &  0.931 &  {---} &  0.935 \\
& \quad $\beta_p$  &  0.945 &  0.898 &  0.951 &  0.925 &  0.940 &  0.918 &  {---} &  0.942 \\
& \quad $\sigma_x$ &  0.950 &  0.856 &  0.952 &  0.934 &  0.951 &  0.936 &  {---} &  0.950 \\
& \quad $\sigma_p$ &  0.973 &  0.970 &  0.954 &  0.953 &  0.946 &  0.937 &  {---} &  0.937 \\
\bottomrule
\end{tabular}
\begin{tablenotes}[flushleft]
\footnotesize
\item \textit{Notes.} The design uses $J = 4$, $\gamma_0 = 1$, $a = 0.667$,
SNFP \rev{reads each market once}, $T_0 = 1{,}000$ pilot, $1{,}000$ replications for each sample size and estimator.
$\theta^\ast = (-5,\, 1.5,\, -2,\, \sqrt{0.5},\, \sqrt{0.5})$. \rev{Where both estimators are reported, NFP and SNFP use
the same main-sample simulated data and integration draws within each replication. NFP is not reported at
$T = 128{,}000$ because it does not complete within the computation-time budget.}
\end{tablenotes}
\end{threeparttable}
\end{table}

\rev{From $T = 8{,}000$ onward, the statistical performance of SNFP is close to that of NFP. The biases
of both estimators are small relative to the RMSE, the RMSEs differ by at most $0.013$ at $T = 8{,}000$ and by at
most $0.005$ at $T = 32{,}000$, and the average confidence interval lengths agree to within $0.003$. As $T$
increases, both bias and RMSE decrease, and from $T = 8{,}000$ onward the RMSE of SNFP roughly halves with each
fourfold increase in $T$, the $T^{-1/2}$ rate predicted by the asymptotic theory in
Section~\ref{section:asymptotic_theory}. The empirical coverage of SNFP is below that of NFP for most parameters at
$T = 8{,}000$, where the lowest rate is $0.885$ for $\beta_x$ against $0.954$ for NFP, and the gap narrows as $T$
grows. At $T = 128{,}000$, where only SNFP is reported because of the computation-time budget, all five coverage
rates lie between $0.935$ and $0.950$.}

\rev{The exception is the smallest sample size, $T = 2{,}000$, where SNFP is noticeably
noisier than NFP for the linear coefficients. The RMSE of the intercept is $0.133$ against $0.089$, and the coverage
rates for the intercept, $\beta_x$ and $\sigma_x$ are $0.839$, $0.800$ and $0.856$ against $0.946$, $0.955$ and
$0.950$ for NFP. The under-coverage reflects bias rather than a mis-scaled variance estimate, since the SNFP
interval lengths match the NFP lengths to within $0.026$. The intervals have the right width but are not correctly
centered. This is the small-sample regime identified by the accuracy level of Table~\ref{tab:time_memory}. Together
with the computational results above, these findings support the scalability of SNFP for estimating BLP models with
substantially large sample sizes.}

\subsection{Simulation Error}\label{subsection:mc_simulation_error}

We now investigate the effect of simulation integration error arising from the random draws. In the simulations reported so far, the inner loop of both NFP and SNFP uses the same ($R=1{,}000$) draws as those used to generate the data. Thus, the simulated share function coincides with the data-generating share function, and simulation integration error is zero by construction. We now generate the data using effectively exact integration and estimate the model using independent random draws. For each market, the draws used in estimation are generated independently, with their number increasing with the sample size at the rate ($R(T)=\lceil 2T^{2/3}\rceil$), which satisfies the rate condition of Theorem~\ref{thm:negligible-sim}. All other aspects of the design, including the market-level data, the pilot sample, and the tuning parameters, are identical to those underlying Table~\ref{tab:stat_accuracy}\rev{, and SNFP again reads each market once}.

\rev{With fresh draws and $R$ growing faster than $\sqrt T$, the simulation error is negligible.} Table~\ref{tab:sim_error} reports the results, which are nearly indistinguishable from the SNFP columns of Table~\ref{tab:stat_accuracy}. \rev{From $T = 8{,}000$ onward, the RMSE and the confidence interval length differ from Table~\ref{tab:stat_accuracy} by at most $0.003$, the bias remains small relative to the RMSE, and the RMSE roughly halves with each fourfold increase in $T$. The coverage also follows Table~\ref{tab:stat_accuracy}. It is below the nominal level for some parameters at $T = 2{,}000$ and approaches $0.95$ as $T$ grows, reaching $0.94$--$0.96$ at $T = 128{,}000$. The only visible trace of the simulation error appears at $T = 2{,}000$, where $R = 318$ and the RMSE differs from Table~\ref{tab:stat_accuracy} by up to $9\%$.} As predicted by Theorem~\ref{thm:negligible-sim} and Corollary~\ref{thm:asymptotic_dist}, when $R$ grows faster than $\sqrt T$ the simulation error is asymptotically negligible and the plug-in sandwich variance remains valid\rev{, so the coverage approaches the nominal level as $T$ grows}.

% =============================================================================
% Table: SNFP with independent estimation draws, R(T) = ceil(2 T^(2/3)).
% Numbers: codes/2026_09_01_R_rate_T23_GH (Gauss-Hermite 40x40 exact-normal DGP,
% fresh per-market draws, ep1, 1000 reps per tier, sandwich CI). Same market
% primitives, T0, epochs, and step sizes as tab:stat_accuracy.
% =============================================================================
\begin{table}[!htbp]
\centering
\small
\caption{SNFP with independent simulation draws, $R(T) = \lceil 2T^{2/3} \rceil$.}
\label{tab:sim_error}
\begin{threeparttable}
\begin{tabular}{ll *{4}{S[table-format=-1.3]}}
\toprule
& & {$T = 2{,}000$} & {$T = 8{,}000$} & {$T = 32{,}000$} & {$T = 128{,}000$} \\
& & {$R = 318$} & {$R = 800$} & {$R = 2{,}016$} & {$R = 5{,}080$} \\
\midrule
& Bias \\
& \quad Intercept    &  0.058 &  0.015 &  0.007 &  0.001 \\
& \quad $\beta_x$    &  0.007 &  0.003 &  0.002 &  0.001 \\
& \quad $\beta_p$    & -0.086 & -0.029 & -0.014 & -0.004 \\
& \quad $\sigma_x$   &  0.006 &  0.003 &  0.001 &  0.000 \\
& \quad $\sigma_p$   &  0.030 & -0.008 & -0.005 & -0.004 \\
[3pt] & RMSE \\
& \quad Intercept    &  0.135 &  0.054 &  0.025 &  0.011 \\
& \quad $\beta_x$    &  0.040 &  0.016 &  0.007 &  0.003 \\
& \quad $\beta_p$    &  0.245 &  0.104 &  0.049 &  0.022 \\
& \quad $\sigma_x$   &  0.025 &  0.010 &  0.005 &  0.002 \\
& \quad $\sigma_p$   &  0.102 &  0.057 &  0.029 &  0.014 \\
[3pt] & CI length \\
& \quad Intercept    &  0.358 &  0.174 &  0.086 &  0.043 \\
& \quad $\beta_x$    &  0.100 &  0.049 &  0.024 &  0.012 \\
& \quad $\beta_p$    &  0.716 &  0.340 &  0.168 &  0.084 \\
& \quad $\sigma_x$   &  0.075 &  0.037 &  0.018 &  0.009 \\
& \quad $\sigma_p$   &  0.437 &  0.212 &  0.104 &  0.051 \\
[3pt] & Coverage (95\%) \\
& \quad Intercept    &  0.856 &  0.896 &  0.926 &  0.943 \\
& \quad $\beta_x$    &  0.819 &  0.886 &  0.927 &  0.936 \\
& \quad $\beta_p$    &  0.911 &  0.908 &  0.922 &  0.943 \\
& \quad $\sigma_x$   &  0.873 &  0.936 &  0.934 &  0.955 \\
& \quad $\sigma_p$   &  0.967 &  0.940 &  0.934 &  0.943 \\
\bottomrule
\end{tabular}
\begin{tablenotes}[flushleft]
\footnotesize
\item \textit{Notes.} The design uses $J = 4$, $\gamma_0 = 1$, $a = 0.667$,
\rev{SNFP reads each market once}, $T_0 = 1{,}000$ pilot, $1{,}000$ replications.
$\theta^\ast = (-5,\, 1.5,\, -2,\, \sqrt{0.5},\, \sqrt{0.5})$. Market shares in the data are integrated exactly over the
normal distribution of the random coefficients.
\end{tablenotes}
\end{threeparttable}
\end{table}

The experiment above draws a fresh block for each market, as Assumption~\ref{ass:sampling} requires. \rev{An alternative design}
shares one block across all markets, and Theorem~\ref{thm:common-draws} suggests that this may come at a cost. Our last
experiment examines this cost. We rerun the design of Table~\ref{tab:sim_error} at the sample sizes used throughout the
paper, $T \in \{2{,}000,\,8{,}000,\,32{,}000\}$, with $R$ a power of two between $128$ and $8{,}192$, so that a
scrambled Sobol' block is a complete net and $c = T/R$ ranges from $0.98$ to $62.5$. We compare three draw schemes on
the \emph{same} simulated data within each replication. Under \emph{per-market draws}, a fresh block is
drawn for each market. Under \emph{shared i.i.d.\ draws}, one i.i.d.\ block is used by every market and
every online update. Under \emph{shared Sobol' draws}, one block of scrambled Sobol' points is used in the
same way. The three schemes also share the same starting value within a replication. Ten $(T,R)$ combinations are
run, with $1{,}000$ replications each. Tables~\ref{tab:common_draws} and~\ref{tab:common_draws_cov} report the RMSE
and the coverage at the largest sample, where the effect in Theorem~\ref{thm:common-draws} is expected to be
largest. The results for $T = 2{,}000$ and $T = 8{,}000$ are reported in
Tables~\ref{tab:common_draws_T2000}--\ref{tab:common_draws_T8000_cov} of Appendix~\ref{appendix:common_draws_mc}.

All three schemes are consistent, as Theorem~\ref{thm:common-draws}(i) asserts, but they differ in precision
and inference. With per-market draws, the RMSE stays close to the full-sample GMM benchmark
$\sqrt{\Sigma_{\theta,jj}/T}$ at every $c$ (Table~\ref{tab:common_draws}) and coverage stays \rev{between $0.90$ and $0.95$}
(Table~\ref{tab:common_draws_cov}). With shared i.i.d.\ draws, the RMSE
grows with $c$ and the coverage of the plug-in sandwich interval falls well below the nominal level, most visibly for
$\beta_x$ and $\sigma_x$. Moreover, increasing $T$ with $R$ held fixed lowers the RMSE of $\beta_x$ substantially
under per-market draws but only slightly under shared i.i.d.\ draws. These patterns are consistent
with Theorem~\ref{thm:common-draws}(ii), which predicts a variance inflation that is linear in $c$, and with
Theorem~\ref{thm:common-draws}(iii), which predicts the loss of the $\sqrt T$ rate.

Two remedies appear to help. Widening the sandwich variance by the estimated inflation term
\eqref{eq:widened-variance} brings coverage back to the level obtained with per-market draws\rev{, $0.85$--$0.97$ against $0.80$--$0.97$ over all ten combinations}
(Table~\ref{tab:common_draws_cov}). Shared Sobol'
draws remain close to per-market draws at moderate $c$ and deteriorate far more slowly than shared i.i.d.\ draws as
$c$ grows, in line with Corollary~\ref{cor:qmc}. \rev{At $T = 32{,}000$ and $c \ge 31.3$, $\sqrt T/R$ is
$0.17$--$0.35$ and shared Sobol' coverage falls to $0.74$--$0.93$, against $0.90$--$0.95$ for per-market draws.
Corollary~\ref{cor:qmc} is asymptotic and makes no prediction at these $(T,R)$. The pattern shows that the
displacement of a shared low-discrepancy block can remain material at moderate $R$, in line with the bound
$\sqrt T\,\epsilon_R = O\bigl(c\,(\log R)^{d_v}/\sqrt T\bigr)$.} In practice, a shared block is therefore better drawn from a
low-discrepancy sequence, and a common i.i.d.\ block should be used with a small $c = T/R$ and corrected standard
errors.

\rev{At a fixed ratio $c = T/R$, the distortion from a shared i.i.d.\ block persists as $T$ grows, whereas that from a shared Sobol' block vanishes. Theorem~\ref{thm:common-draws} and Corollary~\ref{cor:qmc} predict this difference.}
Under a shared i.i.d.\ block, the inflation term $c\,\Lambda_\nu$ in Theorem~\ref{thm:common-draws}(ii) depends on
$c$ alone, so at fixed $c$ the \rev{relative loss in precision (the RMSE ratio to per-market draws)} and \rev{the} coverage should not change with $T$. Under a shared
low-discrepancy block, Corollary~\ref{cor:qmc} bounds the displacement by \rev{a multiple of $\epsilon_R$, and
$\epsilon_R = O\bigl(R^{-1}(\log R)^{d_v}\bigr)$ with $d_v = 2$ in this design,
so $\sqrt T\,\epsilon_R = O\bigl(c\,(\log R)^{d_v}/\sqrt T\bigr)$} vanishes as $T$ grows with $c$
fixed. Table~\ref{tab:common_draws_fixedc} collects the cells with $c = 15.6$ from
Tables~\ref{tab:common_draws}--\ref{tab:common_draws_cov} and
Tables~\ref{tab:common_draws_T2000}--\ref{tab:common_draws_T8000_cov}. \rev{With shared i.i.d.\ draws, the coverage of $\beta_x$ stays at $0.42$--$0.43$ and that of $\sigma_x$ at $0.44$--$0.48$ as $T$ rises from $2{,}000$ to $32{,}000$. The RMSE relative to per-market draws rises from $2.6$ to $3.5$ for $\beta_x$ and from $2.8$ to $3.2$ for $\sigma_x$, because with a single pass the per-market RMSE at $T = 2{,}000$ still carries the small-sample error of SNFP seen in Table~\ref{tab:stat_accuracy}. With shared Sobol' draws, the coverage of $\beta_x$ rises from $0.76$ to $0.92$ and that of $\sigma_x$ from $0.73$ to $0.90$ over the same range, and the RMSE ratio falls from $1.2$--$1.3$ to $1.1$. The constant coverage under shared i.i.d.\ draws is consistent with the constant-in-$T$ inflation of Theorem~\ref{thm:common-draws}(ii), and the rising coverage under shared Sobol' draws with the vanishing displacement of Corollary~\ref{cor:qmc}.}

In sum, the Monte Carlo evidence highlights two points. First, SNFP offers substantial savings in computation time
and memory relative to NFP and remains feasible at sample sizes where NFP does not finish. Second, the results are
consistent with the asymptotic theory of Section~\ref{section:asymptotic_theory}. SNFP performs comparably to NFP in
bias, RMSE, and coverage, and simulation error has little effect with a fresh block of draws for each market. With a
common block, the RMSE grows and the standard intervals undercover, in line with Theorem~\ref{thm:common-draws}.

% =============================================================================
% Table: common-draw study, T = 32,000 slice of
% simulations_in_the_paper/tab4_common_draws_v2 (table_common_draws_main.tex and
% the feasible \hat\Lambda_\nu block of table_common_draws_inference.tex).
% 1,000 replications per cell; epoch-1 snapshots. Numbers hard-coded on purpose:
% the generated 10-cell fragment is too wide at \small.
% =============================================================================
\begin{table}[!htbp]
\centering
\small
\caption{RMSE under the three draw schemes at $T = 32{,}000$.}
\label{tab:common_draws}
\begin{threeparttable}
\begin{tabular}{ll *{4}{S[table-format=1.3]}}
\toprule
& & {$c = 7.8$} & {$c = 15.6$} & {$c = 31.3$} & {$c = 62.5$} \\
& & {$R = 4{,}096$} & {$R = 2{,}048$} & {$R = 1{,}024$} & {$R = 512$} \\
\midrule
\multicolumn{6}{l}{\emph{Per-market draws}} \\
& \quad Intercept   & 0.023 & 0.023 & 0.024 & 0.024 \\
& \quad $\beta_x$   & 0.007 & 0.007 & 0.007 & 0.007 \\
& \quad $\beta_p$   & 0.046 & 0.046 & 0.046 & 0.048 \\
& \quad $\sigma_x$  & 0.005 & 0.005 & 0.005 & 0.005 \\
& \quad $\sigma_p$  & 0.027 & 0.027 & 0.027 & 0.027 \\
[3pt] \multicolumn{6}{l}{\emph{Shared i.i.d.\ draws}} \\
& \quad Intercept   & 0.026 & 0.029 & 0.034 & 0.042 \\
& \quad $\beta_x$   & 0.017 & 0.023 & 0.032 & 0.047 \\
& \quad $\beta_p$   & 0.059 & 0.067 & 0.084 & 0.110 \\
& \quad $\sigma_x$  & 0.012 & 0.016 & 0.022 & 0.031 \\
& \quad $\sigma_p$  & 0.039 & 0.046 & 0.059 & 0.081 \\
[3pt] \multicolumn{6}{l}{\emph{Shared Sobol' draws}} \\
& \quad Intercept   & 0.023 & 0.024 & 0.024 & 0.026 \\
& \quad $\beta_x$   & 0.007 & 0.007 & 0.008 & 0.010 \\
& \quad $\beta_p$   & 0.046 & 0.047 & 0.050 & 0.057 \\
& \quad $\sigma_x$  & 0.005 & 0.006 & 0.006 & 0.009 \\
& \quad $\sigma_p$  & 0.028 & 0.029 & 0.032 & 0.039 \\
\bottomrule
\end{tabular}
\begin{tablenotes}[flushleft]
\footnotesize
\item \textit{Notes.} $c = T/R$. The design is as in Table~\ref{tab:sim_error}
with $1{,}000$ replications. \rev{The Sobol' block is Owen-scrambled independently in each replication and shared by
the pilot and all online markets within the replication. Reported statistics average over data and scrambles.}
\end{tablenotes}
\end{threeparttable}
\end{table}

\begin{table}[!htbp]
\centering
\small
\caption{Coverage of $95\%$ confidence intervals under the three draw schemes at $T = 32{,}000$.}
\label{tab:common_draws_cov}
\begin{threeparttable}
\begin{tabular}{ll *{4}{S[table-format=1.3]}}
\toprule
& & {$c = 7.8$} & {$c = 15.6$} & {$c = 31.3$} & {$c = 62.5$} \\
& & {$R = 4{,}096$} & {$R = 2{,}048$} & {$R = 1{,}024$} & {$R = 512$} \\
\midrule
\multicolumn{6}{l}{\emph{Per-market draws}} \\
& \quad Intercept   & 0.934 & 0.942 & 0.936 & 0.933 \\
& \quad $\beta_x$   & 0.924 & 0.932 & 0.931 & 0.931 \\
& \quad $\beta_p$   & 0.932 & 0.938 & 0.935 & 0.924 \\
& \quad $\sigma_x$  & 0.935 & 0.937 & 0.920 & 0.898 \\
& \quad $\sigma_p$  & 0.949 & 0.946 & 0.946 & 0.950 \\
[3pt] \multicolumn{6}{l}{\emph{Shared i.i.d.\ draws}} \\
& \quad Intercept   & 0.901 & 0.888 & 0.805 & 0.720 \\
& \quad $\beta_x$   & 0.556 & 0.417 & 0.301 & 0.217 \\
& \quad $\beta_p$   & 0.859 & 0.824 & 0.700 & 0.595 \\
& \quad $\sigma_x$  & 0.548 & 0.442 & 0.334 & 0.236 \\
& \quad $\sigma_p$  & 0.822 & 0.747 & 0.626 & 0.494 \\
[3pt] \multicolumn{6}{l}{\emph{Shared i.i.d.\ draws, widened interval}} \\
& \quad Intercept   & 0.944 & 0.942 & 0.941 & 0.919 \\
& \quad $\beta_x$   & 0.954 & 0.943 & 0.944 & 0.936 \\
& \quad $\beta_p$   & 0.941 & 0.949 & 0.941 & 0.917 \\
& \quad $\sigma_x$  & 0.939 & 0.943 & 0.949 & 0.931 \\
& \quad $\sigma_p$  & 0.935 & 0.932 & 0.928 & 0.902 \\
[3pt] \multicolumn{6}{l}{\emph{Shared Sobol' draws}} \\
& \quad Intercept   & 0.936 & 0.933 & 0.923 & 0.909 \\
& \quad $\beta_x$   & 0.927 & 0.915 & 0.876 & 0.812 \\
& \quad $\beta_p$   & 0.934 & 0.927 & 0.927 & 0.861 \\
& \quad $\sigma_x$  & 0.923 & 0.901 & 0.843 & 0.738 \\
& \quad $\sigma_p$  & 0.938 & 0.934 & 0.894 & 0.836 \\
\bottomrule
\end{tabular}
\begin{tablenotes}[flushleft]
\footnotesize
\item \textit{Notes.} Coverage uses the plug-in sandwich variance, which is
widened by the estimated inflation term \eqref{eq:widened-variance} in the widened-interval panel. See
Table~\ref{tab:common_draws}.
\end{tablenotes}
\end{threeparttable}
\end{table}

\begin{table}[!htbp]
\centering
\small
\caption{Shared draws at a fixed ratio $c = T/R = 15.6$: RMSE and coverage of $95\%$ confidence intervals.}
\label{tab:common_draws_fixedc}
\begin{threeparttable}
\begin{tabular}{ll *{3}{S[table-format=1.3]} *{3}{S[table-format=1.3]}}
\toprule
& & \multicolumn{3}{c}{RMSE} & \multicolumn{3}{c}{Coverage} \\
\cmidrule(lr){3-5} \cmidrule(lr){6-8}
& & {$T = 2{,}000$} & {$T = 8{,}000$} & {$T = 32{,}000$} & {$T = 2{,}000$} & {$T = 8{,}000$} & {$T = 32{,}000$} \\
& & {$R = 128$} & {$R = 512$} & {$R = 2{,}048$} & {$R = 128$} & {$R = 512$} & {$R = 2{,}048$} \\
\midrule
\multicolumn{8}{l}{\emph{Per-market draws}} \\
& \quad $\beta_x$   & 0.042 & 0.016 & 0.007 & 0.802 & 0.895 & 0.932 \\
& \quad $\sigma_x$  & 0.026 & 0.011 & 0.005 & 0.860 & 0.926 & 0.937 \\
[3pt] \multicolumn{8}{l}{\emph{Shared i.i.d.\ draws}} \\
& \quad $\beta_x$   & 0.107 & 0.045 & 0.023 & 0.428 & 0.420 & 0.417 \\
& \quad $\sigma_x$  & 0.072 & 0.032 & 0.016 & 0.465 & 0.475 & 0.442 \\
[3pt] \multicolumn{8}{l}{\emph{Shared Sobol' draws}} \\
& \quad $\beta_x$   & 0.051 & 0.018 & 0.007 & 0.756 & 0.839 & 0.915 \\
& \quad $\sigma_x$  & 0.034 & 0.013 & 0.006 & 0.733 & 0.848 & 0.901 \\
\bottomrule
\end{tabular}
\begin{tablenotes}[flushleft]
\footnotesize
\item \textit{Notes.} Entries are taken from Tables~\ref{tab:common_draws}, \ref{tab:common_draws_cov},
\ref{tab:common_draws_T2000}, \ref{tab:common_draws_T2000_cov}, \ref{tab:common_draws_T8000}, and
\ref{tab:common_draws_T8000_cov}. Coverage uses the plug-in sandwich variance without correction. The design is as
in Table~\ref{tab:sim_error} with $1{,}000$ replications.
\end{tablenotes}
\end{threeparttable}
\end{table}

\FloatBarrier
\section{Empirical Application: Yogurt Demand}\label{section:empirical}

We apply SNFP to yogurt demand using store-level scanner data from the Information Resources Inc.\ (IRI) dataset \citep{bronnenberg2008database} for 2001--2012. A market is a store--week pair and a product is a brand. We keep the top five brands in each store--week. \rev{The instruments are a constant and a cubic polynomial in the product's own lagged price. Table~\ref{tab:product_summary} reports summary statistics for the full sample with  $T=817{,}031$ markets.}

\begin{table}[!b]
\centering
\caption{Summary Statistics for the Full Sample}
\label{tab:product_summary}
\begin{threeparttable}
\begin{tabular}{lcccc}
\toprule
 & Mean & SD & Min & Max \\
\midrule
Price (2012 dollars per 6 oz) & 0.746 & 0.265 & 0.076 & 5.604 \\
Product share (\%)            & 0.448 & 0.958 & 0.000 & 87.6 \\
Outside share (\%)            & 97.8  & 3.84  & 0.05  & 100.0 \\
Products per market, $J_t$    & 4.94  & 0.26  & 1     & 5 \\
\midrule
Markets, $T$ & \multicolumn{4}{c}{817,031} \\
Weeks & \multicolumn{4}{c}{625} \\
Brands & \multicolumn{4}{c}{253} \\
\bottomrule
\end{tabular}
\begin{tablenotes}
\footnotesize
\item \textit{Notes:} A market is a store--week, and a product is one of its top five brands. The lagged price, which is the instrument, has nearly the same distribution as the price.
\end{tablenotes}
\end{threeparttable}
\end{table}

\begin{table}[!tb]
\centering
\caption{Full-Sample Demand Estimates for the Yogurt Market}
\label{tab:yogurt_estimates}
\begin{threeparttable}
\begin{tabular}{lccc}
\toprule
&
\multicolumn{2}{c}{Standard Logit}
&
\multicolumn{1}{c}{Random-Coefficient Logit} \\
\cmidrule(lr){2-3} \cmidrule(lr){4-4}
&
\multicolumn{1}{c}{OLS}
&
\multicolumn{1}{c}{IV}
&
\multicolumn{1}{c}{SNFP} \\
\midrule
Intercept
    & \rev{-4.871} & \rev{-4.968} & \rev{-2.902} \\
    & \rev{(0.002)} & \rev{(0.002)} & \rev{(0.032)} \\

Price
    & -1.793 & \rev{-1.663} & \rev{-7.145} \\
    & \rev{(0.003)} & (0.003) & \rev{(0.101)} \\

$\sigma_{\text{price}}$
    & --- & --- & \rev{2.767} \\
    & --- & --- & \rev{(0.035)} \\

\midrule
Own-price elasticity
    & \rev{-1.332} & \rev{-1.236} & \rev{-1.869}\\
    & \rev{(0.475)} & \rev{(0.441)} & \rev{(0.292)}\\
Cross-price elasticity
    & 0.005 & 0.005 & \rev{0.020}\\
    & (0.011) & (0.010) & \rev{(0.042)}\\

\midrule
Observations (Markets)
    & \rev{817,031} & 817,031 & 817,031 \\
\bottomrule
\end{tabular}

\begin{tablenotes}
\footnotesize
\item \textit{Notes:} Standard errors are in parentheses below the coefficients\rev{, heteroskedasticity-robust for OLS and IV and plug-in sandwich for SNFP}. For the elasticities, the table reports means across observations with standard deviations in parentheses.
\end{tablenotes}
\end{threeparttable}
\end{table}

The model has a random coefficient on price and a fixed intercept. Price is measured in 2012 dollars per 6-oz equivalent.\footnote{Market size is the population within a 2-mile radius of the store, scaled by seven, assuming each person consumes one 6-oz serving per day.} \rev{We estimate the model by SNFP, as set out in Algorithm~\ref{alg:snfp-skeleton}. A pilot of $T_0 = 5{,}000$ randomly drawn markets supplies the starting value and the scaling matrix $\hat A_{T_0}$, which is held fixed afterwards.\footnote{The pilot has to be drawn at random. The data are ordered by store and week, and a pilot of the first $1{,}000$ markets does not pin down $\sigma_{\text{price}}$ well enough to fix the scaling matrix.} The online pass visits the remaining $812{,}031$ markets once, in random order. Simulated shares use $R = 1{,}024$ fresh draws per market, as in Assumption~\ref{ass:sampling}. The standard errors are the plug-in sandwich standard errors of Corollary~\ref{thm:asymptotic_dist}.} \rev{For comparison, we also estimate the standard logit model without instruments (OLS) and with the same instruments (IV) as SNFP.}

Table~\ref{tab:yogurt_estimates} reports the estimates. \rev{The random-coefficient model and the standard logit model give very different pictures of demand. In the logit model every consumer has the same price coefficient, $-1.793$ by OLS and $-1.663$ by IV. SNFP estimates a mean price coefficient of $-7.145$ and a standard deviation of $\sigma_{\text{price}} = 2.767$ across consumers, so price sensitivity is strongly heterogeneous. The coefficients of the two models are not directly comparable, but their elasticities are. The mean own-price elasticity is $-1.869$ under SNFP, about $50\%$ larger in magnitude than the IV logit's $-1.236$.\footnote{At the estimates, $0.49\%$ of simulated consumers have a positive price coefficient. We cap individual price coefficients at $-0.001$ when computing elasticities. Without the cap the mean own-price elasticity is $-1.523$.} The two models also differ in how elasticities vary across brands. In the logit model the own-price elasticity is proportional to price, so the most expensive brands are the most elastic. In the random-coefficient model expensive brands attract the less price-sensitive consumers, and their demand is less elastic than that of cheap brands. Cross-price elasticities are small in both models, because the outside option accounts for about $98\%$ of each market. They are four times larger under SNFP, $0.020$ against $0.005$. In the logit model substitution depends only on market shares, whereas in the random-coefficient model consumers substitute more toward brands with similar prices.}

\rev{Estimating the model on a subsample, a common way to reduce the computational burden of NFP, gives misleading estimates. Figure~\ref{fig:subsample_compare} shows this by comparing the full-sample SNFP estimate with SNFP estimates from $1{,}000$ random subsamples at each of three sizes: 1\%, 5\%, and 25\% of the markets. Every replication uses the same pilot, scaling matrix, and per-market simulation draws as the full-sample run and makes one pass. Only the sampled markets and the order in which they are visited differ.}

\begin{figure}[!t]
    \centering
    \caption{\rev{SNFP} Estimates from Repeated Subsamples}
    \label{fig:subsample_compare}
    \vspace{2pt}
    \includegraphics[width=0.95\textwidth]{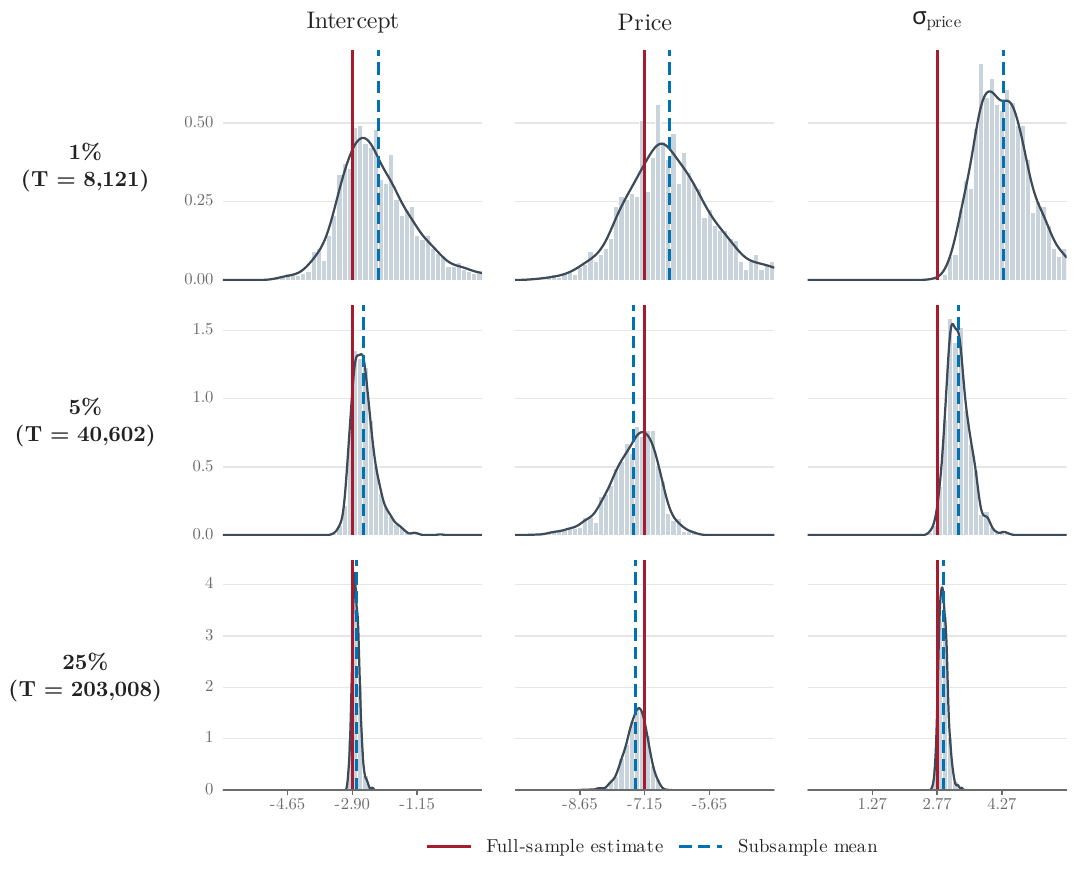}

    \vspace{2pt}
    \begin{minipage}{0.95\textwidth}
    \footnotesize \rev{\textit{Notes:} Red solid line: full-sample estimate. Blue dashed line: mean of the subsample estimates.}
    \end{minipage}
\end{figure}

\rev{The subsample estimates are biased relative to the full-sample estimate. The full-sample estimate of $\sigma_{\text{price}}$ is $2.767$. The subsample estimates average $4.294$ at 1\%, $3.261$ at 5\%, and $2.910$ at 25\%. At 1\% and 5\%, almost the entire distribution lies above the full-sample value. The bias shrinks as the subsample grows, but it is still visible at 25\%. The subsample estimates are also dispersed. The standard deviation of the $\sigma_{\text{price}}$ estimates is $0.618$ at 1\%, $0.270$ at 5\%, and $0.100$ at 25\%. Together, the bias and the dispersion make a small subsample unreliable. At 1\%, the middle $90\%$ of the $\sigma_{\text{price}}$ estimates ranges from $3.37$ to $5.40$, and only $0.4\%$ of the replications lie within $10\%$ of the full-sample value. This share rises to $20\%$ at 5\% and to $91\%$ at 25\%. The root mean squared error of $\sigma_{\text{price}}$ falls from $1.65$ at 1\% to $0.56$ at 5\% and $0.17$ at 25\%. It falls faster than the square root of the number of markets, because the bias shrinks along with the dispersion. The real data thus show the pattern of the simulation in Figure~\ref{fig:small_sample_blp}. With few markets the estimates are dispersed and displaced, and they concentrate around the full-sample value only as the number of markets grows. Using every market pays off in both precision and accuracy, and SNFP makes it feasible.}

\FloatBarrier

\section{Conclusion}\label{section:conclusion}

This paper develops the stochastic nested fixed point (SNFP) algorithm, an online stochastic-GMM estimator for the
random coefficients logit demand model. SNFP embeds the BLP demand inversion inside a stochastic gradient update and
processes one market at a time, so the memory required per step is $O(J \cdot R)$, independent of the number of markets
$T$. The method therefore scales to the disaggregated regime with $T > 10^6$ markets on standard hardware, \rev{where
full-sample nested fixed point (NFP) estimation becomes costly in both time and memory}.

We show that the Polyak--Ruppert average of the SNFP iterates is $\sqrt T$-consistent and asymptotically normal with the
classical GMM sandwich covariance, attaining the same asymptotic efficiency as the offline GMM estimator under the same
weighting matrix. \rev{The theory also covers implementations that share one block of draws across markets. Sharing
an i.i.d.\ block adds a variance term proportional to $T/R$ that the plug-in standard errors omit. A shared block of
quasi-Monte Carlo draws, such as a scrambled Sobol' or Halton sequence, preserves the $\sqrt T$ rate and the usual
inference, provided its integration error vanishes faster than $T^{-1/2}$.} The Monte Carlo
experiments are equally satisfactory. \rev{With a single pass over the data,} SNFP reproduces the bias \rev{and} RMSE of NFP \rev{at all but the
smallest sample size, its coverage approaches the nominal level as $T$ grows},
while \rev{it cuts} wall-clock time by one to two orders of magnitude and \rev{remains} well calibrated at sample sizes that NFP
\rev{does not reach within our time budget}. This paper suggests several extensions. The most natural is to dynamic discrete choice models with large
state spaces, where the inner problem is a Bellman fixed point rather than a demand inversion and the repeated solution
of the dynamic program over the full state space is the binding computational constraint. An online estimator
that
updates the structural parameters observation by observation while solving the dynamic program only locally could bring
estimation within reach at scales that are currently prohibitive. A complementary direction is a stochastic MPEC variant
that folds the inner fixed point into a constraint within the online update.

\newpage

\bibliography{SGD_BLP}

\clearpage

%---------------------------------------------------------------------------------------------------------------
%---------------------------------------------------------------------------------------------------------------
%
% Appendix 
%
%---------------------------------------------------------------------------------------------------------------
%---------------------------------------------------------------------------------------------------------------

\appendix

\begin{center}
{\Large\bfseries Appendix}
\end{center}

\section{Notation and Preliminary Lemmas}\label{appendix:lemmas}

This appendix collects the notation and the preliminary lemmas. These cover regularity of the BLP moment and
moment bounds on the simulation error, both implied by Assumption~\ref{ass:regular-markets}
(Lemmas~\ref{lem:blp-regularity}--\ref{lem:sim-error}), a deterministic recursion bound
(Lemma~\ref{lem:recursion-bound}), $L^2$ rates and a containment result for the iterates (Lemma~\ref{lem:L2_rate}),
and the asymptotic inactivity of the projection (Lemma~\ref{lem:projection}). The theorems are proved in
Appendix~\ref{appendix:thm1_proof}.

Throughout the appendix, $C < \infty$ and $c > 0$ denote generic constants that do not depend on $t$, $T$, $T_0$,
or $R$ and may change from line to line. Recall $\hat A := \hat A_{T_0}$ and write $\Delta A := \hat A - A^\ast$.
The filtration is
\[
  \mathcal F_{T_0} := \sigma\bigl(O_1,\ldots,O_{T_0},\{\nu_r^{\mathrm{pilot}}\}_{r=1}^R\bigr), \qquad
  \mathcal F_t := \mathcal F_{T_0} \vee \sigma\bigl(O_s,\{\nu_{r,s}\}_{r=1}^R : T_0 < s \le t\bigr), \quad t > T_0,
\]
so that $\hat\theta_{T_0}$ and $\hat A$ are $\mathcal F_{T_0}$-measurable and, for $t > T_0$, the pair
$(O_t, \{\nu_{r,t}\}_r)$ is independent of $\mathcal F_{t-1}$ by Assumption~\ref{ass:sampling}. Define
\[
  g_t^R(\theta) := g(O_t;\theta,P_t^R), \qquad g_t(\theta) := g(O_t;\theta,P^\ast), \qquad
  g_t^\circ := g_t(\theta^\ast),
\]
\[
  \bar m(\theta) := E[g(O;\theta,P^\ast)], \qquad \bar m_R(\theta) := E[g(O;\theta,P^R)], \qquad
  \beta_R(\theta) := \bar m_R(\theta) - \bar m(\theta),
\]
\[
  u_t^R(\theta) := g_t^R(\theta) - g_t(\theta), \qquad \zeta_t(\theta) := g_t(\theta) - \bar m(\theta), \qquad
  e_t := \theta_t - \theta_t^o,
\]
and $V_t := \|\theta_t - \theta^\ast\|^2$, $V_t^o := \|\theta_t^o - \theta^\ast\|^2$,
$\Gamma_t := \sum_{s=T_0+1}^t \gamma_s$. By Lemma~\ref{lem:blp-regularity} below, $\bar m$ is
continuously differentiable with Jacobian $G(\theta) := G(\theta,P^\ast)$ satisfying
$\|G(\theta) - G(\theta')\| \le \bar L\|\theta - \theta'\|$ for some $\bar L < \infty$. Since
$A^\ast G^\ast = I_d$ and $G$ is Lipschitz, we may fix $r_0 > 0$ such that the closed ball
$N_0 := \{\theta : \|\theta - \theta^\ast\| \le r_0\} \subseteq \Theta$ and the symmetric part of $A^\ast G(\theta)$
is bounded below by $\tfrac12 I_d$ on $N_0$. Finally, fix a small constant $r_1 > 0$, whose value is determined in
the proofs, and define the $\mathcal F_{T_0}$-measurable event
\[
  E_0 := \bigl\{\|\Delta A\| \le r_1\bigr\} \cap \bigl\{\|\hat\theta_{T_0} - \theta^\ast\| \le r_0/2 \bigr\};
\]
$\mathbb P(E_0) \to 1$ by Assumption~\ref{ass:pilot} together with $T_0 \to \infty$ and $R \to \infty$
(Assumption~\ref{ass:rates}).

The first lemma derives the regularity of the BLP moment function from the regular-markets condition. It replaces
high-level envelope, Lipschitz, and smoothness assumptions on $g$. All such properties are consequences of the
structure of the mixed-logit model.

\begin{lemma}[Regularity of the BLP moment]\label{lem:blp-regularity}
Under Assumptions~\ref{ass:sampling}, \ref{ass:param-space}, and \ref{ass:regular-markets}, there exists $C < \infty$ such that, uniformly in
$\theta \in \Theta$ and $t$:
\begin{enumerate}
\item[(i)] almost surely, $\|g(O_t;\theta,P)\| \le C$ and $\|\nabla_\theta g(O_t;\theta,P)\| \le C$, and both maps
are $C$-Lipschitz in $\theta$ on $\Theta$, where $P$ denotes either $P^\ast$ or any realization of $P_t^R$.
\item[(ii)] $\bar m$ is continuously differentiable on $\Theta$ with $C$-Lipschitz Jacobian $G(\cdot)$, and
$\Omega(\theta) := \mathrm{Var}[g(O;\theta,P^\ast)]$ is continuous on $\Theta$.
\end{enumerate}
\end{lemma}

\begin{proof}
Write $\sigma_j(\delta,\theta_2,\nu)$ for the inner logit choice probability, so that the (simulated) share map is
$\tilde s_j(\delta,\theta_2;P) = \int \sigma_j(\delta,\theta_2,\nu)\,dP(\nu)$, linear in $P$. Each $\sigma_j$ takes
values in $(0,1)$ and is smooth in $(\delta,\theta_2)$ with derivatives bounded by polynomials in $\|\nu\|$ with
coefficients controlled by $C_0$ and the compactness of $\Theta$ (Assumptions~\ref{ass:param-space} and
\ref{ass:regular-markets}). Since the mixing distribution is Gaussian, all derivative envelopes are integrable.
At the solution
$\tilde\delta_t = \tilde\delta(S_t,\theta_2;P)$ of $S_t = \tilde s(\delta,\theta_2;P)$, which exists and is unique
by \citet{berry1994estimating} and \citet{berry2013connected}, the fitted shares equal the observed shares, so
$\tilde s_j(\tilde\delta_t) = S_{jt} \ge \underline s$ for $j = 0, 1, \ldots, J_t$ by
Assumption~\ref{ass:regular-markets}(iii). The share Jacobian $\partial \tilde s/\partial\delta'$ is then invertible at
$\tilde\delta_t$ with inverse bounded by a constant depending only on $(\bar J, C_0, \underline s)$, by the strict
diagonal dominance of the mixed-logit share Jacobian when all fitted shares are bounded away from zero. The
implicit function theorem therefore yields that $\theta_2 \mapsto \tilde\delta(S_t,\theta_2;P)$ is bounded and
Lipschitz, with bounded and Lipschitz derivative, uniformly over $\theta \in \Theta$, over markets satisfying
Assumption~\ref{ass:regular-markets}, and over simulation measures $P$. The detailed verification follows the
arguments of
\citet{berry2004limit} and \citet{freyberger2015asymptotic} and is omitted. Part (i) follows since
$g = Z_t'(\tilde\delta_t - X_t\theta_1)$ with $\|Z_t\|, \|X_t\|$ bounded. Part (ii) follows from (i) by dominated
convergence.
\end{proof}

The second lemma records the moment bounds on the simulation error.

\begin{lemma}[Simulation-error moments]\label{lem:sim-error}
Under Assumptions~\ref{ass:sampling}, \ref{ass:param-space}, and \ref{ass:regular-markets}, uniformly in $\theta, \theta' \in \Theta$ and
$t > T_0$:
\begin{enumerate}
\item[(i)] $\sup_\theta E\|g_t^R(\theta)\|^2 \le C$ and $\sup_\theta E\|g_t(\theta)\|^2 \le C$.
\item[(ii)] $\|\beta_R(\theta)\| \le C R^{-1}$, and $\theta \mapsto \beta_R(\theta)$ is Lipschitz with constant
$C R^{-1}$.
\item[(iii)] $E\|u_t^R(\theta)\|^2 \le C R^{-1}$.
\item[(iv)] $E\|u_t^R(\theta) - u_t^R(\theta')\|^2 \le C R^{-1}\|\theta - \theta'\|^2$.
\end{enumerate}
Moreover, for any $\mathcal F_{t-1}$-measurable $\tilde\theta$, $E[u_t^R(\tilde\theta) \mid \mathcal F_{t-1}]
= \beta_R(\tilde\theta)$ and the bounds in (i), (iii), and (iv) hold for the conditional moments given
$\mathcal F_{t-1}$ with $\theta$ replaced by $\tilde\theta$.
\end{lemma}

\begin{proof}
Part (i) is immediate from Lemma~\ref{lem:blp-regularity}(i).
For the remaining parts, fix $\theta$ and condition on $O_t$. The simulated share map is linear in $P_t^R$ with
integrands bounded as in the proof of Lemma~\ref{lem:blp-regularity}, and the moment
$g(O_t;\theta,P) = Z_t'(\tilde\delta(S_t,\theta_2;P) - X_t\theta_1)$ is a twice continuously differentiable
function of the share values with derivatives bounded uniformly by Lemma~\ref{lem:blp-regularity}. A second-order
Taylor expansion of $\tilde\delta(S_t,\theta_2;\cdot)$ in the share-approximation error
$\tilde s(\cdot;P_t^R) - \tilde s(\cdot;P^\ast)$ therefore gives
\[
  u_t^R(\theta) = \mathcal L_t(\theta)\bigl[\tilde s(\cdot;P_t^R) - \tilde s(\cdot;P^\ast)\bigr] + \rho_t(\theta),
  \qquad \|\rho_t(\theta)\| \le C\,\bigl\|\tilde s(\cdot;P_t^R) - \tilde s(\cdot;P^\ast)\bigr\|^2,
\]
where
$\mathcal L_t(\theta) :=
-Z_t'\bigl[\partial \tilde s(\tilde\delta(S_t,\theta_2;P^\ast),\theta_2;P^\ast)/\partial\delta'\bigr]^{-1}$,
the derivative of the inverse share map with respect to the shares, is a bounded linear map that is Lipschitz in
$\theta$, both uniformly.
Conditionally on $O_t$, the share error is an average of $R$ i.i.d.\ mean-zero terms bounded by one, so the linear
term has conditional mean zero and conditional variance $O(R^{-1})$, while
$E[\|\rho_t(\theta)\|\mid O_t] = O(R^{-1})$ and $E[\|\rho_t(\theta)\|^2\mid O_t] = O(R^{-2})$.
Taking expectations yields (ii) (bias comes only from $\rho_t$, and its Lipschitz property in $\theta$ from that
of the expansion coefficients) and (iii). For (iv), apply the same expansion to
$u_t^R(\theta) - u_t^R(\theta')$. The linear coefficients and the integrands are Lipschitz in $\theta$, so the
conditional variance of the linear term is bounded by $CR^{-1}\|\theta-\theta'\|^2$ and the remainder contributes
$O(R^{-2}\|\theta - \theta'\|^2)$. Analogous expansions are developed in detail by \citet{berry2004limit} and
\citet{freyberger2015asymptotic}.
The conditional statements follow because $(O_t,\{\nu_{r,t}\}_r)$ is independent of $\mathcal F_{t-1}$
(Assumption~\ref{ass:sampling}), so conditional moments given $\mathcal F_{t-1}$ equal the corresponding unconditional
moment functions evaluated at $\tilde\theta$.
\end{proof}

The third lemma is a deterministic bound for the recursions that arise below.

\begin{lemma}[Recursion bound]\label{lem:recursion-bound}
Let $\gamma_t = \gamma_0 t^{-a}$ with $a \in (1/2,1)$ and let $c > 0$. There exist $t^\ast \in \mathbb N$ and
$M < \infty$, depending only on $(a, \gamma_0, c, C_1)$, such that any nonnegative sequence $(v_t)_{t \ge t_0}$
with $t_0 \ge t^\ast$ satisfying
\[
  v_t \le (1 - c\gamma_t)\,v_{t-1} + C_1\gamma_t^2 + C_2\gamma_t\kappa, \qquad t > t_0,
\]
for some $C_1, C_2 \ge 0$ and $\kappa \ge 0$, obeys
\[
  v_t \le M\gamma_t + \frac{C_2}{c}\,\kappa + v_{t_0}\exp\bigl(-c(\Gamma_t - \Gamma_{t_0})\bigr), \qquad t > t_0,
\]
where $\Gamma_t - \Gamma_{t_0} = \sum_{s=t_0+1}^t \gamma_s$.
\end{lemma}

\begin{proof}
Choose $t^\ast$ so large that $c\gamma_t \le 1/2$ and $\gamma_{t-1} - \gamma_t \le (c/2)\gamma_t\gamma_{t-1}$ for
all $t > t^\ast$. The latter is possible because
$(\gamma_{t-1}-\gamma_t)/(\gamma_t\gamma_{t-1}) = \gamma_0^{-1}\{t^a - (t-1)^a\} = O(t^{a-1}) \to 0$ for $a < 1$.
Set $M := 2C_1/c$. We prove by induction that
$v_t \le M\gamma_t + (C_2/c)\kappa + v_{t_0}\prod_{s=t_0+1}^t(1 - c\gamma_s)$ for $t \ge t_0$. The claim then
follows from $\prod_s (1 - c\gamma_s) \le \exp(-c\sum_s \gamma_s)$. The base case is trivial. For the induction
step, the hypothesis and the assumed recursion give
\[
  v_t \le (1 - c\gamma_t)M\gamma_{t-1} + C_1\gamma_t^2 + \frac{C_2}{c}\kappa
  + v_{t_0}\prod_{s=t_0+1}^{t}(1 - c\gamma_s),
\]
and $(1-c\gamma_t)M\gamma_{t-1} + C_1\gamma_t^2 \le M\gamma_t$ is equivalent to
$M(\gamma_{t-1} - \gamma_t) + C_1\gamma_t^2 \le cM\gamma_t\gamma_{t-1}$, which holds because
$M(\gamma_{t-1}-\gamma_t) \le (c/2)M\gamma_t\gamma_{t-1}$ by the choice of $t^\ast$ and
$C_1\gamma_t^2 \le (c/2)M\gamma_t\gamma_{t-1}$ by the choice of $M$ and $\gamma_{t-1} \ge \gamma_t$.
\end{proof}

The fourth lemma provides $L^2$ rates for both recursions, a containment result, and an $L^2$ rate for their
difference. Recall $e_t = \theta_t - \theta_t^o$, and define the stopped difference
$\tilde e_t := e_t \mathbf 1\{\tau > t\}$, where
$\tau := \inf\{t > T_0 : \theta_t \notin N_0 \text{ or } \theta_t^o \notin N_0\}$.

\begin{lemma}[$L^2$ rates and containment]\label{lem:L2_rate}
Let Assumptions~\ref{ass:sampling}--\ref{ass:pilot} hold and let $r_1$ be sufficiently small. Then, for all $T$ large
enough and all $t > T_0$, on the event $E_0$:
\begin{enumerate}
\item[(i)] $E[V_t \mid \mathcal F_{T_0}] \le C\bigl(\gamma_t + R^{-2}\bigr)
  + C V_{T_0}\exp\bigl(-c(\Gamma_t - \Gamma_{T_0})\bigr)$, and the same bound with $R^{-2}$ deleted holds for
  $E[V_t^o \mid \mathcal F_{T_0}]$.
\item[(ii)] for every $r > 0$,
\[
  \mathbb P\Bigl( \max_{T_0 < t \le T} \max\{V_t, V_t^o\} > r^2 \Bigr)
  \le \frac{C}{r^2}\Bigl( E[V_{T_0}\mathbf 1_{E_0}] + \sum_{t > T_0} \gamma_t^2
  + R^{-2}\sum_{t=T_0+1}^{T}\gamma_t \Bigr) + \mathbb P(E_0^c);
\]
\item[(iii)] $E[\|\tilde e_t\|^2 \mid \mathcal F_{T_0}] \le C\bigl(R^{-2} + \gamma_t R^{-1}\bigr)$.
\end{enumerate}
\end{lemma}

\begin{proof}
\emph{Part (i).}
Since $\Pi_\Theta$ is nonexpansive and $\theta^\ast \in \Theta$,
$V_t \le \|\theta_{t-1} - \theta^\ast - \gamma_t \hat A g_t^R(\theta_{t-1})\|^2$. Expanding and taking conditional
expectations, with $E[g_t^R(\theta_{t-1})\mid\mathcal F_{t-1}] = \bar m_R(\theta_{t-1})$,
\begin{equation}\label{eq:V-recursion}
  E[V_t \mid \mathcal F_{t-1}]
  \le V_{t-1} - 2\gamma_t\,(\theta_{t-1}-\theta^\ast)^\top \hat A\,\bar m_R(\theta_{t-1})
  + \gamma_t^2\,E\bigl[\|\hat A g_t^R(\theta_{t-1})\|^2 \mid \mathcal F_{t-1}\bigr].
\end{equation}
For the drift term, decompose
$\hat A \bar m_R(\theta) = A^\ast \bar m(\theta) + \Delta A\, \bar m(\theta) + \hat A\,\beta_R(\theta)$. By
Assumption~\ref{ass:drift}, $(\theta - \theta^\ast)^\top A^\ast \bar m(\theta) \ge \tilde C\|\theta-\theta^\ast\|^2$
on $\Theta$. By the Lipschitz property of $\bar m$ and $\bar m(\theta^\ast) = 0$,
$|(\theta-\theta^\ast)^\top \Delta A \bar m(\theta)| \le \|\Delta A\| L_G \|\theta - \theta^\ast\|^2$. By
Lemma~\ref{lem:sim-error}(ii),
$|(\theta-\theta^\ast)^\top \hat A \beta_R(\theta)| \le C R^{-1}\|\theta-\theta^\ast\|
\le (\tilde C/4)\|\theta-\theta^\ast\|^2 + C R^{-2}$. Hence, on $E_0$ with $r_1 \le \tilde C/(2L_G)$,
\[
  (\theta - \theta^\ast)^\top \hat A\,\bar m_R(\theta) \ge \frac{\tilde C}{4}\,\|\theta - \theta^\ast\|^2 - CR^{-2},
  \qquad \forall\,\theta \in \Theta.
\]
For the quadratic term, $E[\|g_t^R(\theta)\|^2] \le C$ uniformly by Lemma~\ref{lem:sim-error}(i). Substituting into
\eqref{eq:V-recursion},
\[
  E[V_t \mid \mathcal F_{t-1}] \le \Bigl(1 - \frac{\tilde C}{2}\gamma_t\Bigr) V_{t-1} + C\gamma_t R^{-2}
  + C\gamma_t^2,
\]
and Lemma~\ref{lem:recursion-bound} (applicable for $T$ large since $T_0 \to \infty$) yields the claim after
iterating conditional expectations. The oracle bound is identical with $\beta_R \equiv 0$.

\emph{Part (ii).}
Drop the negative drift in the display above to obtain
$E[V_t\mid\mathcal F_{t-1}] \le (1 + C\gamma_t^2) V_{t-1} + C\gamma_t(R^{-2} + \gamma_t)$ on $E_0$. Let
$\rho_t := \prod_{s=T_0+1}^t (1 + C\gamma_s^2) \le \rho_\infty < \infty$ and define
\[
  M_t := \frac{V_t}{\rho_t} + C\sum_{s = t+1}^{T}\gamma_s\bigl(R^{-2} + \gamma_s\bigr), \qquad T_0 \le t \le T,
\]
which is a nonnegative supermartingale on $E_0$. By Ville's maximal inequality,
\[
\mathbb P\Bigl(\max_{T_0 < t\le T} M_t \ge \lambda \,\Big|\, \mathcal F_{T_0}\Bigr) \le M_{T_0}/\lambda,
\qquad \text{and} \qquad \max_t V_t \le \rho_\infty \max_t M_t.
\]
Taking expectations restricted to $E_0$, applying the same bound to $V_t^o$, and adding $\mathbb P(E_0^c)$ for the
complementary event delivers (ii).

\emph{Part (iii).}
Note $\tilde e_{T_0} = e_{T_0} = 0$ and $\{\tau > t-1\} \in \mathcal F_{t-1}$. On $\{\tau > t-1\}$, both
$\theta_{t-1}$ and $\theta_{t-1}^o$ lie in $N_0$. By nonexpansiveness,
$\|e_t\| \le \|e_{t-1} - \gamma_t \hat A D_t\|$ with
$D_t := g_t^R(\theta_{t-1}) - g_t(\theta_{t-1}^o)
= [\bar m(\theta_{t-1}) - \bar m(\theta_{t-1}^o)] + [\zeta_t(\theta_{t-1}) - \zeta_t(\theta_{t-1}^o)]
+ u_t^R(\theta_{t-1})$.
Taking conditional expectations in $\|e_t\|^2 \le \|e_{t-1}\|^2 - 2\gamma_t e_{t-1}^\top \hat A D_t
+ \gamma_t^2\|\hat A D_t\|^2$, we bound the three components of $D_t$ in turn. The $\zeta$-difference is a
martingale difference. By the mean-value representation
$\bar m(\theta_{t-1}) - \bar m(\theta_{t-1}^o) = \int_0^1 G\bigl(\theta_{t-1}^o + s e_{t-1}\bigr)ds\; e_{t-1}$ with
the segment contained in the convex set $N_0$, the choice of $N_0$ and $r_1 \le 1/(4 L_G)$ give
$e_{t-1}^\top \hat A[\bar m(\theta_{t-1}) - \bar m(\theta_{t-1}^o)] \ge \tfrac14\|e_{t-1}\|^2$. By
Lemma~\ref{lem:sim-error}(ii) and the inequality $2xy \le \tfrac18 x^2 + 8y^2$,
$2|e_{t-1}^\top \hat A\,\beta_R(\theta_{t-1})| \le \tfrac{1}{8}\|e_{t-1}\|^2 + CR^{-2}$. Finally,
by Lemma~\ref{lem:blp-regularity}(i) and Lemma~\ref{lem:sim-error}(iii),
$E[\|D_t\|^2 \mid \mathcal F_{t-1}] \le C(\|e_{t-1}\|^2 + R^{-1})$. Hence, on $\{\tau > t-1\} \cap E_0$ and for
$T$ large,
\[
  E[\|e_t\|^2 \mid \mathcal F_{t-1}] \le \Bigl(1 - \frac{\gamma_t}{8}\Bigr)\|e_{t-1}\|^2
  + C\gamma_t R^{-2} + C\gamma_t^2 R^{-1}.
\]
Since $\tilde e_t = e_t \mathbf 1\{\tau > t\}$ and $\|\tilde e_t\| \le \|e_t\|\mathbf 1\{\tau > t-1\}$ with
$e_{t-1} = \tilde e_{t-1}$ on $\{\tau > t-1\}$, the same inequality holds for $E[\|\tilde e_t\|^2\mid\mathcal
F_{t-1}]$ with $\|e_{t-1}\|^2$ replaced by $\|\tilde e_{t-1}\|^2$ (on $\{\tau \le t-1\}$ the left side is zero).
Lemma~\ref{lem:recursion-bound} with $v_{T_0} = 0$, $\kappa = R^{-2}$, and $C_1 = CR^{-1}$ gives
$E[\|\tilde e_t\|^2 \mid \mathcal F_{T_0}] \le C(\gamma_t R^{-1} + R^{-2})$.
\end{proof}

The final lemma verifies that the projection constraint is asymptotically inactive. Its proof for the oracle
recursion invokes only the almost-sure convergence established in the first step of the proof of
Theorem~\ref{thm:oracle-clt}, which does not rely on this lemma.

\begin{lemma}[The projection is asymptotically inactive]\label{lem:projection}
Let Assumptions~\ref{ass:sampling}--\ref{ass:pilot} hold and let $\rho := \mathrm{dist}(\theta^\ast,\partial\Theta) > 0$
(Assumption~\ref{ass:param-space}). Then:
\begin{enumerate}
\item[(i)] conditionally on $\mathcal F_{T_0}$, on the event $E_0$, there exists an almost surely finite time
$\bar t$ such that the projection in the oracle recursion~\eqref{eq:snfp-oracle} is inactive for all $t \ge \bar t$.
That is, the projection binds at most finitely many times, almost surely.
\item[(ii)] if in addition $T^{1-a}/R^2 \to 0$ (implied by $\sqrt T/R \to 0$), then
\[
\mathbb P\bigl(\text{the projection in \eqref{eq:snfp} is active for some } t \in (T_0, T]\bigr) \to 0.
\]
\end{enumerate}
\end{lemma}

\begin{proof}
\emph{(i)} By the first step of the proof of Theorem~\ref{thm:oracle-clt}, $V_t^o \to 0$ almost surely, so
$\theta_t^o \in B(\theta^\ast, \rho/2)$ for all $t$ large. Moreover,
$\sum_t \mathbb P\bigl(\gamma_t\|\hat A g_t(\theta_{t-1}^o)\| > \rho/2 \mid \mathcal F_{t-1}\bigr)
\le C\rho^{-2}\sum_t \gamma_t^2 < \infty$ almost surely, by Chebyshev's inequality and
Lemma~\ref{lem:sim-error}(i), so L\'evy's extension of the Borel--Cantelli lemma implies
$\gamma_t\|\hat A g_t(\theta_{t-1}^o)\| \le \rho/2$ for all $t$ large, almost surely. On the intersection of these
events, the unprojected point lies in $B(\theta^\ast, \rho) \subseteq \Theta$, so the projection is inactive.
We note that this argument rests only on the nonexpansiveness of $\Pi_\Theta$, the interiority
$\theta^\ast \in \mathrm{int}(\Theta)$, and the almost-sure convergence delivered by the drift condition. The
compactness of $\Theta$ plays no role, and the conditional form of the Borel--Cantelli argument remains valid when
the conditional moment bound of Lemma~\ref{lem:sim-error}(i) is relaxed to grow with $V_{t-1}^o$.

\emph{(ii)} By Lemma~\ref{lem:L2_rate}(ii) with $r = \rho/2$,
$\mathbb P(\max_{T_0 < t \le T} V_t > \rho^2/4) \to 0$, using $E[V_{T_0}\mathbf 1_{E_0}] \to 0$ as in the proof
of Theorem~\ref{thm:negligible-sim},
$\sum_{t>T_0}\gamma_t^2 \to 0$, and $R^{-2}\sum_{t \le T}\gamma_t = O(T^{1-a}/R^2) \to 0$. Moreover,
$\mathbb P\bigl(\exists\, t \in (T_0,T]: \gamma_t\|\hat A g_t^R(\theta_{t-1})\| > \rho/2\bigr)
\le C\rho^{-2}\sum_{t > T_0}\gamma_t^2 \to 0$ by the union bound, Chebyshev's inequality, and
Lemma~\ref{lem:sim-error}(i). On the complement of these events the unprojected point lies in
$B(\theta^\ast,\rho) \subseteq \Theta$ for every $t \in (T_0,T]$.
\end{proof}

\section{Proofs of the Theorems}\label{appendix:thm1_proof}

\begin{proof}[Proof of Theorem~\ref{thm:consistency-snfp}]
By Lemma~\ref{lem:L2_rate}(i),
\[
  E[V_T \mathbf 1_{E_0}] \le C\bigl(\gamma_T + R^{-2}\bigr)
  + C\,E[V_{T_0}\mathbf 1_{E_0}]\exp\bigl(-c(\Gamma_T - \Gamma_{T_0})\bigr).
\]
As $T \to \infty$, $\gamma_T \to 0$ and $R \to \infty$ (Assumption~\ref{ass:rates}). Moreover,
$\Gamma_T - \Gamma_{T_0} \ge c\gamma_0\bigl(T^{1-a} - T_0^{1-a}\bigr) \to \infty$ because $T_0/T \to 0$, and
$V_{T_0} \le r_0^2/4$ on the event $E_0$ by the definition of $E_0$. Hence $E[V_T\mathbf 1_{E_0}] \to 0$, and since
$\mathbb P(E_0) \to 1$, $\theta_T \pto \theta^\ast$. For the average, Jensen's inequality gives
\[
  E\bigl[\|\bar\theta_T - \theta^\ast\|\mathbf 1_{E_0}\bigr]
  \le \frac{1}{T - T_0}\sum_{t = T_0+1}^T \Bigl( C\sqrt{\gamma_t} + CR^{-1}
  + C\exp\bigl(-\tfrac{c}{2}(\Gamma_t - \Gamma_{T_0})\bigr) \Bigr).
\]
The first average is $O(T^{-a/2})$, and the second is $O(R^{-1})$. For the third, the telescoping bound
$\gamma_t e^{-c(\Gamma_t - \Gamma_{T_0})} \le c^{-1}\bigl(e^{-c(\Gamma_{t-1}-\Gamma_{T_0})}
- e^{-c(\Gamma_t - \Gamma_{T_0})}\bigr)$ implies
$\sum_{t > T_0} e^{-c(\Gamma_t - \Gamma_{T_0})} \le (c\gamma_T)^{-1} = O(T^a)$, so the third average is
$O(T^{a-1}) \to 0$ since $a < 1$. Hence $\bar\theta_T \pto \theta^\ast$. The oracle statements follow from the
identical argument with $R^{-1}$ and $R^{-2}$ deleted.
\end{proof}

\begin{proof}[Proof of Theorem~\ref{thm:oracle-clt}]
Take $r_1$ small enough that, on $E_0$, in addition to the requirements of Lemma~\ref{lem:L2_rate}, every
eigenvalue of $\hat A G^\ast = I_d + \Delta A\, G^\ast$ has real part at least $1/2$. This is possible since
$\|\Delta A\,G^\ast\| \le r_1\|G^\ast\|$. Throughout, we argue conditionally on $\mathcal F_{T_0}$ on the event
$E_0$, so that $\hat A$ is a fixed matrix with $\hat A G^\ast$ Hurwitz.

First, $\theta_t^o \to \theta^\ast$ almost surely, conditionally on $\mathcal F_{T_0}$. Indeed, the drift bound in
the proof of Lemma~\ref{lem:L2_rate}(i) with $\beta_R \equiv 0$ gives
$E[V_t^o \mid \mathcal F_{t-1}] \le (1 + C\gamma_t^2)V_{t-1}^o - \tilde c\,\gamma_t V_{t-1}^o + C\gamma_t^2$,
so the almost-supermartingale convergence theorem of \citet{robbins1971convergence} implies that $V_t^o$ converges
almost surely and $\sum_t \gamma_t V_{t-1}^o < \infty$. Since $\sum_t \gamma_t = \infty$, the limit is zero.

Second, by Lemma~\ref{lem:projection}(i), there is an almost surely finite
time beyond which the projection in~\eqref{eq:snfp-oracle} is inactive, so the tail of the oracle sequence obeys
the unprojected recursion
$\theta_t^o = \theta_{t-1}^o - \gamma_t\bigl[h(\theta_{t-1}^o) + \varepsilon_t\bigr]$ with mean field
$h(\theta) := \hat A\,\bar m(\theta)$ and noise $\varepsilon_t := \hat A\,\zeta_t(\theta_{t-1}^o)$.
The conditions of Theorem~2 of \citet{polyak1992acceleration} hold. We have $h(\theta^\ast) = 0$, and
$\nabla h(\theta^\ast) = \hat A G^\ast$ is Hurwitz with $\nabla h$ Lipschitz near $\theta^\ast$. The noise
$\varepsilon_t$ is a martingale difference sequence whose conditional covariance
$\hat A\,\Omega(\theta_{t-1}^o)\hat A^\top \to \hat A\,\Omega\,\hat A^\top$ almost surely, where
$\Omega(\theta)$ is continuous at $\theta^\ast$ by Lemma~\ref{lem:blp-regularity}(ii), and the Lindeberg
condition is immediate because the moments are uniformly bounded by Lemma~\ref{lem:blp-regularity}(i).
Therefore, conditionally on $\mathcal F_{T_0}$ on the event $E_0$,
\[
  \sqrt{T - T_0}\,\bigl(\bar\theta_T^o - \theta^\ast\bigr) \;\xrightarrow{d}\;
  N\bigl(0,\;\Sigma(\hat A)\bigr), \qquad
  \Sigma(\hat A) := (\hat A G^\ast)^{-1}\,\hat A\,\Omega\,\hat A^\top\,(\hat A G^\ast)^{-\top}.
\]
This convergence is uniform over scaling matrices with $\|\hat A - A^\ast\| \le r_1$, since the constants in the
argument of \citet{polyak1992acceleration} depend on the gain matrix only through the Hurwitz margin of
$\hat A G^\ast$ and the moment bounds, both of which are uniform on $E_0$.

Third, we remove the conditioning. Let $f$ be bounded and Lipschitz. By the previous step,
$E[f(\sqrt{T-T_0}(\bar\theta_T^o - \theta^\ast)) \mid \mathcal F_{T_0}]\,\mathbf 1_{E_0}
= E[f(Z_{\Sigma(\hat A)})]\,\mathbf 1_{E_0} + o_p(1)$, where $Z_\Sigma \sim N(0,\Sigma)$. Since
$\Delta A \pto 0$ (Assumption~\ref{ass:pilot} with $T_0, R \to \infty$), $\Sigma(\hat A) \pto \Sigma(A^\ast)
= A^\ast\Omega A^{\ast\top} = \Sigma_\theta$ by continuity of $\Sigma(\cdot)$ at $A^\ast$ and the identity
$A^\ast G^\ast = I_d$. Hence $E[f(Z_{\Sigma(\hat A)})] \pto E[f(Z_{\Sigma_\theta})]$, and bounded convergence
yields $E[f(\sqrt{T-T_0}(\bar\theta_T^o - \theta^\ast))] \to E[f(Z_{\Sigma_\theta})]$. Finally,
$\sqrt{T}/\sqrt{T-T_0} \to 1$ because $T_0 = o(T)$, and Slutsky's theorem completes the proof.
\end{proof}

\begin{proof}[Proof of Theorem~\ref{thm:negligible-sim}]
Let $\tau$ and $\tilde e_t$ be as in Lemma~\ref{lem:L2_rate}. On the event $E_0 \cap \{\tau > T\}$, $e_t =
\tilde e_t$ for all $T_0 < t \le T$, so
\[
  \sqrt T\,\|\bar\theta_T - \bar\theta_T^o\| \le \frac{\sqrt T}{T - T_0}\sum_{t=T_0+1}^T \|\tilde e_t\|.
\]
By Lemma~\ref{lem:L2_rate}(iii) and Jensen's inequality,
$E[\|\tilde e_t\|\mathbf 1_{E_0}] \le C(R^{-1} + \gamma_t^{1/2}R^{-1/2})$, whence
\[
  E\Bigl[\frac{\sqrt T}{T-T_0}\sum_{t=T_0+1}^T \|\tilde e_t\|\,\mathbf 1_{E_0}\Bigr]
  \le C\Bigl(\frac{\sqrt T}{R} + \frac{T^{(1-a)/2}}{\sqrt R}\Bigr) \to 0,
\]
where the first term vanishes by hypothesis and the second because $\sqrt T/R \to 0$ and $a > 1/2$ imply
$R / T^{1-a} = (R/\sqrt T)\,T^{a - 1/2} \to \infty$. By Markov's inequality,
$\sqrt T(\bar\theta_T - \bar\theta_T^o)\,\mathbf 1_{E_0 \cap \{\tau > T\}} \pto 0$.
It remains to show $\mathbb P(\{\tau \le T\} \cup E_0^c) \to 0$. By Lemma~\ref{lem:L2_rate}(ii) with $r = r_0$,
\[
  \mathbb P(\tau \le T) \le \frac{C}{r_0^2}\Bigl(E[V_{T_0}\mathbf 1_{E_0}] + \sum_{t>T_0}\gamma_t^2
  + R^{-2}\sum_{t \le T}\gamma_t\Bigr) + \mathbb P(E_0^c).
\]
Here $E[V_{T_0}\mathbf 1_{E_0}] \to 0$ by dominated convergence, since $V_{T_0}\mathbf 1_{E_0} \le r_0^2/4$ by
the definition of $E_0$ and $V_{T_0} \pto 0$ by consistency of the pilot estimate (Assumption~\ref{ass:pilot}).
Next, $\sum_{t > T_0}\gamma_t^2 = O(T_0^{1-2a}) \to 0$ since $a > 1/2$ and $T_0 \to \infty$. Finally,
$R^{-2}\sum_{t\le T}\gamma_t = O(T^{1-a}/R^2) = O\bigl((\sqrt T/R)^2\,T^{-a}\bigr) \to 0$. Since
$\mathbb P(E_0^c) \to 0$, the proof is complete.
\end{proof}

\begin{proof}[Proof of Corollary~\ref{thm:asymptotic_dist}]
Write $\sqrt T(\bar\theta_T - \theta^\ast) = \sqrt T(\bar\theta_T^o - \theta^\ast)
+ \sqrt T(\bar\theta_T - \bar\theta_T^o)$ and apply Theorems~\ref{thm:oracle-clt} and \ref{thm:negligible-sim}
with Slutsky's theorem. When $W = \Omega^{-1}$, $\Sigma_\theta$ reduces to $(G^{*\top}\Omega^{-1}G^\ast)^{-1}$ by
direct substitution.
\end{proof}

\clearpage
\section{Common Simulation Draws: Auxiliary Results and Proofs}\label{appendix:common_draws}

Throughout this appendix Assumption~\ref{ass:sampling-common} is in force, and the superscript $c$ marks objects that
are random through the common block, namely $\bar m^{c}_R(\theta) := E[g(O;\theta,P^R) \mid P^R]$ and
$b_R(\theta) := \bar m^{c}_R(\theta) - \bar m(\theta,P^\ast)$. These are \emph{not} the deterministic $\bar m_R$
and $\beta_R$ of Appendix~\ref{appendix:lemmas}, which describe the fresh-draw scheme of
Lemma~\ref{lem:sim-error}. In particular the bound $\|\beta_R\| \le C R^{-1}$ of Lemma~\ref{lem:sim-error}(ii)
must not be carried across, since $\|b_R\|_\infty$ is of exact order $R^{-1/2}$. The influence function $\psi$ is as
in~\eqref{eq:psi-common}, $\mathcal L_O(\theta)$ is the inversion multiplier $\mathcal L_t(\theta)$ of the proof of
Lemma~\ref{lem:sim-error} written for a generic market $O$, and $N_0$ is the ball around $\theta^\ast$ on which the
symmetric part of $A^\ast G(\theta)$ exceeds $\tfrac12 I_d$. Write $\Delta A := \hat A_{T_0} - A^\ast$ and define
the two good events
\[
  E_0 := \bigl\{\|\Delta A\| \le r_1\bigr\}, \qquad
  E_1 := \Bigl\{\|b_R\|_\infty + \mathrm{Lip}(b_R)
    + \sup_{\theta\in\Theta}\bigl\|\nabla_\theta \bar m^{c}_R(\theta) - G(\theta)\bigr\| \le r_2\Bigr\},
\]
where $r_1, r_2 > 0$ are fixed and small. Here $\mathbb P(E_0) \to 1$ by Assumption~\ref{ass:pilot} and
$\mathbb P(E_1) \to 1$ by Lemma~\ref{lem:perturbation}.

\begin{lemma}[Perturbation of the mean moment]\label{lem:perturbation}
Under Assumption~\ref{ass:sampling-common} and Assumptions~\ref{ass:param-space} and~\ref{ass:regular-markets},
\[
  b_R(\theta) = \frac1R\sum_{r=1}^R \psi(\nu_r;\theta) + r_R(\theta), \qquad \|r_R\|_\infty = O_p(R^{-1}),
\]
and consequently $\|b_R\|_\infty = O_p(R^{-1/2})$, $\mathrm{Lip}(b_R) = O_p(R^{-1/2})$, and
$\sup_{\theta\in\Theta}\|\nabla_\theta \bar m^{c}_R(\theta) - G(\theta)\| = O_p(R^{-1/2})$.
\end{lemma}

\begin{proof}
Take expectations over $O$, given the block, in the second-order expansion of the BLP inversion established in the
proof of Lemma~\ref{lem:sim-error},
$u^R(\theta) = \mathcal L_O(\theta)[\tilde s(\cdot;P^R) - \tilde s(\cdot;P^\ast)] + \rho_O(\theta)$ with
$\|\rho_O(\theta)\| \le C\|\tilde s(\cdot;P^R) - \tilde s(\cdot;P^\ast)\|^2$. At any $(\delta,\theta_2)$ the share
error is the centered average $R^{-1}\sum_r\{\sigma(\delta,\theta_2,\nu_r) - E_\nu[\sigma(\delta,\theta_2,\nu)]\}$,
so Fubini gives $E_O[\mathcal L_O(\theta)\{\cdot\}] = R^{-1}\sum_r \psi(\nu_r;\theta)$ with $\psi$ as
in~\eqref{eq:psi-common}. Write
$Z_R := \sup_{(\delta,\theta_2)}\|\tilde s(\delta,\theta_2;P^R) - \tilde s(\delta,\theta_2;P^\ast)\|$. The class
$\{\nu \mapsto \sigma_j(\delta,\theta_2,\nu)\}$ is bounded by one, indexed by a compact set, and Lipschitz in the
index with a square-integrable envelope (derivatives of the mixed logit bring down polynomials in $\|\nu\|$), so
its bracketing entropy integral is finite and the maximal inequality for empirical processes
\citep{van1996weak} gives $E[Z_R^2] \le C R^{-1}$, and hence $\|r_R\|_\infty \le C Z_R^2 = O_p(R^{-1})$.
The same
maximal inequality, applied to the uniformly bounded, centered, index-Lipschitz classes
$\{\nu \mapsto \psi(\nu;\theta)\}$ and $\{\nu \mapsto \nabla_\theta\psi(\nu;\theta)\}$, gives $O_p(R^{-1/2})$ for
the corresponding suprema. The three stated rates follow, using
$\nabla_\theta b_R = R^{-1}\sum_r \nabla_\theta\psi(\nu_r;\cdot) + \nabla_\theta r_R$ and the fact that the
remainder in the expansion is Lipschitz in $\theta$ with constant $O_p(Z_R^2)$.
\end{proof}

\begin{lemma}[Pseudo-true parameter]\label{lem:pseudo-true}
Let Assumption~\ref{ass:sampling-common} and Assumptions~\ref{ass:param-space}, \ref{ass:identification},
\ref{ass:regular-markets} and~\ref{ass:drift} hold, with $r_2$ small. Then, on $E_1$:
\begin{enumerate}
\item[(i)] $A^\ast\bar m^{c}_R(\theta) = 0$ has a unique solution $\theta^R$ in $\Theta$, it lies in $N_0$, and
$\|\theta^R - \theta^\ast\| \le C\|b_R\|_\infty$.
\item[(ii)] $\theta^R - \theta^\ast = -A^\ast b_R(\theta^\ast) + O_p(R^{-1})$ and
$\sqrt R\,(\theta^R - \theta^\ast) \xrightarrow{d} N(0,\Lambda_\nu)$.
\end{enumerate}
If in addition Assumption~\ref{ass:pilot} holds and $r_1$ is small, then, on $E_0 \cap E_1$, the equation
$\hat A_{T_0}\,\bar m^{c}_R(\theta) = 0$ has a unique solution $\theta_c^R$ in $N_0$, and
$\|\theta_c^R - \theta^R\| \le C\|\Delta A\|\,\|b_R\|_\infty$.
\end{lemma}

\begin{proof}
\emph{(i).} On $E_1$, Assumption~\ref{ass:drift} gives
$(\theta-\theta^\ast)^\top A^\ast\bar m^{c}_R(\theta) \ge \tilde C\|\theta-\theta^\ast\|^2
- \|A^\ast\|\,\|b_R\|_\infty\|\theta-\theta^\ast\| > 0$ whenever $\|\theta-\theta^\ast\| > \|A^\ast\|r_2/\tilde C$,
so every solution lies in a ball contained in $N_0$ for $r_2$ small. On $N_0$ the map
$F(\theta) := A^\ast\bar m^{c}_R(\theta)$ has Jacobian
$A^\ast G(\theta) + A^\ast[\nabla_\theta\bar m^{c}_R(\theta) - G(\theta)]$, whose symmetric part exceeds
$\tfrac14 I_d$ for $r_2$ small, so $F$ is strongly monotone on the compact convex set $N_0$ and has there a unique,
interior zero $\theta^R$. Uniqueness on all of $\Theta$ follows from the first display. Evaluating
Assumption~\ref{ass:drift} at $\theta^R$, where $\bar m(\theta^R,P^\ast) = -b_R(\theta^R)$, gives
$\tilde C\|\theta^R-\theta^\ast\|^2 \le \|A^\ast\|\,\|b_R\|_\infty\|\theta^R-\theta^\ast\|$.

\emph{(ii).} A mean-value expansion of $0 = A^\ast\bar m^{c}_R(\theta^R)$ about $\theta^\ast$, together with
$A^\ast G^\ast = I_d$ (Assumption~\ref{ass:identification}), the Lipschitz property of $G$, and
$b_R(\theta^R) = b_R(\theta^\ast) + O(\mathrm{Lip}(b_R)\|\theta^R-\theta^\ast\|)$, makes every correction term
$O_p(R^{-1})$ by part~(i) and Lemma~\ref{lem:perturbation}, whence
$\theta^R-\theta^\ast = -A^\ast b_R(\theta^\ast) + O_p(R^{-1})$. By Lemma~\ref{lem:perturbation},
$\sqrt R\,b_R(\theta^\ast) = R^{-1/2}\sum_r \psi(\nu_r;\theta^\ast) + O_p(R^{-1/2})$, an i.i.d.\ sum of bounded,
mean-zero vectors, so the multivariate central limit theorem and the map $x \mapsto -A^\ast x$ give the stated
limit with $\Lambda_\nu = A^\ast\Omega_\nu A^{\ast\top}$.

\emph{Final clause.} Under Assumption~\ref{ass:pilot}, $\hat A_{T_0} \pto A^\ast$, so $\mathbb P(E_0) \to 1$. On
$E_0 \cap E_1$ the same strong-monotonicity argument applies to $\hat F(\theta) := \hat A_{T_0}\bar m^{c}_R(\theta)$,
whose Jacobian's symmetric part exceeds $\tfrac18 I_d$ for $r_1, r_2$ small. With $\hat F(\theta_c^R) = 0$ and
$A^\ast\bar m^{c}_R(\theta^R) = 0$, strong monotonicity gives
$\tfrac18\|\theta_c^R-\theta^R\|^2 \le -(\theta_c^R-\theta^R)^\top\Delta A\,\bar m^{c}_R(\theta^R)$, and
$\|\bar m^{c}_R(\theta^R)\| \le L_G\|\theta^R-\theta^\ast\| + \|b_R\|_\infty \le C\|b_R\|_\infty$ by part~(i).
\end{proof}

\begin{proof}[Proof of Theorem~\ref{thm:common-draws}]
Condition on $\mathcal F_{T_0}$, which now contains the block. Given the block, $g(O_t;\theta,P^R)$ is an
i.i.d.-across-$t$ unbiased draw of $\bar m^{c}_R(\theta)$, so the recursion is an \emph{exact} stochastic
approximation to $\bar m^{c}_R$. There is no simulation-error term, and the coupling argument of
Theorem~\ref{thm:negligible-sim} is not needed. On $E_0 \cap E_1$ the proofs of
Theorems~\ref{thm:consistency-snfp} and~\ref{thm:oracle-clt} therefore apply verbatim with
$(P^\ast, \theta^\ast, G^\ast, \Omega, \beta_R, u_t^R)$ replaced by
$(P^R, \theta_c^R, \nabla_\theta\bar m^{c}_R(\theta_c^R),
\mathrm{Var}[g(O;\theta_c^R,P^R)\mid P^R], 0, 0)$. Two substitutions need checking, and both are supplied by
Lemmas~\ref{lem:perturbation}--\ref{lem:pseudo-true}. First, the drift transfers. For $r_1, r_2$ small,
$(\theta-\theta_c^R)^\top \hat A_{T_0}\bar m^{c}_R(\theta) \ge c_0\|\theta-\theta_c^R\|^2$ on all of $\Theta$,
because on $N_0$ this is the lower bound on the symmetric part of $\hat A_{T_0}\nabla_\theta\bar m^{c}_R$ used in
the proof of Lemma~\ref{lem:pseudo-true}, while off $N_0$ each of the three replacements
$A^\ast \to \hat A_{T_0}$, $\bar m(\cdot,P^\ast) \to \bar m^{c}_R$, $\theta^\ast \to \theta_c^R$ perturbs the
two-region bound of Assumption~\ref{ass:drift} by at most $C(r_1+r_2)$, and compactness lifts it. Second, the
conditional sandwich variance converges, $\Sigma(\hat A_{T_0},P^R) \pto \Sigma_\theta$, since
$\hat A_{T_0} \pto A^\ast$, $\nabla_\theta\bar m^{c}_R(\theta_c^R) \pto G^\ast$, and
$\sup_\theta E_O\|g(O;\theta,P^R) - g(O;\theta,P^\ast)\|^2 = O_p(R^{-1})$. Hence, writing
$\xi_T := \sqrt{T-T_0}\,(\bar\theta_T - \theta_c^R)$, the conditional limit of $\xi_T$ is $N(0,\Sigma_\theta)$ and
\begin{equation}\label{eq:cd-decomposition}
  \sqrt T\,(\bar\theta_T - \theta^\ast)
  = \sqrt{\tfrac{T}{T-T_0}}\;\xi_T \;+\; \sqrt T\,(\theta_c^R - \theta^R)
  \;+\; \sqrt{\tfrac{T}{R}}\;\sqrt R\,(\theta^R - \theta^\ast).
\end{equation}
Part~(i) follows from $\bar\theta_T - \theta_c^R \pto 0$ and $\theta_c^R - \theta^\ast = O_p(R^{-1/2})$, which
require only $T_0 \to \infty$ and $R \to \infty$. For part~(ii), the first factor tends to one because
$T_0 = o(T)$. The middle term is $o_p(1)$ because
$\sqrt T\|\theta_c^R - \theta^R\| \le C\|\Delta A\|\,O_p(\sqrt{T/R})$ with $T/R$ bounded and $\Delta A \pto 0$.
Finally, the remaining two terms are asymptotically independent, $\sqrt R(\theta^R-\theta^\ast)$ being
$\mathcal F_{T_0}$-measurable while the conditional limit of $\xi_T$ does not depend on the block, so
Lemma~\ref{lem:pseudo-true}(ii) and Slutsky's theorem give $N(0,\Sigma_\theta + c\Lambda_\nu)$. For part~(iii),
multiply~\eqref{eq:cd-decomposition} by $\sqrt{R/T}$ \emph{before} bounding any term. The first becomes
$\sqrt{R/T}\,\xi_T(1+o(1)) = o_p(1)$ since $R/T \to 0$ and $\xi_T$ is conditionally tight, the middle becomes
$\sqrt R\,\|\theta_c^R-\theta^R\| \le C\|\Delta A\|\,O_p(1) = o_p(1)$, and the last is
$\sqrt R(\theta^R-\theta^\ast) \xrightarrow{d} N(0,\Lambda_\nu)$. The bound used for the middle term in part~(ii),
of order $\|\Delta A\|\sqrt{T/R}$, is not $o_p(1)$ when $R/T \to 0$, which is why the rescaling must precede it.
\end{proof}

\begin{proof}[Proof of Corollary~\ref{cor:qmc}]
Under the stated quadrature condition the linear term of Lemma~\ref{lem:perturbation} is bounded by $C\epsilon_R$,
so $\|b_R\|_\infty + \mathrm{Lip}(b_R) + \sup_\theta\|\nabla_\theta\bar m^{c}_R - G\| \le C\epsilon_R$. Choosing
$R_0$ with $C\epsilon_R \le r_2$ for all $R \ge R_0$ makes $E_1$ hold surely, and
Lemma~\ref{lem:pseudo-true}(i) gives $\|\theta^R - \theta^\ast\| \le C\epsilon_R$ with no central limit term,
part~(ii) of that lemma being unnecessary because the displacement is now a deterministic bias. The conditional
analysis in the previous proof is unchanged, with $O_p(R^{-1})$ replaced by $C\epsilon_R^2$ in the variance
step, because it uses only the $\mathcal F_{T_0}$-measurability of $P^R$, the martingale-difference property of
the noise, the transferred drift, and the bounds defining $E_1$. In~\eqref{eq:cd-decomposition} the last term is
then $O(\sqrt T\,\epsilon_R) \to 0$ and the middle term is $O_p(\|\Delta A\|\sqrt T\,\epsilon_R) = o_p(1)$, leaving
the limit $N(0,\Sigma_\theta)$. Conditional and unconditional centering coincide, so the plug-in sandwich
interval requires no correction. For a randomly scrambled net the same argument runs conditionally on the scramble with
$\epsilon_R = O_p(R^{-1+\eta})$, and the unconditional statements follow by dominated convergence.
\end{proof}
\clearpage
\section{Simulation Details: Small-Sample Bias in the BLP Estimator}\label{appendix:small_sample_details}

This appendix states the data-generating process behind Figure~\ref{fig:small_sample_blp} and
reports the accompanying accuracy statistics.

We use a three-parameter BLP model with $J = 4$ products per
market. The mean utility is $\delta_{jt} = \alpha + \mu \, p_{jt} + \xi_{jt}$, with a single
random coefficient $\sigma_p$ on price. Consumer $i$'s utility is
$u_{ijt} = \delta_{jt} + \sigma_p \, p_{jt} \, \nu_i + \varepsilon_{ijt}$, where
$\nu_i \sim N(0,1)$ and $\varepsilon_{ijt}$ is type I extreme value. True values are
$\alpha_0 = -3$, $\mu_0 = -1$, $\sigma_{p,0} = 0.5$. The data-generating process draws
$\xi_{jt}$ and $\eta_{jt}$ from a truncated normal distribution and sets prices as
$p_{jt} = 0.5 |2 + \rho \, \xi_{jt} + \eta_{jt}|$, where $\rho$ controls endogeneity. The
instrument is $z_{1,jt} = a_{1,jt} + \lambda_z \, \eta_{jt}$, where
$a_{1,jt} \sim \text{Uniform}(0,1)$ and $\lambda_z$ governs the correlation between the
instrument and the price shock $\eta_{jt}$. The
instrument vector is $Z_{jt} = (1, z_{1,jt}, z_{1,jt}^2, z_{1,jt}^3)'$. Estimation uses NFP
with $R = 500$ simulation draws, L-BFGS-B optimization over $\sigma_p \in [0.01, 5]$ with
random starting value $\sigma_p^{(0)} \sim \text{Uniform}(0.01, 5)$, and the plug-in sandwich
variance estimator.

Table~\ref{tab:app_batch_detail} reports bias, RMSE, and the boundary rate for the NFP
estimator with random starting values. All results use $\lambda_z = 0.5$, $\rho = 2$,
$J = 4$, $R = 500$, and $1{,}000$ replications per sample size. Bias and RMSE are computed
across all converged replications. The boundary rate is the share of replications with
$\hat\sigma_p \leq 0.02$ (the small-sample artifact discussed in
Section~\ref{section:intro}), and decreases monotonically from $33\%$ at $T = 2{,}000$
to less than $1\%$ at $T = 128{,}000$.

\begin{table}[h!]
\centering
\caption{Small-sample performance of NFP.}
\label{tab:app_batch_detail}
\small
\begin{tabular}{ll *{4}{S[table-format=-1.4]}}
\toprule
& & {$T = 2{,}000$} & {$T = 8{,}000$} & {$T = 32{,}000$} & {$T = 128{,}000$} \\
\midrule
\multicolumn{2}{l}{\textit{Bias}} \\
& Intercept   &  0.0720 &  0.0039 & -0.0031 & -0.0044 \\
& $\mu$       & -0.1706 & -0.0156 &  0.0040 &  0.0078 \\
& $\sigma_p$  & -0.0177 & -0.0780 & -0.0339 & -0.0115 \\[3pt]
\multicolumn{2}{l}{\textit{RMSE}} \\
& Intercept   &  0.3367 &  0.1757 &  0.0966 &  0.0467 \\
& $\mu$       &  0.7259 &  0.3477 &  0.1864 &  0.0895 \\
& $\sigma_p$  &  0.5305 &  0.3426 &  0.1942 &  0.0818 \\[3pt]
\multicolumn{2}{l}{\textit{Boundary rate ($\hat\sigma_p \leq 0.02$)}} \\
&             & {33.1\%} & {24.9\%} & {6.1\%}  & {0.1\%}  \\
\bottomrule
\end{tabular}
\end{table}

Figure~\ref{fig:app_nfp_tstat} plots the $t$-statistic density of the NFP estimator
against the $N(0,1)$ reference, one panel per parameter. As $T$ grows, the
densities concentrate around zero and sharpen toward the standard-normal shape,
confirming the asymptotic approximation in
Section~\ref{section:asymptotic_theory}. At $T = 2{,}000$ the intercept and
$\mu$ densities exhibit visible skew, whereas by $T = 128{,}000$ all three distributions
are close to $N(0,1)$.
\begin{figure}[h!]
    \centering
    \caption{$t$-statistic densities of the NFP estimator at
            $\lambda_z = 0.5$, $\rho = 2$.} \includegraphics[width=0.95\linewidth]{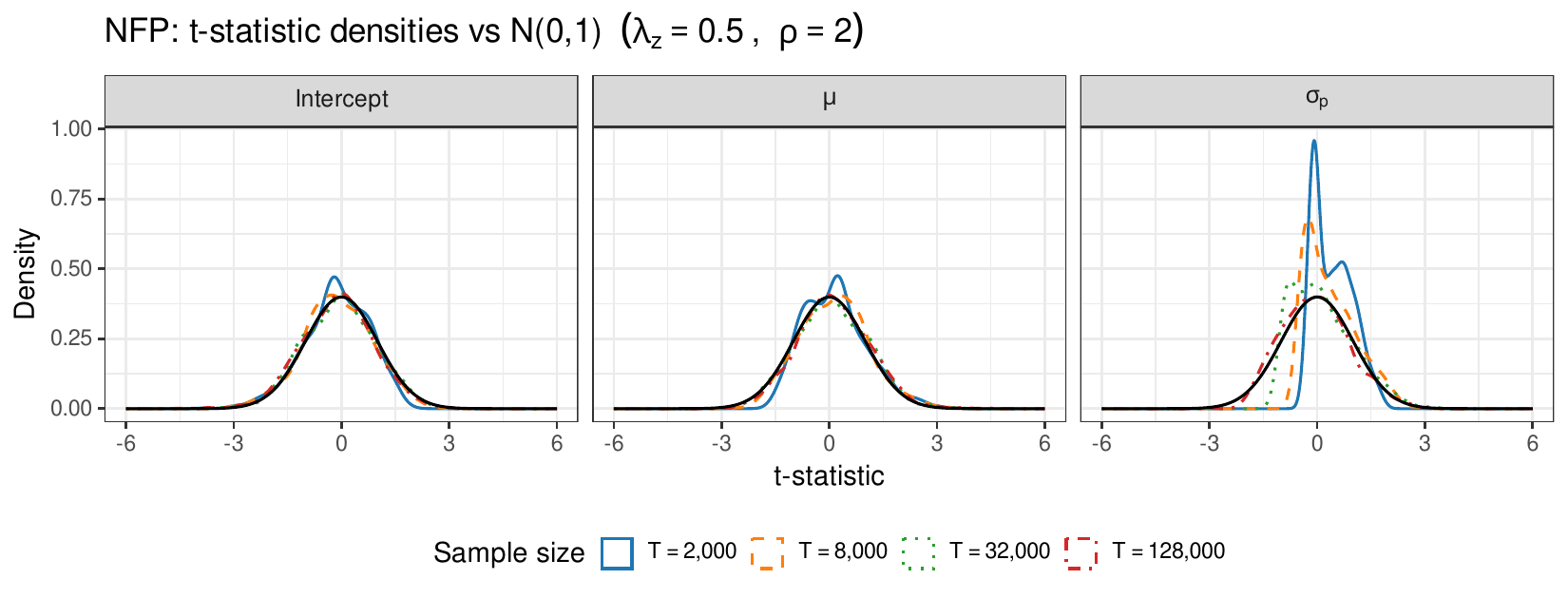} \label{fig:app_nfp_tstat} 
\end{figure} 
Figure~\ref{fig:app_qq_nfp_overlap} displays the QQ-plots of the studentized NFP estimator, $(\hat\theta - \theta^\ast)/\hat{\mathrm{SE}}$, against the standard normal. All four sample sizes are overlaid on a single panel per parameter, which makes the collapse toward the $45^\circ$ line as $T$ grows immediately visible. For $\sigma_p$, boundary replications ($\hat\sigma_p \leq 0.02$) are excluded, consistent with Figure~\ref{fig:small_sample_blp} and Table~\ref{tab:app_batch_detail}. 
\begin{figure}[htbp] \centering
\caption{QQ-plots of the studentized NFP estimator, all sample sizes overlaid. One panel per parameter, with
$\lambda_z = 0.5$, $\rho = 2$.} \label{fig:app_qq_nfp_overlap}
    \begin{subfigure}[t]{0.32\textwidth} 
        \includegraphics[width=\textwidth]{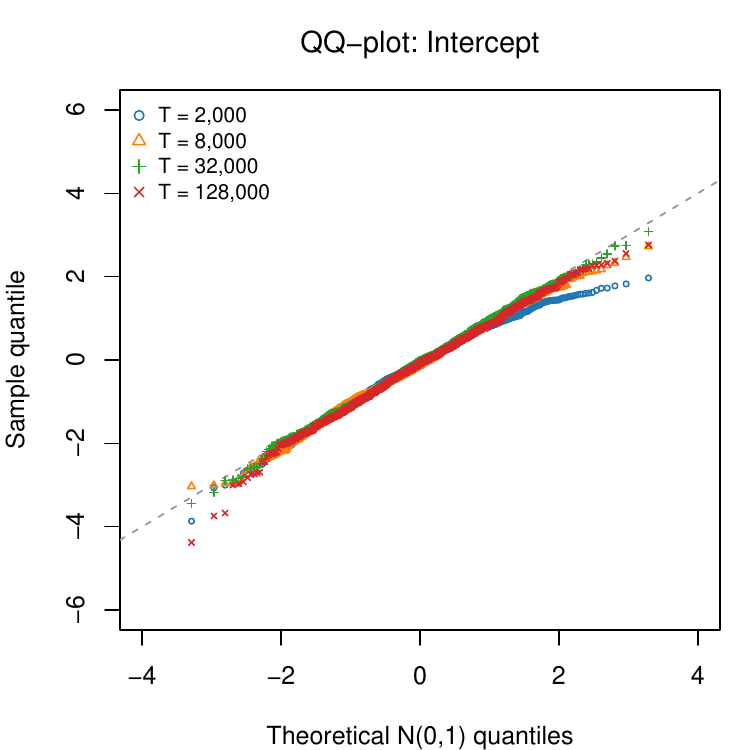} \caption{Intercept} \end{subfigure} \begin{subfigure}[t]{0.32\textwidth} \includegraphics[width=\textwidth]{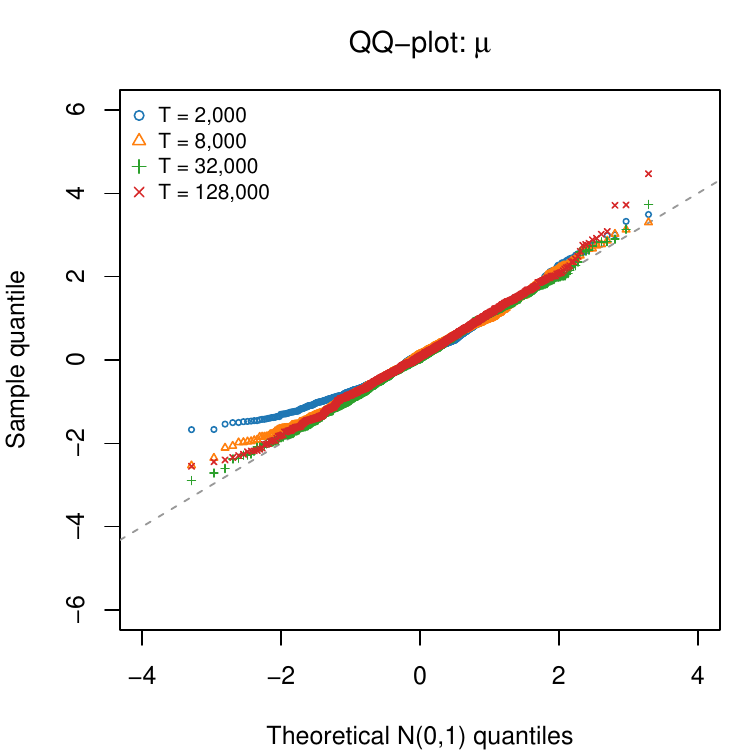} \caption{$\mu$} 
    \end{subfigure} 
    \begin{subfigure}[t]{0.32\textwidth} 
        \includegraphics[width=\textwidth]{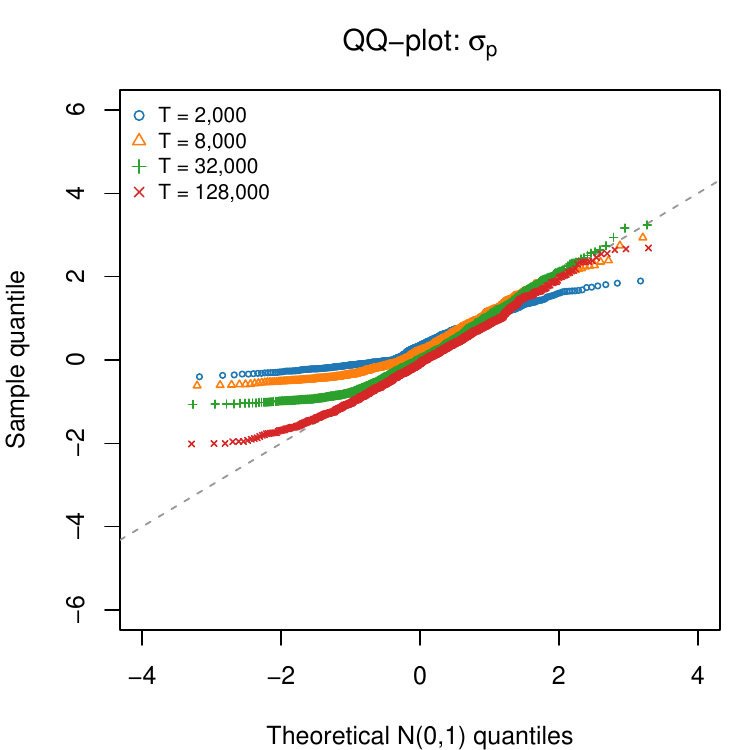} \caption{$\sigma_p$} 
    \end{subfigure} 
\end{figure} 
\clearpage

\section{SNFP with Multiple Epochs}\label{appendix:additional_mc}%
\label{appendix:time_memory_table}\label{appendix:snfp_ep1}

\rev{The Monte Carlo results of Section~\ref{section: Monte Carlo} run SNFP for a single epoch, so that each market is
visited exactly once, which is the setting covered by the theory of Section~\ref{section:asymptotic_theory}. In
practice, SNFP can also be run for several epochs. A later epoch reads a copy of the data file with the markets in a
new random order, and the Polyak--Ruppert average continues across epochs. Revisiting a market breaks the
independence between the current observation and the past iterates used in our proofs, so we treat additional
epochs as a finite-sample device. This appendix reports two complementary comparisons. The first matches accuracy:
SNFP runs until it agrees with the full-sample estimator to within a quarter of an NFP standard error, and we record
how many passes this takes and how long. The second fixes the budget at three shuffled epochs and reports bias, RMSE,
interval length and coverage alongside the single-epoch and NFP results of Table~\ref{tab:stat_accuracy}. All design
choices are those of Section~\ref{section: Monte Carlo}, namely the same DGP, a $T_0 = 1{,}000$ pilot,
$\theta_0 \sim \mathrm{Unif}(\theta^\ast \pm 0.1)$, the same box bounds and sandwich standard errors, and the same
hardware.}

\subsection{Accuracy-Matched Comparison}\label{appendix:accuracy_matched}

\rev{We stop SNFP when it matches the accuracy of NFP. Let $d_t$ be the distance defined in
Section~\ref{subsection:mc_efficiency}, the largest coordinate-wise gap between the Polyak--Ruppert average after
market $t$ and the NFP estimate on the same replication, in units of NFP standard errors.} We evaluate $d_t$ on a
log-spaced grid of checkpoints, about 50 per decade of $t$, and SNFP stops the first time $d_t \le \tau$ holds at
three consecutive checkpoints, with $\tau = 0.25$, so that every coordinate agrees with the full-sample estimator to
within a quarter of an NFP standard error. The reported time is the time at which this target is first attained.
The rule is an oracle time-to-target criterion \citep{bottou2008tradeoffs,dolan2002benchmarking}.
\rev{Table~\ref{tab:time_memory_tau} reports the results. Its NFP columns repeat those of Table~\ref{tab:time_memory}.}

% =============================================================================
% Numbers: Experimental_Design_2026_09_10, tau = 0.25, denom = plugin_rep.
%   epochs / SNFP seconds : results/tables/grid_stopping.csv (hit_epoch,
%                           hit_est_sec), mean and sd over 10 reps.
%   SNFP peak memory      : results/timing/timing_grid_B1_E*.csv (max_rss_kb).
%   NFP time and memory   : retained from tab:time_memory unchanged.
% =============================================================================
\begin{table}[!htbp]
\centering
\caption{Computational cost of NFP vs.\ accuracy-matched SNFP.}
\label{tab:time_memory_tau}
\small
\begin{threeparttable}
\begin{tabular}{r c r r r r r r}
\toprule
\multicolumn{1}{c}{\multirow{2}{*}{$T$}} & \multirow{2}{*}{\shortstack{Median\\ epochs}} & \multicolumn{3}{c}{Time (seconds)} & \multicolumn{3}{c}{Peak memory (MB)} \\
\cmidrule(lr){3-5} \cmidrule(lr){6-8}
 & & SNFP & NFP & NFP/SNFP & SNFP & NFP & NFP/SNFP \\
\midrule
$2{,}000$   & $6.63$ & $\underset{(4.1)}{4.4}$ & $\underset{(20.4)}{34.9}$ & $7.9\times$ & $\underset{(0.9)}{86.8}$ & $\underset{(0.1)}{87.1}$ & $1.0\times$ \\[4pt]
$8{,}000$   & $2.75$ & $\underset{(1.4)}{4.9}$ & $\underset{(83.3)}{161.9}$ & $32.7\times$ & $\underset{(0.0)}{87.1}$ & $\underset{(0.2)}{98.5}$ & $1.1\times$ \\[4pt]
$32{,}000$  & $2.04$ & $\underset{(10.9)}{15.7}$ & $\underset{(487.2)}{1{,}008.9}$ & $64.3\times$ & $\underset{(1.0)}{86.2}$ & $\underset{(0.6)}{140.5}$ & $1.6\times$ \\[4pt]
$128{,}000$ & $1.87$ & $\underset{(13.7)}{48.2}$ & $\underset{(1{,}533.9)}{2{,}291.8}$ & $47.6\times$ & $\underset{(1.0)}{86.4}$ & $\underset{(1.0)}{288.6}$ & $3.3\times$ \\[4pt]
$512{,}000$ & $1.08$ & $\underset{(74.0)}{133.3}$ & $\underset{(5{,}149.8)}{12{,}441.4}$ & $93.4\times$ & $\underset{(2.7)}{87.6}$ & $\underset{(1.2)}{859.0}$ & $9.8\times$ \\
\bottomrule
\end{tabular}
\begin{tablenotes}[flushleft]
\footnotesize
\item \textit{Notes.} Time and memory are means across 10 replications, with standard
deviations in parentheses.
\end{tablenotes}
\end{threeparttable}
\end{table}

The median epochs column of Table~\ref{tab:time_memory_tau} reports the number of passes through the sample at
which $d_t \le \tau$ first holds at three consecutive checkpoints, as the median over the 10 replications. It
falls from $6.63$ at $T = 2{,}000$ to $1.08$ at $T = 512{,}000$. Larger samples therefore require fewer passes
through the data to reproduce the full-sample estimator. We report the median because the number of passes at a
small sample varies widely across replications, from $2.37$ to $37.25$ at $T = 2{,}000$, as the stopping time
depends on the order in which the markets are read.

Under this stopping rule, SNFP remains substantially faster than NFP at every sample size, with
speedup factors between $7.9$ and $93.4$. At $T = 512{,}000$, SNFP reaches the target in $133.3$ seconds, while
NFP requires $12{,}441.4$ seconds. \rev{The speedup factors are smaller than those of the single-epoch comparison in
Table~\ref{tab:time_memory}, because SNFP now reads the sample several times at small $T$ and slightly more than
once at large $T$. The two comparisons also differ in how the speedup varies with $T$. Under a fixed budget of one
pass, the accuracy level of Table~\ref{tab:time_memory} improves from $1.73$ to $0.35$ standard errors across the
grid, so the same nominal budget is roughly five times less demanding at the top of the grid than at the bottom,
whereas under a common accuracy target the speedup shows no such pattern.} The standard deviation of computation
time is much smaller for SNFP at every $T$, \rev{as in the single-epoch comparison.} Peak memory is unaffected by
the number of passes, because SNFP holds one market in memory at a time, \rev{so the memory columns coincide with
those of Table~\ref{tab:time_memory} up to replication noise.}

\subsection{Three-Epoch SNFP}\label{appendix:three_epochs}

\rev{Table~\ref{tab:stat_accuracy_epochs} reports the bias, RMSE, average confidence interval length and coverage of
SNFP run for three shuffled epochs, alongside the single-epoch SNFP and NFP entries of Table~\ref{tab:stat_accuracy},
with $1{,}000$ replications at each sample size.} Three-epoch SNFP tracks NFP at every sample size, with a largest RMSE
gap of $0.016$ and agreement to within $0.003$ from $T = 8{,}000$ onward. \rev{Its coverage at $T = 2{,}000$ lies
between $0.924$ and $0.972$, against $0.800$ to $0.970$ for a single epoch, so the additional passes remove most of
the small-sample under-coverage documented in Section~\ref{subsection:mc_accuracy}. Its RMSE decreases approximately
at the $T^{-1/2}$ rate, and its coverage remains close to the nominal $95\%$ level at every sample size, including
$T = 128{,}000$ where only SNFP is reported. From $T = 8{,}000$ onward, the single-epoch RMSE is within $0.010$ of the
three-epoch RMSE, whereas the two coverage rates differ by up to $0.046$ at $T = 8{,}000$ and by at most $0.013$ at
$T = 128{,}000$.} Both variants approach the benchmark as $T$ increases.

% =============================================================================
% Numbers: SNFP v4 ep3 (T0=1000, S=3) for T in {2k, 8k, 32k}; SNFP v6 ep3 for
% T=128k (redraw-protected). NFP from second_trial with finite-only filter,
% n=1000 per tier.
% =============================================================================
\begin{table}[htbp]
\centering
\scriptsize
\setlength{\tabcolsep}{2pt}
\caption{Statistical accuracy of SNFP by number of epochs.}
\label{tab:stat_accuracy_epochs}
\begin{threeparttable}
\begin{tabular}{ll S[table-format=-1.3] S[table-format=-1.3] S[table-format=-1.3] S[table-format=-1.3] S[table-format=-1.3] S[table-format=-1.3] S[table-format=-1.3] S[table-format=-1.3] S[table-format=-1.3] S[table-format=-1.3] S[table-format=-1.3] S[table-format=-1.3]}
\toprule
& & \multicolumn{3}{c}{$T = 2{,}000$} & \multicolumn{3}{c}{$T = 8{,}000$} & \multicolumn{3}{c}{$T = 32{,}000$} & \multicolumn{3}{c}{$T = 128{,}000$} \\
\cmidrule(lr){3-5}\cmidrule(lr){6-8}\cmidrule(lr){9-11}\cmidrule(lr){12-14}
\multicolumn{2}{l}{Method} & {SNFP} & {SNFP} & {NFP} & {SNFP} & {SNFP} & {NFP} & {SNFP} & {SNFP} & {NFP} & {SNFP} & {SNFP} & {NFP} \\
\multicolumn{2}{l}{Epochs} & {1} & {3} & {---} & {1} & {3} & {---} & {1} & {3} & {---} & {1} & {3} & {---} \\
\midrule
& Bias \\
& \quad Intercept &  0.061 &  0.025 &  0.003 &  0.014 &  0.004 & -0.003 &  0.007 &  0.003 &  0.001 & -0.003 &  0.000 &  {---} \\
& \quad $\beta_x$ &  0.009 &  0.005 &  0.001 &  0.003 &  0.001 & -0.001 &  0.002 &  0.001 &  0.000 & -0.001 &  0.000 &  {---} \\
& \quad $\beta_p$ & -0.086 & -0.040 & -0.002 & -0.025 & -0.009 &  0.006 & -0.014 & -0.007 & -0.001 &  0.007 & -0.001 &  {---} \\
& \quad $\sigma_x$ &  0.003 &  0.001 &  0.000 &  0.001 &  0.001 &  0.000 &  0.001 &  0.000 &  0.000 &  0.000 &  0.000 &  {---} \\
& \quad $\sigma_p$ &  0.028 & -0.004 & -0.008 & -0.009 & -0.010 & -0.004 & -0.006 & -0.004 &  0.000 & -0.005 & -0.003 &  {---} \\
[3pt] & RMSE \\
& \quad Intercept &  0.133 &  0.099 &  0.089 &  0.053 &  0.047 &  0.045 &  0.025 &  0.023 &  0.023 &  0.012 &  0.011 &  {---} \\
& \quad $\beta_x$ &  0.044 &  0.029 &  0.025 &  0.016 &  0.013 &  0.013 &  0.007 &  0.007 &  0.006 &  0.003 &  0.003 &  {---} \\
& \quad $\beta_p$ &  0.234 &  0.192 &  0.176 &  0.101 &  0.091 &  0.088 &  0.050 &  0.047 &  0.045 &  0.023 &  0.022 &  {---} \\
& \quad $\sigma_x$ &  0.025 &  0.020 &  0.019 &  0.010 &  0.009 &  0.009 &  0.005 &  0.005 &  0.005 &  0.002 &  0.002 &  {---} \\
& \quad $\sigma_p$ &  0.099 &  0.106 &  0.114 &  0.056 &  0.056 &  0.054 &  0.030 &  0.029 &  0.028 &  0.015 &  0.014 &  {---} \\
[3pt] & CI length \\
& \quad Intercept &  0.356 &  0.349 &  0.346 &  0.173 &  0.173 &  0.172 &  0.087 &  0.087 &  0.086 &  0.043 &  0.043 &  {---} \\
& \quad $\beta_x$ &  0.101 &  0.100 &  0.099 &  0.050 &  0.049 &  0.049 &  0.025 &  0.025 &  0.025 &  0.012 &  0.012 &  {---} \\
& \quad $\beta_p$ &  0.707 &  0.689 &  0.681 &  0.339 &  0.337 &  0.336 &  0.170 &  0.170 &  0.170 &  0.084 &  0.084 &  {---} \\
& \quad $\sigma_x$ &  0.075 &  0.074 &  0.074 &  0.037 &  0.037 &  0.037 &  0.018 &  0.018 &  0.018 &  0.009 &  0.009 &  {---} \\
& \quad $\sigma_p$ &  0.428 &  0.435 &  0.435 &  0.213 &  0.212 &  0.210 &  0.106 &  0.105 &  0.105 &  0.052 &  0.052 &  {---} \\
[3pt] & Coverage (95\%) \\
& \quad Intercept &  0.839 &  0.934 &  0.946 &  0.910 &  0.937 &  0.939 &  0.923 &  0.938 &  0.943 &  0.944 &  0.945 &  {---} \\
& \quad $\beta_x$ &  0.800 &  0.924 &  0.955 &  0.885 &  0.931 &  0.954 &  0.931 &  0.946 &  0.948 &  0.935 &  0.947 &  {---} \\
& \quad $\beta_p$ &  0.898 &  0.935 &  0.945 &  0.925 &  0.945 &  0.951 &  0.918 &  0.933 &  0.940 &  0.942 &  0.946 &  {---} \\
& \quad $\sigma_x$ &  0.856 &  0.938 &  0.950 &  0.934 &  0.955 &  0.952 &  0.936 &  0.954 &  0.951 &  0.950 &  0.953 &  {---} \\
& \quad $\sigma_p$ &  0.970 &  0.972 &  0.973 &  0.953 &  0.955 &  0.954 &  0.937 &  0.943 &  0.946 &  0.937 &  0.950 &  {---} \\
\bottomrule
\end{tabular}
\begin{tablenotes}[flushleft]
\footnotesize
\item \textit{Notes.} \rev{The single-epoch SNFP and NFP entries repeat Table~\ref{tab:stat_accuracy}, and the
three-epoch entries use the same simulated data, draws, pilot and tuning parameters within each replication.}
NFP is not reported at $T = 128{,}000$ because it does not complete within the computation-time budget.
\end{tablenotes}
\end{threeparttable}
\end{table}

The two tables together indicate how many passes SNFP requires. When the sample size is small, a single pass is
not sufficient, and additional epochs are needed before SNFP reproduces the full-sample estimator. When $T$ is large,
the number of passes required to attain a given accuracy decreases rather than increases, and at the largest
sample size some replications of SNFP attain the target after reading the sample less than once. \rev{A single epoch
is therefore an adequate budget for large datasets, and additional shuffled epochs are a cheap remedy at small $T$,
where each pass costs little.}

\FloatBarrier

\section{Common Simulation Draws: Additional Simulation Results}\label{appendix:common_draws_mc}

Tables~\ref{tab:common_draws_T2000}--\ref{tab:common_draws_T8000_cov} report the common-draw experiment of
Section~\ref{section: Monte Carlo} at $T = 2{,}000$ and $T = 8{,}000$. The design and the reported statistics are
the same as in Tables~\ref{tab:common_draws} and~\ref{tab:common_draws_cov}.

\begin{table}[!htbp]
\centering
\small
\caption{RMSE under the three draw schemes at $T = 2{,}000$.}
\label{tab:common_draws_T2000}
\begin{threeparttable}
\begin{tabular}{ll *{3}{S[table-format=1.3]}}
\toprule
& & {$c = 0.98$} & {$c = 3.9$} & {$c = 15.6$} \\
& & {$R = 2{,}048$} & {$R = 512$} & {$R = 128$} \\
\midrule
\multicolumn{5}{l}{\emph{Per-market draws}} \\
& \quad Intercept   & 0.126 & 0.131 & 0.140 \\
& \quad $\beta_x$   & 0.040 & 0.040 & 0.042 \\
& \quad $\beta_p$   & 0.226 & 0.234 & 0.266 \\
& \quad $\sigma_x$  & 0.022 & 0.023 & 0.026 \\
& \quad $\sigma_p$  & 0.098 & 0.100 & 0.114 \\
[3pt] \multicolumn{5}{l}{\emph{Shared i.i.d.\ draws}} \\
& \quad Intercept   & 0.128 & 0.137 & 0.177 \\
& \quad $\beta_x$   & 0.045 & 0.059 & 0.107 \\
& \quad $\beta_p$   & 0.232 & 0.253 & 0.355 \\
& \quad $\sigma_x$  & 0.026 & 0.037 & 0.072 \\
& \quad $\sigma_p$  & 0.103 & 0.125 & 0.183 \\
[3pt] \multicolumn{5}{l}{\emph{Shared Sobol' draws}} \\
& \quad Intercept   & 0.126 & 0.125 & 0.145 \\
& \quad $\beta_x$   & 0.040 & 0.041 & 0.051 \\
& \quad $\beta_p$   & 0.225 & 0.227 & 0.275 \\
& \quad $\sigma_x$  & 0.022 & 0.024 & 0.034 \\
& \quad $\sigma_p$  & 0.098 & 0.103 & 0.136 \\
\bottomrule
\end{tabular}
\begin{tablenotes}[flushleft]
\footnotesize
\item \textit{Notes.} $c = T/R$. See Table~\ref{tab:common_draws}.
\end{tablenotes}
\end{threeparttable}
\end{table}

\begin{table}[!htbp]
\centering
\small
\caption{Coverage of $95\%$ confidence intervals under the three draw schemes at $T = 2{,}000$.}
\label{tab:common_draws_T2000_cov}
\begin{threeparttable}
\begin{tabular}{ll *{3}{S[table-format=1.3]}}
\toprule
& & {$c = 0.98$} & {$c = 3.9$} & {$c = 15.6$} \\
& & {$R = 2{,}048$} & {$R = 512$} & {$R = 128$} \\
\midrule
\multicolumn{5}{l}{\emph{Per-market draws}} \\
& \quad Intercept   & 0.855 & 0.843 & 0.851 \\
& \quad $\beta_x$   & 0.817 & 0.801 & 0.802 \\
& \quad $\beta_p$   & 0.901 & 0.899 & 0.887 \\
& \quad $\sigma_x$  & 0.907 & 0.881 & 0.860 \\
& \quad $\sigma_p$  & 0.957 & 0.966 & 0.962 \\
[3pt] \multicolumn{5}{l}{\emph{Shared i.i.d.\ draws}} \\
& \quad Intercept   & 0.841 & 0.832 & 0.797 \\
& \quad $\beta_x$   & 0.741 & 0.625 & 0.428 \\
& \quad $\beta_p$   & 0.895 & 0.886 & 0.816 \\
& \quad $\sigma_x$  & 0.838 & 0.692 & 0.465 \\
& \quad $\sigma_p$  & 0.955 & 0.915 & 0.839 \\
[3pt] \multicolumn{5}{l}{\emph{Shared i.i.d.\ draws, widened interval}} \\
& \quad Intercept   & 0.848 & 0.858 & 0.862 \\
& \quad $\beta_x$   & 0.869 & 0.901 & 0.909 \\
& \quad $\beta_p$   & 0.914 & 0.932 & 0.922 \\
& \quad $\sigma_x$  & 0.928 & 0.944 & 0.925 \\
& \quad $\sigma_p$  & 0.974 & 0.969 & 0.961 \\
[3pt] \multicolumn{5}{l}{\emph{Shared Sobol' draws}} \\
& \quad Intercept   & 0.855 & 0.850 & 0.829 \\
& \quad $\beta_x$   & 0.816 & 0.795 & 0.756 \\
& \quad $\beta_p$   & 0.901 & 0.908 & 0.875 \\
& \quad $\sigma_x$  & 0.897 & 0.879 & 0.733 \\
& \quad $\sigma_p$  & 0.959 & 0.952 & 0.891 \\
\bottomrule
\end{tabular}
\begin{tablenotes}[flushleft]
\footnotesize
\item \textit{Notes.} See Table~\ref{tab:common_draws_cov}.
\end{tablenotes}
\end{threeparttable}
\end{table}

\begin{table}[!htbp]
\centering
\small
\caption{RMSE under the three draw schemes at $T = 8{,}000$.}
\label{tab:common_draws_T8000}
\begin{threeparttable}
\begin{tabular}{ll *{3}{S[table-format=1.3]}}
\toprule
& & {$c = 0.98$} & {$c = 3.9$} & {$c = 15.6$} \\
& & {$R = 8{,}192$} & {$R = 2{,}048$} & {$R = 512$} \\
\midrule
\multicolumn{5}{l}{\emph{Per-market draws}} \\
& \quad Intercept   & 0.052 & 0.053 & 0.054 \\
& \quad $\beta_x$   & 0.016 & 0.016 & 0.016 \\
& \quad $\beta_p$   & 0.099 & 0.100 & 0.105 \\
& \quad $\sigma_x$  & 0.010 & 0.010 & 0.011 \\
& \quad $\sigma_p$  & 0.053 & 0.053 & 0.054 \\
[3pt] \multicolumn{5}{l}{\emph{Shared i.i.d.\ draws}} \\
& \quad Intercept   & 0.053 & 0.057 & 0.067 \\
& \quad $\beta_x$   & 0.019 & 0.027 & 0.045 \\
& \quad $\beta_p$   & 0.102 & 0.115 & 0.154 \\
& \quad $\sigma_x$  & 0.013 & 0.019 & 0.032 \\
& \quad $\sigma_p$  & 0.056 & 0.066 & 0.097 \\
[3pt] \multicolumn{5}{l}{\emph{Shared Sobol' draws}} \\
& \quad Intercept   & 0.052 & 0.053 & 0.055 \\
& \quad $\beta_x$   & 0.016 & 0.016 & 0.018 \\
& \quad $\beta_p$   & 0.099 & 0.100 & 0.108 \\
& \quad $\sigma_x$  & 0.010 & 0.011 & 0.013 \\
& \quad $\sigma_p$  & 0.053 & 0.054 & 0.061 \\
\bottomrule
\end{tabular}
\begin{tablenotes}[flushleft]
\footnotesize
\item \textit{Notes.} $c = T/R$. See Table~\ref{tab:common_draws}.
\end{tablenotes}
\end{threeparttable}
\end{table}

\begin{table}[!htbp]
\centering
\small
\caption{Coverage of $95\%$ confidence intervals under the three draw schemes at $T = 8{,}000$.}
\label{tab:common_draws_T8000_cov}
\begin{threeparttable}
\begin{tabular}{ll *{3}{S[table-format=1.3]}}
\toprule
& & {$c = 0.98$} & {$c = 3.9$} & {$c = 15.6$} \\
& & {$R = 8{,}192$} & {$R = 2{,}048$} & {$R = 512$} \\
\midrule
\multicolumn{5}{l}{\emph{Per-market draws}} \\
& \quad Intercept   & 0.904 & 0.908 & 0.897 \\
& \quad $\beta_x$   & 0.890 & 0.891 & 0.895 \\
& \quad $\beta_p$   & 0.921 & 0.921 & 0.905 \\
& \quad $\sigma_x$  & 0.931 & 0.932 & 0.926 \\
& \quad $\sigma_p$  & 0.958 & 0.955 & 0.957 \\
[3pt] \multicolumn{5}{l}{\emph{Shared i.i.d.\ draws}} \\
& \quad Intercept   & 0.902 & 0.876 & 0.841 \\
& \quad $\beta_x$   & 0.796 & 0.643 & 0.420 \\
& \quad $\beta_p$   & 0.910 & 0.869 & 0.803 \\
& \quad $\sigma_x$  & 0.851 & 0.683 & 0.475 \\
& \quad $\sigma_p$  & 0.945 & 0.884 & 0.765 \\
[3pt] \multicolumn{5}{l}{\emph{Shared i.i.d.\ draws, widened interval}} \\
& \quad Intercept   & 0.910 & 0.908 & 0.911 \\
& \quad $\beta_x$   & 0.926 & 0.934 & 0.949 \\
& \quad $\beta_p$   & 0.924 & 0.920 & 0.921 \\
& \quad $\sigma_x$  & 0.935 & 0.937 & 0.933 \\
& \quad $\sigma_p$  & 0.961 & 0.944 & 0.919 \\
[3pt] \multicolumn{5}{l}{\emph{Shared Sobol' draws}} \\
& \quad Intercept   & 0.906 & 0.903 & 0.903 \\
& \quad $\beta_x$   & 0.893 & 0.887 & 0.839 \\
& \quad $\beta_p$   & 0.920 & 0.913 & 0.900 \\
& \quad $\sigma_x$  & 0.933 & 0.921 & 0.848 \\
& \quad $\sigma_p$  & 0.956 & 0.949 & 0.917 \\
\bottomrule
\end{tabular}
\begin{tablenotes}[flushleft]
\footnotesize
\item \textit{Notes.} See Table~\ref{tab:common_draws_cov}.
\end{tablenotes}
\end{threeparttable}
\end{table}

\FloatBarrier

\section{NFP with Existing Packages}\label{appendix:pkg_comparison}

The NFP benchmark in Tables~\ref{tab:time_memory} and \ref{tab:time_memory_tau} is our own implementation. To check that the cost comparison
does not reflect a slow implementation of NFP, we also estimate the model with two existing packages, \texttt{PyBLP}
(version 1.2.0) \citep{conlon2020best} and \texttt{BLPestimatoR} (version 0.3.4). The design is the one of
Section~\ref{subsection:mc_efficiency}. Every estimator uses the same data file, the same $R = 1{,}000$ common draws
in every market, the same eight instruments, one-step GMM with weighting matrix $(Z'Z)^{-1}$, and L-BFGS-B over
$\sigma \in [10^{-4}, 2]^2$ from the same starting value. The contraction stops when the sup-norm change in $\delta$
is below $10^{-6}$.\footnote{\texttt{BLPestimatoR} iterates on $\exp(\delta)$, and we set its tolerance to
$10^{-12}$ so that the implied tolerance on $\delta$ is at most $10^{-6}$ for every product.} \texttt{PyBLP} uses a
finite-difference gradient, as our NFP does, and \texttt{BLPestimatoR} uses its analytic gradient. Both packages
compute standard errors as part of estimation, and we exclude that step from the reported time because our NFP
computes none; including it adds 7--12\%. Each replication runs on one core of the same AMD EPYC 9654 nodes. In
every replication, the estimates of both packages agree with our NFP estimate to within $0.005$ standard errors in
every coordinate.

Table~\ref{tab:pkg_comparison} reports the results. Neither package changes the conclusion of
Section~\ref{subsection:mc_efficiency}. \texttt{BLPestimatoR} is the fastest NFP implementation at $T \le 8{,}000$,
but its time grows faster than linearly in $T$. It is $358$ times slower than SNFP at $T = 128{,}000$, and at
$T = 512{,}000$ it did not finish within a 24-hour limit. \texttt{PyBLP} is between $12.8$ and $157.1$ times slower
than SNFP and takes $5.8$ hours at $T = 512{,}000$. The memory gap is larger still. Both packages take the draws
as market-level agent data, and their peak memory grows roughly in proportion to $T$, at about $0.17$~MB
(\texttt{PyBLP}) and $0.20$~MB (\texttt{BLPestimatoR}) per market. At $T = 512{,}000$, \texttt{PyBLP} needs
$86.8$~GB, $990$ times the $87.6$~MB used by SNFP.

\begin{table}[tbp]
\centering
\caption{Computational cost of SNFP and of NFP by implementation.}
\label{tab:pkg_comparison}
\small
\begin{threeparttable}
\begin{tabular}{r r r r r r r r}
\toprule
\multicolumn{1}{c}{$T$} & \multicolumn{1}{c}{SNFP} & \multicolumn{2}{c}{NFP (ours)} & \multicolumn{2}{c}{\texttt{PyBLP}}
 & \multicolumn{2}{c}{\texttt{BLPestimatoR}} \\
\cmidrule(lr){3-4} \cmidrule(lr){5-6} \cmidrule(lr){7-8}
 & & \multicolumn{1}{c}{Mean} & \multicolumn{1}{c}{/SNFP} & \multicolumn{1}{c}{Mean} & \multicolumn{1}{c}{/SNFP}
 & \multicolumn{1}{c}{Mean} & \multicolumn{1}{c}{/SNFP} \\
\midrule
\multicolumn{8}{l}{\textit{Panel A: Time (seconds)}} \\[2pt]
$2{,}000$ & $\underset{(4.1)}{4.4}$ & $\underset{(20.4)}{34.9}$ & $7.9\times$ & $\underset{(8.3)}{56.7}$ & $12.8\times$
 & $\underset{(3.4)}{16.9}$ & $3.8\times$ \\[4pt]
$8{,}000$ & $\underset{(1.4)}{4.9}$ & $\underset{(83.3)}{161.9}$ & $32.7\times$ & $\underset{(29.7)}{219.9}$
 & $44.5\times$ & $\underset{(13.5)}{111.8}$ & $22.6\times$ \\[4pt]
$32{,}000$ & $\underset{(10.9)}{15.7}$ & $\underset{(487.2)}{1{,}008.9}$ & $64.3\times$ & $\underset{(69.5)}{761.3}$
 & $48.5\times$ & $\underset{(150.8)}{1{,}166.4}$ & $74.3\times$ \\[4pt]
$128{,}000$ & $\underset{(13.7)}{48.2}$ & $\underset{(1{,}533.9)}{2{,}291.8}$ & $47.6\times$
 & $\underset{(252.6)}{3{,}370.3}$ & $69.9\times$ & $\underset{(1{,}862.6)}{17{,}235.8}$ & $357.6\times$ \\[4pt]
$512{,}000$ & $\underset{(74.0)}{133.3}$ & $\underset{(5{,}149.8)}{12{,}441.4}$ & $93.4\times$
 & $\underset{(4{,}948.9)}{20{,}931.6}$ & $157.1\times$ & \multicolumn{1}{c}{---} & \multicolumn{1}{c}{---} \\[4pt]
\midrule
\multicolumn{8}{l}{\textit{Panel B: Peak memory (MB)}} \\[2pt]
$2{,}000$ & $\underset{(0.9)}{86.8}$ & $\underset{(0.1)}{87.1}$ & $1.0\times$ & $\underset{(4.6)}{475.9}$ & $5.5\times$
 & $\underset{(0.8)}{718.8}$ & $8.3\times$ \\[4pt]
$8{,}000$ & $\underset{(0.0)}{87.1}$ & $\underset{(0.2)}{98.5}$ & $1.1\times$ & $\underset{(4.3)}{1{,}512.1}$
 & $17.4\times$ & $\underset{(0.5)}{2{,}046.2}$ & $23.5\times$ \\[4pt]
$32{,}000$ & $\underset{(1.0)}{86.2}$ & $\underset{(0.6)}{140.5}$ & $1.6\times$ & $\underset{(31.3)}{5{,}570.9}$
 & $64.6\times$ & $\underset{(0.2)}{7{,}518.9}$ & $87.2\times$ \\[4pt]
$128{,}000$ & $\underset{(1.0)}{86.4}$ & $\underset{(1.0)}{288.6}$ & $3.3\times$ & $\underset{(31.2)}{21{,}804.8}$
 & $252.4\times$ & $\underset{(1{,}646.5)}{25{,}763.4}$ & $298.2\times$ \\[4pt]
$512{,}000$ & $\underset{(2.7)}{87.6}$ & $\underset{(1.2)}{859.0}$ & $9.8\times$ & $\underset{(26.2)}{86{,}769.3}$
 & $990.0\times$ & \multicolumn{1}{c}{---} & \multicolumn{1}{c}{---} \\
\bottomrule
\end{tabular}
\begin{tablenotes}[flushleft]
\footnotesize
\item \textit{Notes.} Means across 10 replications, with standard deviations in parentheses. ``/SNFP'' is the ratio
of the mean to the SNFP mean. SNFP is stopped by the accuracy-matched rule of Table~\ref{tab:time_memory_tau}, and
its columns and those of our NFP repeat that table. Time for \texttt{PyBLP} and \texttt{BLPestimatoR} excludes the
standard-error step. ``---'': \texttt{BLPestimatoR} did not finish within a 24-hour limit at $T = 512{,}000$.
\end{tablenotes}
\end{threeparttable}
\end{table}

\section{Mini-Batch SNFP}\label{appendix:minibatch}

The SNFP update in~\eqref{eq:snfp} uses one market per step. A natural variant averages the moment over a
mini-batch of $B$ markets. Let $\mathcal B_s$ denote the $s$-th batch of $B$ consecutive online markets. The update
is
\[
  \theta_s = \Pi_\Theta\Bigl\{\theta_{s-1} - \gamma_s\,\hat A_{T_0}\,\frac{1}{B}\sum_{t\in\mathcal B_s}
  g(O_t;\theta_{s-1},P_t^R)\Bigr\}, \qquad \gamma_s = \gamma_0\, s^{-a},
\]
and $\bar\theta$ is the Polyak--Ruppert average of the iterates $\theta_s$. With $B = 1$ this is the baseline SNFP.
A larger $B$ reduces the noise in each step but also reduces the number of steps per epoch by the factor $B$. It
does not reduce the cost per market, because the inner contraction mapping must be solved market by market.

We rerun the design of Table~\ref{tab:time_memory_tau} with $B \in \{1, 5, 10\}$ and $\gamma_0 \in \{1, 0.5\}$,
using the same data, starting values and stopping rule. The number of epochs is the number of markets processed
until the stopping rule is met, divided by $T$. \rev{The time is the elapsed time at the stop, measured in the
loop. The column for $B = 1$ and $\gamma_0 = 1$ reports the runs of
Table~\ref{tab:time_memory_tau}, summarized here by the mean number of epochs, $10.41$ at $T = 2{,}000$, rather
than the median, $6.63$.}

Table~\ref{tab:batch_g1} reports the results with $\gamma_0 = 1$, the value used in the main text. Mini-batching
reduces the average number of epochs and the time at every sample size. Table~\ref{tab:batch_g05} shows that this
gain appears to come mainly from the smaller effective step size. With $B$ markets per step, an epoch has $B$ times
fewer steps, so the cumulative step size over an epoch is smaller, much as with a smaller $\gamma_0$. With
$\gamma_0 = 0.5$, SNFP with $B = 1$ performs about as well as the mini-batch versions with $\gamma_0 = 1$, and
mini-batching yields no systematic further improvement. The replication-level variation is large, especially at
$T = 2{,}000$, so small differences across columns should not be overinterpreted. Overall, we find no evidence that
mini-batching improves on single-market updates with a suitably chosen step size in this design.

\begin{table}[!htbp]
\centering
\small
\caption{Mini-batch SNFP with $\gamma_0 = 1$.}
\label{tab:batch_g1}
\begin{threeparttable}
\begin{tabular}{r ccc ccc}
\toprule
\multicolumn{1}{c}{\multirow{2}{*}{$T$}} & \multicolumn{3}{c}{Average epochs} & \multicolumn{3}{c}{Time (seconds)} \\
\cmidrule(lr){2-4} \cmidrule(lr){5-7}
 & $B = 1$ & $B = 5$ & $B = 10$ & $B = 1$ & $B = 5$ & $B = 10$ \\
\midrule
$2{,}000$ & $\underset{(10.65)}{10.41}$ & $\underset{(4.35)}{4.41}$ & $\underset{(1.34)}{2.28}$ & $\underset{(4.1)}{4.4}$ & $\underset{(1.6)}{2.0}$ & $\underset{(0.6)}{1.4}$ \\[4pt]
$8{,}000$ & $\underset{(0.96)}{2.97}$ & $\underset{(0.76)}{1.82}$ & $\underset{(1.04)}{2.05}$ & $\underset{(1.4)}{4.9}$ & $\underset{(1.1)}{3.2}$ & $\underset{(1.6)}{3.6}$ \\[4pt]
$32{,}000$ & $\underset{(1.78)}{2.39}$ & $\underset{(0.52)}{1.36}$ & $\underset{(0.52)}{1.25}$ & $\underset{(10.9)}{15.7}$ & $\underset{(3.8)}{9.6}$ & $\underset{(3.5)}{8.6}$ \\[4pt]
$128{,}000$ & $\underset{(0.61)}{1.94}$ & $\underset{(0.48)}{1.39}$ & $\underset{(0.49)}{1.28}$ & $\underset{(13.7)}{48.2}$ & $\underset{(12.8)}{37.0}$ & $\underset{(12.4)}{33.2}$ \\[4pt]
$512{,}000$ & $\underset{(0.72)}{1.32}$ & $\underset{(0.28)}{1.05}$ & $\underset{(0.29)}{1.09}$ & $\underset{(74.0)}{133.3}$ & $\underset{(26.3)}{106.5}$ & $\underset{(30.7)}{111.2}$ \\
\bottomrule
\end{tabular}
\begin{tablenotes}[flushleft]
\footnotesize
\item \textit{Notes.} The table reports means across 10 replications,
with standard deviations in parentheses below. Epochs and times are measured at the stopping rule of
Table~\ref{tab:time_memory_tau}.
\end{tablenotes}
\end{threeparttable}
\end{table}

\begin{table}[!htbp]
\centering
\small
\caption{Mini-batch SNFP with $\gamma_0 = 0.5$.}
\label{tab:batch_g05}
\begin{threeparttable}
\begin{tabular}{r ccc ccc}
\toprule
\multicolumn{1}{c}{\multirow{2}{*}{$T$}} & \multicolumn{3}{c}{Average epochs} & \multicolumn{3}{c}{Time (seconds)} \\
\cmidrule(lr){2-4} \cmidrule(lr){5-7}
 & $B = 1$ & $B = 5$ & $B = 10$ & $B = 1$ & $B = 5$ & $B = 10$ \\
\midrule
$2{,}000$ & $\underset{(3.68)}{3.71}$ & $\underset{(1.06)}{2.32}$ & $\underset{(2.86)}{3.10}$ & $\underset{(1.5)}{1.9}$ & $\underset{(0.4)}{1.4}$ & $\underset{(1.0)}{1.6}$ \\[4pt]
$8{,}000$ & $\underset{(0.67)}{1.99}$ & $\underset{(0.99)}{2.29}$ & $\underset{(0.49)}{2.01}$ & $\underset{(1.0)}{3.8}$ & $\underset{(1.6)}{4.2}$ & $\underset{(0.8)}{3.6}$ \\[4pt]
$32{,}000$ & $\underset{(0.43)}{1.35}$ & $\underset{(0.57)}{1.35}$ & $\underset{(0.55)}{1.29}$ & $\underset{(3.0)}{9.6}$ & $\underset{(3.9)}{9.5}$ & $\underset{(4.3)}{9.5}$ \\[4pt]
$128{,}000$ & $\underset{(0.53)}{1.40}$ & $\underset{(0.42)}{1.32}$ & $\underset{(0.36)}{1.20}$ & $\underset{(14.8)}{39.6}$ & $\underset{(9.9)}{33.0}$ & $\underset{(9.4)}{33.7}$ \\[4pt]
$512{,}000$ & $\underset{(0.24)}{1.01}$ & $\underset{(0.26)}{1.07}$ & $\underset{(0.68)}{1.17}$ & $\underset{(24.1)}{104.6}$ & $\underset{(32.1)}{112.5}$ & $\underset{(70.3)}{119.2}$ \\
\bottomrule
\end{tabular}
\begin{tablenotes}[flushleft]
\footnotesize
\item \textit{Notes.} The table reports means across 10 replications,
with standard deviations in parentheses below. Epochs and times are measured at the stopping rule of
Table~\ref{tab:time_memory_tau}.
\end{tablenotes}
\end{threeparttable}
\end{table}

\FloatBarrier

\end{document}